\documentclass[a4paper,11pt]{article}
\usepackage[margin=1in]{geometry}
\usepackage[utf8]{inputenc}
\usepackage{soul} %

\usepackage{amsthm}

\usepackage{lineno}

\usepackage{graphicx, xcolor}
\usepackage{silence}
\usepackage{hyperref}
\hypersetup{
	colorlinks,
	linkcolor={green!60!black},%
	citecolor={blue!70!black},%
	urlcolor={blue!80!black}
}

\usepackage{tabularx}
\usepackage{amsmath,amsfonts,amssymb}%
\usepackage{cleveref}
\usepackage{mathtools}      %

\usepackage{thmtools}
\usepackage{thm-restate}
\usepackage[algo2e,ruled,noend,resetcount,linesnumbered]{algorithm2e}
\usepackage{algorithm, algorithmicx, algpseudocode} 
\usepackage{xspace}

\newtheorem{conjecture}{Conjecture}[section]
\newtheorem{theorem}{Theorem}[section]
\newtheorem{lemma}[theorem]{Lemma}
\newtheorem{proposition}[theorem]{Proposition}
\newtheorem{corollary}[theorem]{Corollary}
\newtheorem{remark}[theorem]{Remark}

\newtheorem{definition}{Definition}

\newcounter{hypothesis}

{\makeatletter
 \gdef\xxxmark{%
   \expandafter\ifx\csname @mpargs\endcsname\relax %
     \expandafter\ifx\csname @captype\endcsname\relax %
     \marginpar{\textcolor{red}{\textbf{xxx}}}%
          \else
       \textbf{\textcolor{red}{xxx}} %
    \fi
   \else
     \textbf{\textcolor{red}{xxx}} %
   \fi}
 \gdef\xxx{\@ifnextchar[\xxx@lab\xxx@nolab}
 \long\gdef\xxx@lab[#1]#2{\textbf{[\xxxmark \textcolor{{blue}#2} ---{\sc {\color{blue}#1}}]}}
 \long\gdef\xxx@nolab#1{\textbf{[\xxxmark \textcolor{blue}{#1}]}}
}

\newcommand{\poly}{\text{poly}}

\newcommand{\eps}{\varepsilon}

\def \R {\mathbb R}
\def \N {\mathbb N}

\newcommand{\klcs}[1][]{\ifthenelse{\equal{#1}{}}{$k$-LCS}{${#1}$-LCS}}
\newcommand{\kwlcs}[1][]{\ifthenelse{\equal{#1}{}}{$k$-WLCS}{${#1}$-WLCS}}
\newcommand{\kNLstC}[1][]{\ifthenelse{\equal{#1}{}}{$k$-NLstC}{${#1}$-NLstC}}
\newcommand{\kELstC}[1][]{\ifthenelse{\equal{#1}{}}{$k$-ELstC}{${#1}$-ELstC}}
\newcommand{\czkc}[1][]{\ifthenelse{\equal{#1}{}}{FZ$k$C}{FZ${#1}$C}}
\newcommand{\czkch}[1][]{\ifthenelse{\equal{#1}{}}{FZ$k$CH}{FZ${#1}$CH}}
\newcommand{\cfkc}[1][]{\ifthenelse{\equal{#1}{}}{F$\mathfrak{f}k$C}{F$\mathfrak{f}{#1}$C}}
\newcommand{\cfkch}[1][]{\ifthenelse{\equal{#1}{}}{F$\mathfrak{f}k$CH}{F$\mathfrak{f}{#1}$CH}}

\newcommand{\pfksat}[1][]{\ifthenelse{\equal{#1}{}}{$\oplus$F$k$SAT}{$\oplus$F${#1}$SAT}}

\newcommand{\ckov}[1][]{\ifthenelse{\equal{#1}{}}{F$k$-OV}{F${#1}$-OV}}
\newcommand{\ckovh}[1][]{\ifthenelse{\equal{#1}{}}{F$k$-OVH}{F${#1}$-OVH}}

\newcommand{\cksum}[1][]{
\ifthenelse{\equal{#1}{}}
{F$k$-SUM}
{F${#1}$-SUM}
}

\newcommand{\cksumh}[1][]{\ifthenelse{\equal{#1}{}}{F$k$-SUMH}{F${#1}$-SUMH}}

\newcommand{\fzkc}[1][]{\ifthenelse{\equal{#1}{}}{FZ$k$C}{FZ${#1}$C}}

\newcommand{\ckxor}[1][]{\ifthenelse{\equal{#1}{}}{F$k$-XOR}{F${#1}$-XOR}}
\newcommand{\ckxorh}[1][]{\ifthenelse{\equal{#1}{}}{F$k$-XORH}{F${#1}$-XORH}}

\newcommand{\ckfunc}[1][]{\ifthenelse{\equal{#1}{}}{F$k$-$\mathfrak{f}$}{F${#1}$-$\mathfrak{f}$}}
\newcommand{\ckfunch}[1][]{\ifthenelse{\equal{#1}{}}{F$k$-$\mathfrak{f}$H}{F${#1}$-$\mathfrak{f}$H}}

\DeclarePairedDelimiter{\ceil}{\lceil}{\rceil}
\DeclarePairedDelimiter{\floor}{\lfloor}{\rfloor}

\newcommand{\bigCeil}[1]{\left\lceil #1 \right\rceil}
\newcommand{\bigFloor}[1]{\left\lfloor #1 \right\rfloor}

\newcommand{\GoodDPolys}[1][]{\ifthenelse{\equal{#1}{}}{Good $d$-Degree Polynomials}{Good ${#1}$-Degree Polynomials}}
\newcommand{\goodDPoly}[1][]{\ifthenelse{\equal{#1}{}}{good $d$-degree polynomial}{good ${#1}$-degree polynomial}}

\newcommand{\OKPolys}[1][]{\ifthenelse{\equal{#1}{}}{OK $d$-Degree Polynomials}{OK ${#1}$-Degree Polynomials}}
\newcommand{\okPoly}[1][]{\ifthenelse{\equal{#1}{}}{ok $d$-degree polynomial}{ok ${#1}$-degree polynomial}}

\newcommand{\kxor}[1][]{\ifthenelse{\equal{#1}{}}{$k$-XOR}{${#1}$-XOR}}

\newcommand{\ksum}[1][]{\ifthenelse{\equal{#1}{}}{$k$-SUM}{${#1}$-SUM}}

\newcommand{\vksum}[1][]{\ifthenelse{\equal{#1}{}}{V$k$-SUM}{V${#1}$-SUM}}

\newcommand{\ukov}[1][]{\ifthenelse{\equal{#1}{}}{U$k$-OV}{U${#1}$-OV}}
\newcommand{\nukov}[1][]{\ifthenelse{\equal{#1}{}}{NU$k$-OV}{NU${#1}$-OV}}
\newcommand{\ukxor}[1][]{\ifthenelse{\equal{#1}{}}{U$k$-XOR}{U${#1}$-XOR}}

\newcommand{\set}[1]{\{#1\}}

\newcommand{\domain}{\mathcal{X}}

\newcommand{\problem}{\mathcal{P}}

\newcommand{\queries}{\mathcal{Q}}
\newcommand{\innerAlg}{\mathcal{A}}
\newcommand{\outerAlg}{\mathcal{B}}
\newcommand{\distribution}{\mathcal{D}}
\newcommand{\calset}{\mathcal{S}}
\newcommand{\given}{\textrm{\xspace s.t. \xspace}}
\newcommand{\andT}{\textrm{\xspace and \xspace}}

\newcommand{\otherwise}{\textrm{\xspace o/w \xspace}}
\newcommand{\permT}{\mathrm{Perm}}

\newcommand{\tO}[1]{\tilde{O} (#1)}
\newcommand{\tto}[1]{\tilde{\tilde{o}} (#1)}

\newcommand{\bigOmega}[1]{\Omega \left(#1\right)}

\DeclareMathOperator{\supp}{supp}
\DeclareMathOperator{\wt}{wt}
\DeclareMathOperator{\uniqueSet}{set}
\DeclareMathOperator{\sensediff}{s_{err}}
\DeclareMathOperator{\cdiff}{c_{err}}
\DeclareMathOperator{\order}{ord}
\DeclareMathOperator{\distT}{\delta}
\newcommand{\pos}{\mathop{pos}}

\newcommand{\striangleF}{\# s \textrm{-} \triangle}
\newcommand{\striangleI}{s \textrm{-} \triangle}

\usepackage{wrapfig}     
\newcommand{\erclogowrapped}[1]{%
\setlength\intextsep{0pt}%
\begin{wrapfigure}[3]{r}{#1*\real{1.1}}%
\includegraphics[width=#1]{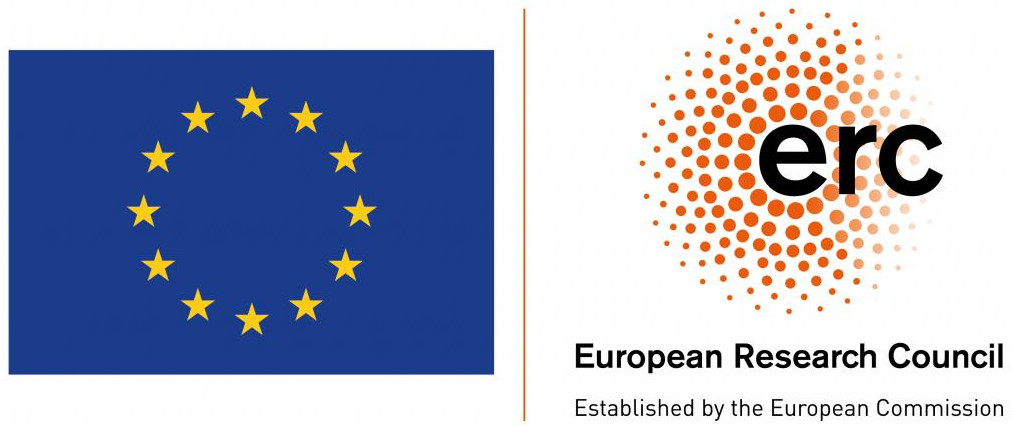}%
\end{wrapfigure}%
}

\title{On the Complexity of Algorithms with Predictions for Dynamic Graph Problems}
\date{July 2023}
\date{}

\author{Monika Henzinger\thanks{
This project has received funding from the European Research Council (ERC) \erclogowrapped{5\baselineskip}
under the European Union's Horizon 2020 research and innovation programme
(Grant agreement No.\ 101019564 ``The Design of Modern Fully Dynamic Data
Structures (MoDynStruct)'' and from the Austrian Science Fund (FWF) project
Z 422-N, and project 
``Fast Algorithms for a Reactive Network Layer (ReactNet)'', P~33775-N, with
additional funding from the \textit{netidee SCIENCE Stiftung}, 2020--2024.}  \and Barna Saha\thanks{University of California Berkeley. This work is supported partly by NSF 1652303, 1909046, and HDR TRIPODS 1934846 grants, and an Alfred P. Sloan Fellowship. This work was also supported by the Simons NTT research fellowship.}\and Martin P. Seybold \and Christopher Ye}

\usepackage[draft]{fixme}   %
\fxsetup{mode=multiuser,theme=color, layout=inline}
\FXRegisterAuthor{mh}{amh}{{\color{purple}\textbf{Monika}}}
\FXRegisterAuthor{ms}{mps}{{\color{blue}\textbf{Martin}}}
\FXRegisterAuthor{bs}{abs}{{\color{violet}\textbf{Barna}}}
\FXRegisterAuthor{cy}{czy}{{\color{green}\textbf{Chris}}}

\begin{document}

\maketitle


\begin{abstract}
\emph{Algorithms with predictions} is a new research direction that incorporates machine learning predictions into algorithm design. So far a plethora of works published in recent years have incorporated the power of predictions to improve on worst-case optimal bounds for online problems. 
In this paper, we initiate the study of complexity of dynamic data structures with predictions, including dynamic graph algorithms.
Unlike in online algorithms, the main goal in dynamic data structures is to maintain the solution \emph{efficiently} with every update.

Motivated by prior work in online algorithms, we investigate three natural models of predictions: 
(1) $\varepsilon$-accurate predictions where each predicted request matches the true request with probability at least $\varepsilon$, 
(2) list-accurate predictions where a true request comes from a list of possible requests, and 
(3) bounded delay predictions where the true requests are some (unknown) permutations of the predicted requests. 
For $\varepsilon$-accurate predictions, we show that lower bounds from the non-prediction setting of a problem carry over, up to a $1-\varepsilon$ factor.
Then we give general reductions among the prediction models for a problem, showing that lower bounds for bounded delay imply lower bounds for list-accurate predictions, which imply lower bounds for $\varepsilon$-accurate predictions. 

Further, we identify two broad problem classes based on lower bounds due to the Online Matrix Vector (OMv) conjecture.
Specifically, we show that dynamic problems that are \emph{locally correctable} have strong conditional lower bounds for list-accurate predictions that are equivalent to the non-prediction setting, unless list-accurate predictions are perfect.
Moreover, we show that dynamic problems that are \emph{locally reducible} have a smooth transition in the running time, for the online and the offline setting with bounded delay predictions.
We categorize problems with known OMv lower bounds accordingly and give several upper bounds in the delay model that show that our lower bounds are almost tight, including problems in dynamic graphs.

We note that concurrent work by v.d.Brand et al. [arXiv:2307.09961] and Liu and Srinivas [arXiv:2307.08890] independently study dynamic graph algorithms with predictions, but their work is mostly focused on showing upper bounds.

\end{abstract}

\newpage
{ %
\tableofcontents
}
\newpage
\setcounter{page}{1}%

\section{Introduction}
\label{sec:intro-motivation}
\sloppy

Modern Machine Learning predictions models are surprisingly accurate in practice and exploiting their, seemingly ever improving, accuracy
is a novel direction in theory.
{Algorithms with predictions} have access to an oracle that provide a hint for solving the problem at hand, that is based on learning from some distribution.
Bounds for handling an input object that are sensitive to prediction quality can improve substantially on worst-case optimal bounds.
An algorithm with prediction is called 
\emph{robust} if 
the algorithm does not perform worse than the best known algorithm that does not use predictions, even if the predictions contain errors.
For algorithm design and analysis, prediction oracles are assumed to have a bounded \emph{error measure} or a bounded \emph{accuracy probability}. %
Analysis that is sensitive to the bounded error assumption requires a meaningful notion of distance between predicted and actual inputs.
Depending on the problem, some measures are better suited to establish upper bounds than others (see, e.g., the survey of Mitzenmacher and Vassilvitskii~\cite{MitzenmacherV21}). Moreover, in general, it is not possible to know the exact error apriori. In contrast, predictions with bounded accuracy probability do not assume bounds on a particular error measure and instead assume that every prediction is correct with certain probability \cite{GuptaPSS22}.

In online problems\footnote{We use \emph{online} to denote problems where the input consists of a sequence of operations, which can modify the input or ask a query about the input, and \emph{offline} to denote that all input is given at once. 
}, both bounded-error (e.g.~\cite{AzarPT22}) and bounded-accuracy (e.g.~\cite{GuptaPSS22}) predictions have been studied extensively.
There are a wide range of problems where predictions allow to improve quality over worst-case optimal competitive ratios, such as counting sketches \cite{DBLP:conf/iclr/HsuIKV19, DBLP:journals/corr/abs-1908-05198, DBLP:conf/iclr/EdenINRSW21}, bloom filters \cite{DBLP:conf/sigmod/KraskaBCDP18}, caching/paging \cite{DBLP:conf/soda/Rohatgi20, DBLP:journals/jacm/LykourisV21, DBLP:conf/soda/BansalCKPV22, pmlr-v162-im22a, pmlr-v202-antoniadis23a}, ski rental \cite{DBLP:conf/nips/PurohitSK18, DBLP:conf/nips/BamasMS20, DBLP:conf/nips/AntoniadisCEPS21, pmlr-v202-shin23c}, correlation clustering \cite{DBLP:conf/iclr/SilwalANMRK23}, among many others.\footnote{See for example, \url{https://algorithms-with-predictions.github.io} or the survey~\cite{MitzenmacherV21}.} 
The standard assumption is that the algorithm is given access to the predictions for the whole sequence of operations (or requests) \emph{before} the algorithm has to produce its first output, i.e, during {preprocessing}. 
Dependent on the correctness of the information provided by the prediction, ideally algorithms with prediction should provide a smooth transition between the online and offline problems.

In this paper, we initiate the study of complexity of dynamic data structures and algorithms with predictions. 
Unlike in online algorithms, the main goal in the dynamic setting is to maintain the solution \emph{efficiently} with every update. To the best of our knowledge, investigating the potential of algorithms with prediction in the dynamic setting has only just started with our present work, and the independent, concurrent work of van~den~Brand, Forster, Nazari, and Polak~\cite{arxiv23} and Liu and Srinivas \cite{DBLP:journals/corr/abs-2307-08890}.
Like in online algorithms, we aim for:
1) The algorithm should be \emph{consistent}, achieving the performance of an optimal offline algorithm when the prediction quality is high.
2) The algorithm should be \emph{robust}, matching the performance of an (online) dynamic algorithm regardless of prediction quality.
3) The algorithm's performance should degrade gracefully between the two extremes as prediction quality deteriorates. Given the predictions, we can allow polynomial preprocessing time. When the actual updates arrive, the dynamic data structure must process them fast provided the available information from preprocessing.

Let us consider as an example the Online Matrix Vector (OMv) problem which has been instrumental in developing lower bounds for dynamic data structures \cite{HenzingerKNS15}. 
In this problem, given a Boolean matrix $M$ of size $n \times n$, and an online sequence of $n$ vectors $\vec v_1, \ldots , \vec v_n $, one needs to report $M \vec v_i$ before seeing any $\vec v_t$ with $ t > i$. 
While the OMv conjecture states that the total time needed to process these $n$ vectors cannot be sub-cubic, if these vectors are given apriori as part of predictions, then one could have preprocessed them in $O(n^\omega)$ time using fast matrix multiplication where $\omega<2.373$, and output the results in $O(n)$ time per vector. 
This is often called the offline-online gap. 
Of course, it is unrealistic to assume that predictions are completely accurate. But this already showcases 
ample room for potential improvements in dynamic data structures due to algorithms with predictions.

For dynamic graph problems, a line of work initiated by \cite{HenzingerKNS15} establishes conditional lower bounds on the time trade-offs between updates and queries for a \emph{large number} of dynamic problems, based on the OMv conjecture (see e.g., \cite{DBLP:conf/icalp/Dahlgaard16, DBLP:conf/esa/GoranciHP17a,BergamaschiHGWW21,DBLP:conf/networking/HenzingerP021, DBLP:conf/esa/HenzingerPS22}).
However, these reductions typically consist of update sequences that present pathological and repetitive behaviour, e.g. repeatedly requesting certain updates, asking a query, and then reverting the updates again. 
Algorithms with predictions might have a tremendous potential for improving the running time bounds on such sequences.
Thus, in this work, the central question that we investigate is:
\begin{center}
\emph{Can predictions lead to provably faster dynamic graph algorithms?}
\end{center}

\subsection{Contribution and Paper Outline} \label{sec:contribution}

\paragraph*{Warm Up: OMv with Predictions}

Given the central importance of the OMv Problem, we start by investigating if algorithms with predictions can bypass the bounds of the OMv Conjecture (Section~\ref{sec:omv-EH-pred}). Recall that in this problem an $n \times n$ Boolean matrix $M$ is given initially and can be preprocessed arbitrarily in polynomial time. 
This is called a \emph{round}. Then a sequence of $n$ Boolean vectors $v_1,v_2,\ldots $ arrives and the Boolean product $M v_i$ needs to be output before seeing the next vector $v_{i+1}$. Let us call the predicted vectors as $(\hat v_1,\ldots, \hat v_n)$. A natural measure to quantify prediction errors could be the maximum $\ell_1$ distance or Hamming distance between $\hat v_i$ and $ v_i$ for $i=1,2,\ldots,n$.  In fact, we consider an even more general notion of error called {\em Extended Hamming Distance} which is always upper bounded by Hamming distance, and show a smooth transition in complexity across the offline-online gap for OMv that uses \emph{predictions with bounded Extended Hamming distance} (Theorem~\ref{thm:OMv-hamming}). Moreover, our lower bound (Theorem~\ref{thm:OMv-hamming-LB}) shows that the algorithm is essentially optimal under the OMv conjecture.

Our study then turns to analyzing conditions that allow or prevent obtaining similar positive results for more general dynamic data structures and graph problems with appropriate prediction models.

\paragraph*{Prediction Models for Dynamic Problems and General Lower Bounds}
In Section~\ref{sec:prediction-models}, we propose and analyze three quality measures for prediction accuracy that are suitable for dynamic problems.
These are 
(1) \emph{$\eps$-accurate predictions}, where each predicted request matches the true request with probability at least $\varepsilon$, 
(2) \emph{list-accurate predictions},  where the prediction for each time step is a list of possible requests, and 
(3) \emph{bounded-delay predictions}, where the true requests are some (unknown) permutations of the predicted requests.

The \emph{$\eps$-accurate predictions} \cite{GuptaPSS22} and \emph{bounded-delay predictions} \cite{DBLP:conf/nips/PurohitSK18, DBLP:conf/soda/Rohatgi20, DBLP:journals/jacm/LykourisV21, DBLP:conf/swat/AntoniadisGS22, DBLP:journals/corr/abs-2202-10199} have already been studied in the online algorithms literature . The \emph{list-accurate predictions} are similar to multi-prediction model \cite{DBLP:conf/icml/Anand0KP22, DBLP:conf/nips/DinitzILMV22, DBLP:conf/icml/0001C00S23} where the best prediction from a list needs to be selected at every step. It can also be seen as a 
generalization of \emph{$\eps$-accurate predictions} in a sense by providing a list of possibilities out of which one is the correct update.
For algorithms with \emph{$\varepsilon$-accurate predictions}, where $\varepsilon \in (0,1)$, we show that any lower bound on the time complexity from the non-prediction setting carries over, reduced by a factor of $1-\varepsilon$ (see Proposition~\ref{prop:eps-accuracy-lower-bound}).
Then we give general reductions among the notions for a problem, showing that 
lower bounds for bounded delay imply lower bounds for list-accurate predictions (Corollary~\ref{cor:LB-delay-to-LB-list}), 
which imply lower bounds for $\varepsilon$-accurate predictions (Corollary~\ref{cor:LB-list-to-LB-eps}). 
We believe that our `Alternating Parallel Simulation' technique to show such reductions (Lemma~\ref{lem:parallel-simulation-search}) is of interest for analyzing further prediction models. In particular, this provides a natural hierarchy in the power of prediction models.

\paragraph*{Locally Correctable Problems: Hardness for List-Accurate Predictions}
In Section~\ref{sec:locally-correctable}, we introduce a class of \emph{locally correctable} dynamic problems (defined informally below), and show that OuMv lower bounds for these problems continue to hold for any algorithm with list accurate predictions, unless the predictions are perfect (i.e. the list of each time step has size $1$). The OuMv problem is a slight generalization of the OMv problem where we are given two sets of $n$ vectors $\vec u_1,  \vec u_2, \ldots, \vec u_n$ and $\vec v_1, \vec v_2,  \ldots, \vec v_n$ along with the matrix $M$ of $n \times n$ dimensions. The products $u_i^\top M v_i$ needs to be computed before seeing $\vec u_j, \vec v_j$ with $j=i+1,\ldots,n$. The OuMv conjecture excludes algorithms with total work that is subcubic in $n$, and follows from the OMv conjecture (see Theorem~\ref{thm:oumv}). 

\begin{theorem}[Informal, cf.\ Theorem~\ref{thm:locally-correctable-lb}]
    \label{thm:locally-correctable-lb-informal}
    Let $\varepsilon>0$ be constant.
    Suppose $\problem$ is a locally correctable problem due to an OuMv reduction that uses $u(n)$ many updates and $q(n)$ many queries.
    Then there is no algorithm solving $\problem$ with $2$-list accurate predictions that has update time $U(n)$ and query time $Q(n)$ satisfying
    
    \begin{equation*}
        n \Big( 
        u(n) U(n) + q(n) Q(n) 
        \Big) 
        = \Omega(n^{3-\varepsilon})~,
    \end{equation*}
    if the OuMv conjecture is true.
\end{theorem}

Our lower bound follows from a reduction from OuMv such that the set of request sequences admit efficiently computable $2$-list accurate predictions. The basic idea is that for this class of problems a generic ``universal request sequence'' can be efficiently created (without knowledge of the exact reduction sequence arising in the hardness reduction). Now any dynamic algorithm (without prediction)  can efficiently construct this universal request sequence for itself in the preprocessing phase and then execute an algorithm with prediction using this universal request sequence as prediction. Thus, no efficient dynamic algorithm with predictions can exist unless the OuMv conjecture is false.

Roughly speaking, a \emph{locally correctable} problem $\problem$ satisfies the following three properties: 
1) any OuMv instance can be simulated by choosing some \emph{subsequence} of a 
universal request sequence, containing updates needed for answering the query as well as updates that will turn out not to be useful for answering the queries, called ``junk'' updates
2) any $\problem$ instance can be \emph{augmented} efficiently into an instance containing both useful as well as useless updates, and 
3) the answer to any query in the augmented instance can be \emph{corrected} efficiently to answer the corresponding query in the original instance.
We can then construct the following reduction from an OuMv instance.
Typically, an OuMv based reduction encodes the round's query into the problem instance using some subsequence of possible updates to modify the data of $\problem$ only where necessary.
Instead, our reductions perform the specified update where necessary and otherwise insert a ``junk" update. %
As a result, obtaining a $2$-list accurate prediction is simple: The list contains the update itself and an arbitrary ``junk" update.

\paragraph*{Locally Reducible Problems: Hardness for Bounded Delay Predictions}
In Section~\ref{sec:locally-reducible-lb-bounded-delay}, we introduce a class of \emph{locally reducible} dynamic problems, where proving lower bounds against algorithms with bounded delay predictions is possible.
As with list accurate predictions, our lower bounds rely on constructing an OuMv-based reduction such that a bounded delay prediction for the resulting sequence can be constructed efficiently.
Roughly speaking, a \emph{locally reducible} problem satisfies two properties: 
1) any OuMv instance can be simulated by choosing some subsequence of a universal request sequence and 
2) each update, if repeated often enough, say $ord$ times, leaves the dynamic data structure unchanged.\footnote{Alternatively, we could model this by giving each update operation a corresponding ``undo" operation.}
Inspired by algebra, we call such operations \emph{cyclic}.
As before, a typical OuMv reduction proceeds by choosing a subset of possible updates in the problem instance in order to encode a round's query.
Our prediction then simply predicts that in each query vector, every possible update will be necessary.
Denote each set of updates required to encode one query round a \emph{block}.
$\problem$ is $(u, q)$-locally reducible if each block contains $u(n)$ updates and $q(n)$ queries.
Since each block consists of some subset of the universal update set, an update will not occur \emph{before} its predicted block, but it is possible an update occurs \emph{after} its predicted block.
To ensure that a request does not occur too long after its expected block (say a request has not occurred in $ord$ blocks but it was contained in the universal request sequence already $ord$ times), we perform the update $ord$ times, using property 2) to show that the underlying data structure does not change. It follows that an update cannot occur more than $ord$ blocks after its predicted block.
Since this sequence has small delay relative to the predicted sequence, the universal request sequence is a prediction with small delay. 
Our lower bound then follows, as any dynamic algorithm can construct this prediction during the preprocessing phase.

\begin{theorem}[Informal, cf.\ Theorem~\ref{thm:locally-reducible-lower-fully-dynamic} and Theorem \ref{thm:locally-reducible-lower-partially-dynamic}]
    \label{thm:locally-reducible-lower-informal}
    Let $\eps > 0$ be constant.
    Suppose $\problem$ is $(u, q)$-locally reducible from OuMv.
    Then there is no algorithm solving $\problem$ with $O(u(n) + q(n))$ delayed predictions with update time $U(n)$ and query time $Q(n)$ satisfying
    
    \begin{equation*}
        n \Big( u(n) U(n) + q(n) Q(n) \Big) = O(n^{3 - \eps})~,
    \end{equation*}
    if the OuMv conjecture is true.
\end{theorem}

Using similar techniques, we show that the lower bound degrades gracefully as the delay error of the prediction decreases (Theorem~\ref{thm:locally-reducible-lower-d}).

\paragraph*{Examples of Locally Correctable and Locally Reducible Problems}

In Section \ref{sec:further-problems}, we use our frameworks to provide lower bounds against algorithms with predictions for the following problems:
\textsc{Subgraph Connectivity \cite{DBLP:journals/algorithmica/FrigioniI00, DBLP:conf/icalp/Duan10, DBLP:journals/siamcomp/ChanPR11, DBLP:conf/soda/KapronKM13, DBLP:conf/focs/AbboudW14, HenzingerKNS15}, Reachability \cite{DBLP:conf/focs/AbboudW14, HenzingerKNS15}, Shortest Path \cite{DBLP:journals/jacm/EvenS81, DBLP:journals/siamcomp/DorHZ00, DBLP:journals/algorithmica/RodittyZ11, DBLP:journals/siamcomp/RodittyZ12}, Distance Spanners/Emulators \cite{BergamaschiHGWW21}, Maximum Matching \cite{guptapeng2013matchings, DBLP:journals/siamcomp/BaswanaGS18, DBLP:conf/focs/Solomon16, DBLP:conf/icalp/Dahlgaard16, DBLP:conf/soda/KopelowitzPP16}, Maximum Flow \cite{mkadry2011graphs, HenzingerKNS15, DBLP:conf/icalp/Dahlgaard16}, Triangle Detection \cite{HenzingerKNS15}, Densest Subgraph \cite{HenzingerKNS15}, $d$-Failure Connectivity \cite{DBLP:conf/stoc/DuanP10, DBLP:conf/soda/KopelowitzPP16}, Vertex Color Distance Oracle \cite{DBLP:conf/icalp/HermelinLWY11, DBLP:conf/esa/Chechik12, DBLP:conf/stoc/LackiOPSZ15, DBLP:conf/isaac/EvaldFW21}, Weighted Diameter \cite{DBLP:conf/soda/FrischknechtHW12,HenzingerKNS15}, Strong Connectivity \cite{DBLP:conf/focs/AbboudW14, HenzingerKNS15}, Electrical Flows \cite{DBLP:conf/esa/GoranciHP17a}, Erickson's Maximum Value Problem \cite{Patrascu10, HenzingerKNS15}, Langerman's Zero Prefix Sum Problem \cite{Patrascu10, HenzingerKNS15}}.
All turn out to be Locally Reducible Problems.
Additionally, all are Locally Correctable Problems as well, with exception of Erickson's Maximum Value Problem.

\paragraph*{Dynamic Algorithms with Predictions}
In Section \ref{sec:algs-with-predictions}, we give several algorithms with bounded delay predictions for the problems $\striangleF$, Subgraph Connectivity, Transitive Closure, All Pairs Shortest Path, and Erickson's Maximum Value Problem. Some of them can even handle outliers, which are updates that were not at all in the predicted set.
See Table~\ref{tbl:list-of-algorithms-with-prediction}.
These algorithms with predictions are optimal (up to lower order terms) with respect to the prediction quality.
That is, matching our conditional lower bounds for $d$-delayed in either the update or query time. Moreover, none of these algorithms need to know the prediction quality (the error parameter $d$ apriori).

\begin{table}[tb]
    \begin{center}
        \begin{tabular}{ |c|c|c|c|c|c|c| } 
             \hline
             Problem & \multicolumn{3}{|c|}{Upper Bounds} %
             & 
             \multicolumn{3}{|c|}{Lower Bounds} %
             \\
             \hline
             & Update & Query & Reference & Update & Query & Reference \\
             \hline
             $\striangleF$ & $d + k$ & $1$ & Thm.~\ref{thm:striangle-delay-alg} & $d^{1 - \eps}$ & $d^{2 - \eps}$ & Thm.~\ref{thm:locally-reducible-lower-d} \\
             $\striangleF$ & $1$ & $(d + k)^2$ & Thm.~\ref{thm:striangle-delay-alg} & $d^{1 - \eps}$ & $d^{2 - \eps}$ & Thm.~\ref{thm:locally-reducible-lower-d} \\
             Subgraph Connectivity & $1$ & $d^2$ & Thm.~\ref{thm:subgraph-connectivty-delay-alg} & $d^{1 - \eps}$ & $d^{2 - \eps}$ & Thm.~\ref{thm:locally-reducible-lower-d} \\ 
             Transitive Closure & $1$ & $d^2$ & Thm.~\ref{thm:transitive-closure-delay-alg} & $d^{1 - \eps}$ & $d^{2 - \eps}$ & Thm.~\ref{thm:locally-reducible-lower-d} \\ 
             All Pairs Shortest Path & $1$ & $d^2$ & Thm.~\ref{thm:apsp-delay-alg} & $d^{1 - \eps}$ & $d^{2 - \eps}$ & Thm.~\ref{thm:locally-reducible-lower-d} \\ 
             Erickson's Problem & $d + k$ & $1$ & Thm.~\ref{thm:ericksons-delay-alg} & $d^{1 - \eps}$ & $d^{2 - \eps}$ & Thm.~\ref{thm:locally-reducible-lower-d} \\
             Erickson's Problem & $1$ & $(d + k)^2$ & Thm.~\ref{thm:ericksons-delay-alg} & $d^{1 - \eps}$ & $d^{2 - \eps}$ & Thm.~\ref{thm:locally-reducible-lower-d} \\
             \hline
        \end{tabular}
    \end{center}
    \caption{List of algorithms with predictions with $d$ bounded delay.
    Each algorithm has polynomial preprocessing time.
    A running time involving $k$ states that the algorithm can handle predictions that are $d$ delayed with $k$ outliers.
    The lower bounds state that no algorithm with polynomial preprocessing time exists for any constant $\varepsilon>0$ that attains the stated update and query times \emph{simultaneously}, unless the OMv-conjecture fails.
    }
    \label{tbl:list-of-algorithms-with-prediction}
\end{table}

To design our algorithms, we show that the difference between the state of the predicted data and the actual data scales with the delay error of the prediction.
Furthermore, this difference can be maintained efficiently.
In the preprocessing phase, we compute the predicted data structures, extracting useful intermediate values we require from the predicted data structure.
Now, given the online request sequence up to some time step $t$, we show that by making small changes to the predicted data structure (on the order of the predictions delay) we can recover the result of the query on the actual data structure from precomputed values on the predicted data structure.

\paragraph*{Concurrent Work}
Independent work of van den Brand, Forster, Nazari, and Polak \cite{arxiv23} and Liu and Srinivas \cite{DBLP:journals/corr/abs-2307-08890} jointly initiate the study of dynamic graph algorithms with predictions, focusing on upper bounds. 
\cite{arxiv23} gives (among many other things) partially dynamic algorithms with bounded delay predictions for transitive closure and all pairs shortest path (cf.\ Theorem~\ref{thm:transitive-closure-delay-alg} and \ref{thm:apsp-delay-alg}), and show that these algorithms are optimal with a lower bound giving the same result as Theorem~\ref{thm:locally-reducible-lower-partially-dynamic} (cf.\ Appendix~\ref{sec:concurrent-work}).
\cite{DBLP:journals/corr/abs-2307-08890} considers the prediction model where a deletion time is predicted for every inserted edge, which is different from our proposed models.
To the best of our knowledge, the above summarizes any overlapping contribution with \cite{arxiv23} and \cite{DBLP:journals/corr/abs-2307-08890}.

\section{Preliminaries and OMv with Predictions}
\label{sec:prelims}

We define dynamic data structures in general. 

\begin{definition}[Dynamic Data Structure]
\label{def:gen-dynamic-model}
Let $\problem$ be a dynamic problem.
For any instance $x$ of $\cal P$, let $\domain(x)$ be the set of possible {\bf updates} on $x$ and let $\queries(x)$ be the set of possible {\bf queries}.
When the instance is clear, we omit the subscripts and write $\domain, \queries$.
In the pre-processing step, the algorithm receives as input an initial data structure $x_0$.
At each time step $t$, the algorithm receives some {\bf request} $\rho_t \in \domain \cup \queries$.
When given a query $\rho_t \in \queries$, the algorithm must answer the query correctly on the current structure $x_t$, where $x_t$ is obtained by applying request sequence $(\rho_1, \rho_2, \dotsc, \rho_{t - 1})$ to the initial set $x_0$.
The query must be answered before the following request $\rho_{t + 1}$ is revealed.

If the updates $\domain$ do not allow element deletions or do not allow element insertions, the data structure is called {\bf partially dynamic}, otherwise it is {\bf fully dynamic}. 
\end{definition}

For a given request sequence $\rho$, let $\rho_{[a, b]}=(\rho_a,\ldots,\rho_b)$ denote the sub-sequence between the $a$-th and $b$-th time step.
Let $\rho_{\leq t} = \rho_{[1, t]}$ denote the prefix of the first $t$ requests in $\rho$ and let $\rho_{<t}$ denote $\rho_{\leq t - 1}$.

\subsection{The Online Matrix Vector Problem}

Our hardness results are built on the OMv Conjecture of \cite{HenzingerKNS15}.
Recall that, in Boolean Matrix-Vector multiplication the arithmetic plus operation is replaced by logical-OR~$\lor$ and the arithmetic multiplication operation is replaced by logical-AND~$\land$ (see, e.g.,~\cite[Section~3]{Williams07}).

\begin{definition}[OMv and OuMv~{\cite[Def.~2.6]{HenzingerKNS15}}]
An algorithm for the OMv (resp. OuMv) problem is given parameter $n$ as its input.
Next, it is given a Boolean matrix $M$ of size $n\times n$ that can be preprocessed in time $p(n)$.
This is followed by $n$-rounds of processing online input vectors.

An {\bf \em OMv} algorithm is given an online sequence of $n$ vectors $\vec v_1, \ldots , \vec v_n $, one vector after the other, and the task is to report each result of Boolean Matrix-Vector multiplication $M \hat v_t$ before $\hat v_{t+1} $ arrives.

An {\bf \em OuMv} algorithm is given an online sequence of $n$ vector pairs $(\vec u_1, \vec v_1), \ldots , (\vec u_n , \vec v_n )$, one pair after the other, and the task is to report each result of Boolean Vector-Matrix-Vector multiplication
$\vec u_t ^\top M \vec v_t$ before $(\vec u_{t+1}, \vec v_{t+1} )$ arrives.
\end{definition}
We call the processing of each individual input vector  a \emph{round}.
Clearly, every \emph{OMv round} can be solved in $O(n^2)$ time, yielding a trivial $O(n^3)$ algorithm. The OMv conjecture claims that this is basically optimal (cf.~{\cite[Conj.~1.1]{HenzingerKNS15}}):

\begin{conjecture}
    \label{thm:omv-conjecture}
    For any constant $\varepsilon > 0$, there is no algorithm with   $O(n^{3-\varepsilon} )$ total time that solves OMv with an error probability of at most $1/3$.  
\end{conjecture}
This conjecture leads to the following result for the OuMv problem (cf.~{\cite[Thm.~2.7]{HenzingerKNS15}}.).
We simply call this the \emph{OuMv conjecture} even though it is just a consequence of the OMv conjecture and \emph{not} a different conjecture.

\begin{theorem}
\label{thm:oumv}
    For any constant $\varepsilon > 0$, 
    the OMv conjecture implies that there is no algorithm with preprocessing time polynomial in $n$ and total time for all requests of  $O(n^{3-\varepsilon} )$ that solves OuMv with an error probability of at most $1/3$. 
\end{theorem}

\subsection{Upper and Lower Bounds for OMv with Predictions}\label{sec:omv-EH-pred}
We begin by discussing how predictions affect the complexity of the Online Matrix-Vector Problem, leading to bounds that have a \emph{smooth transition} across the offline-online gap in terms of the prediction quality.
We also give conditional lower bounds, showing that this result is optimal.

The extended Hamming distance $EH(s,t)$ of two bit-strings $s,t \in \{0,1\}^n$ is defined as follows. 
Let $\ell \geq 1$ be the largest index with $s_1=s_2=\ldots = s_\ell$ and $t_1=t_2=\ldots = t_\ell$, then 

\begin{equation*}
EH( s_1\ldots s_n,~t_1\ldots t_n) := (s_1 + t_1 \mod 2) + EH(s_{\ell+1}\ldots s_n,~t_{\ell+1}\ldots t_n)~.    
\end{equation*}

Since each block-difference is only counted once, the EH-distance is at most the Hamming distance, where the latter is equal to the $L_1$-distance on $\{0,1\}^n$.
The works of Lingas et al.~\cite{Lingas02,Lingas21} introduced the extended Hamming distance to study Boolean matrix-vector multiplication and showed that, after $\tO{n^2}$ pre-processing, each OMv round can be solved in $\tO{n + mst(M)}$ time.
Here $mst(M)=O(n^2)$ is the weight of a (geometric) Minimum Spanning Tree of the row-vectors in $M$.
That is, each row-vector in $M$ is interpreted as one point in the $n$-dimensional Hamming space $\{0,1\}^n$ and the distance between any two row-vectors is the extended Hamming distance.
Our prediction-based algorithm in this section is due to a small adaptation of the online $O(\log n)$-approximate Minimum Spanning Tree heuristic in~\cite[Section~3.2]{Lingas21}.
Our algorithm shows however that bounded-error predictions allow to bypass query bounds that are sensitive to the, potentially quadratic, weight $mst(M)$.
That is predication-based algorithms yield a smooth transition between the $O(n^3)$ online and $O(n^\omega)$ offline bound.

\begin{theorem}[OMv with Predictions]\label{thm:OMv-hamming}
    Let $M \in \{0,1\}^{n \times n}$ and $\Delta \in [0,n]$.
    Given predictions $(\hat v_1,\ldots, \hat v_n)$ that have $EH(\hat v_i, \vec v_i)\leq \Delta$ for each online input $\vec v_i\in \{0,1\}^n$,
    each arithmetic product $M \vec v_i$ can be computed in time $Q(n,\Delta)=O(n(1+\Delta))$, after preprocessing of $M$ and $(\hat v_1,\ldots \hat v_n)$ in $O(n^\omega)$ time.

    In particular, each Boolean-result vector can be computed in ${O(n(1+\Delta))}$ time.
\end{theorem}
Recall that bounds that are sensitive to EH-distance are stronger than bounds sensitive to $L_1$.
We remark that the following proof can be extended to rectangular matrices $M$ with non-binary entries (e.g. small integers or real values).

\begin{proof}
    We first describe the preprocessing of the algorithm.
    Compute the matrix $(\hat y_1,\ldots,\hat y_n) = M (\hat v_1,\ldots, \hat v_n)$ using fast matrix multiplication.
    Compute a Prefix-Sum array $D_i$ for each row $i$ of $M$.
    That is, $D_i$ allows to obtain, for an block-range $[a,b]$ of column indices, the value $\sum_{k \in [a,b]} M_{i,k}=D_i(b)-D_i(a-1)$ in $O(1)$ time.
    Since each data structure $D_i$ can be computed in $O(n)$ time, preprocessing takes $O(n^2 + n^\omega)$ time.

    To compute the arithmetic result $\vec y = M \vec v$ for the online input $\vec v\in \{0,1\}^n$ of some OMv round, we determine in $O(n)$ time the blocks $b_1,\ldots,b_\Delta$ of index-ranges that contribute to the extended Hamming distance $EH(\hat v,\vec v)$.
    Next, we initialize the result vector $\vec y$ with the precomputed vector $\hat y$.
    Then we iterate, over each row $j \in [n]$ and each block $b_i=[b_i^-,b_i^+]$, and add to the $j$-th entry in $\vec y$ the value $( \vec v_{b_i^+} - \hat v_{b_i^+} ) \sum_{k \in b_i} M_{i,k}$, where $(\vec v_{b_i^+} - \hat v_{b_i^+}) \in\{-1,+1\}$ since $\vec v$ and $\hat v$ are bit-vectors. This requires time $O(1 + \Delta)$.
\end{proof}

Note that, after preprocessing, our algorithm with prediction requires  time $O(n (1 +\Delta))$ per round.
As we show next, this upper bound is also nearly the best possible that can be achieved for algorithms with predictions that have extended Hamming distance at most $\Delta$, unless the OMv-conjecture is false.
Specifically, we show a conditional lower bound stating that no algorithm can achieve query time $\tto{n (1+\Delta)}$ with polynomial improvement (see Definition~\ref{def:tto}).

\begin{restatable}{theorem}{omvhamminglb}
    \label{thm:OMv-hamming-LB}
    Let $\eps \in (0,1)$ be a constant and $\Delta \in [1, n^\eps ]$.
    There is no algorithm with predictions of EH-distance at most $\Delta$ for the OMv problem with amortized time $Q(n,\Delta)= \tto{n \Delta}$ per round, if the OMv-conjecture is true.
\end{restatable}
(See Appendix~\ref{sec:omv-pred-lb-proofs} for the proof.)

\section{Dynamic Prediction Models with General Lower Bounds}
\label{sec:prediction-models}
The first question to ask is how to define algorithms with predictions in the dynamic setting.
In this section, we generalize definitions of algorithms with predictions from the static setting and the corresponding error-measures to the dynamic setting and show that some definitions of predictions are so weak that almost the same lower bounds on the time complexity can be shown as in the setting without prediction.

\subsection{$\eps$-Accurate Predictions}

We begin with the simplest, very general formulation that lead to improved algorithms (see e.g.~\cite{GuptaPSS22}).
Informally, a predicted request sequence $\hat{\rho}$ is $\eps$-accurate if each predicted request matches the respective online request with probability at least $\eps$.

\begin{definition}[$\eps$-Accurate Predictions]
    \label{def:eps-accurate-predictions}
    Let $ \eps \in [0,1]$.
    Consider a dynamic problem with update set $\domain$ and query set $\queries$.
    Let $\distribution$ be a distribution over sequences of $T$ requests, i.e. over $(\domain \cup \queries)^T$, and
    $\hat{\rho} = (\hat{\rho}_1, \hat{\rho}_2, \dotsc, \hat{\rho}_T)$ be a sequence of $T$ predicted requests.
    
    Then $\hat{\rho}$ is an {\bf $\eps$-accurate} prediction for $\distribution$, if each $\hat{\rho}_t$ in $\hat{\rho}$ has
        $\displaystyle \Pr_{\rho \sim \distribution}[\rho_t = \hat{\rho}_t] \geq \eps$ .
\end{definition}

This is a natural model of prediction as accuracy is one of the most common metrics to evaluate the performance of a machine learning model. However, we show that it is too weak of a notion to design efficient dynamic algorithms with prediction.
The following proposition shows that for any constant $\eps<1$, a dynamic problem that is hard in the online setting remains hard even if an $\eps$-accurate prediction is available in advance.
For example, even if a dynamic algorithm has a prediction that is correct for $99.9\%$ of future requests, known lower bounds for the online problem still hold.

\begin{proposition}[Request Amplification]
    \label{prop:eps-accuracy-lower-bound}
    Let $\eps \in (0,1)$ be a constant and $f$ a non-constant, non-decreasing function in $n$.
    Suppose there is a dynamic problem $\cal P$ with query set $\queries \neq \emptyset$ such that any algorithm processing $T$ requests on instances of size $n$ requires  worst-case time $\Omega(Tf(n))$.
    Then there exists a distribution $\distribution$ of request sequences from $\cal P$ such that any algorithm $\innerAlg$ with $\varepsilon$-accurate predictions for $\distribution$ requires amortized time $\Omega({(1 - \eps)} f(n))$ per request.
\end{proposition}

\begin{proof}
    Let $\rho$ be a worst-case request sequence of length $T$ for the algorithm and let $q^* \in \queries$ be an arbitrary fixed query. 
    We give a randomized reduction that maps  $\rho$  to an
    \emph{augmented request sequence} $\rho'$ as follows:
    Let $a = \ceil{1/(1-\eps)}$. Consider the set of all request sequences of length $(a+1)T$ that contain $\rho$ as a sub-sequence and $aT$ many copies of query $q^*$. The \emph{augmented request sequence} $\rho'$ is drawn uniformly at random from this set. Note that  for every $t \in [1, T]$, the $t$-th position consists of the query $q^*$ with probability 
    $\frac{aT}{(a+1)T} \ge \frac{1/(1-\eps)}{1/(1-\eps) + 1} = \frac{1}{2-\eps} \ge \eps$. Here the first inequality is due to $\frac{x}{x+1}$ having a positive derivative and the second inequality is true as $0\geq 2 \eps -\eps^2-1=-(\eps-1)^2$.
    
    By assumption, any algorithm correctly processing $\rho$ requires $\Omega(T f(n))$ time.
    As the sequence $\rho'$ contains the sequence $\rho$ as a sub-sequence, and the added queries do not change the underlying data of the problem instance, any algorithm correctly processing the request sequence $\rho'$ requires time $\Omega(T f(n))$.

    Now, consider an algorithm $\innerAlg$ with the  prediction $\hat{\rho} = (q^*, \dotsc, q^*)$ of length $(a+1)T$. 
    By construction, this is an $\eps$-accurate prediction.
    Since any dynamic algorithm that has to answer $\rho$ can itself generate $\hat \rho$ in time $O(aT)$ and use it as an (arguably useless) prediction for preprocessing.
    Thus, the total time for such a dynamic algorithm is $O(aT + T_p)$, where $T_p$ is the total time for the algorithm with an $\eps$-accurate prediction to process $\rho'$, but this must be $\Omega(T f(n))$.
    It follows that the amortized time of $\innerAlg$, even given an $\eps$-accurate prediction, must be at least $
        \bigOmega{\frac{T f(n) - aT}{aT}} = \Omega((1 - \eps) f(n))$ per request.    
\end{proof}

Thus, for any constant $\eps<1$, the amortized time per request is at least $\bigOmega{f(n)}$.
This shows that $\eps$-accurate predictions are not particularly powerful for dynamic problems with known, conditional or unconditional, lower bounds.
This motivates the search for alternative stronger models of prediction under which it may be possible to harness the power of efficient offline algorithms.

One shortcoming of $\eps$-accurate predictions are their generality. Regardless of the prediction, the support of the distribution $\distribution$ can be \emph{every} possible request sequence of length $T$, if we assign small enough probability to request sequences that do not match well with the prediction.
We thus investigate also more restrictive models that restrict the possible input sequences $\calset \subset (\domain \cup \queries)^T$ for a given prediction.

\subsection{List Accurate Predictions}
Next we investigate a deterministic model for predictions that does not require a prediction to exactly specify the $t$-th request, but only to reveal `some information' about it.
If each request is represented by a bit-string of $O(\log |\domain\cup\queries|)$ bits, the following prediction model can be thought of as revealing a subset of these bits, i.e. the prediction $\hat \rho_t$ for step $t$ is a set, of size at most $L$, of possible requests such that the $t$-th request $\rho_t  \in \hat \rho_t$.

\begin{definition}[$L$-List Accurate Predictions]
    \label{def:list-predictions}
    In a dynamic problem with update set $\domain$ and query set $\queries$, 
    let $\calset \subseteq (\domain \cup \queries)^T$ be a set of sequences with $T$ requests.
    
    A sequence of $T$ sets $\hat{\rho} = (\hat{\rho}_1, \hat{\rho}_2, \dotsc, \hat{\rho}_T)$, where each $\hat{\rho}_t \subseteq \domain \cup \queries$, is called an {\bf $L$-list accurate} prediction for $\calset$,
    if each set $\hat{\rho}_t$ contains at most $L$ elements and we have for each sequence $\rho \in \calset$ and all $t \in [T]$ that $\rho_t \in \hat{\rho}_t$.
\end{definition}

Clearly, having an $L$-list accurate prediction with $L = 1$ is a perfect prediction.
Also note that there is always a $|\domain \cup \queries|$-list accurate prediction for all inputs, i.e. $\calset = (\domain \cup \queries)^T$.
In the OMv-problem for example, we have that the queries are from $\queries = \{0,1\}^n$ and having an $L$-list accurate prediction for preprocessing allows to solve each OMv round in $O(n)$ time, after spending $O(L n^\omega)$ time for preprocessing.

\begin{lemma} 
    \label{lem:list-accuracy-eps-accuracy}
    
    Given an $L$-list accurate prediction $\hat{\rho}=(\hat{\rho}_1,\ldots,\hat{\rho}_T)$ for $\calset$, one can compute in $O(\sum_t |\hat{\rho}_t|)=O(LT)$ time
    an $\frac{1}{L}$-accurate prediction $\hat{\rho}'$ for the uniform distribution on $\calset$. 
\end{lemma}
\begin{proof}
Given an $L$-list accurate prediction $\hat{\rho}$, we can choose an element from the $t$-th set $\hat{\rho}'_t \in \hat{\rho}_t$ uniformly at random in time $O(|\hat{\rho}_t|)$, for $t=1$, $t=2$, and so forth.
Thus, computing $\hat{\rho}'$ takes $O(LT)$ time.
Since $\hat{\rho}$ is an $L$-list accurate prediction for $\calset$, we have that $\Pr[ \rho_t = \hat{\rho}'_t] \geq 1/L$ for all $\rho \in \calset,~ t \in [T]$.
\end{proof}
The lemma's reduction, from $L$-list accurate to $\eps$-accurate algorithms, shows that lower bounds against algorithms with $L$-list accurate predictions also yield lower bounds against algorithms with $\eps$-accurate predictions for a problem.

\begin{corollary}\label{cor:LB-list-to-LB-eps} 

Suppose there is a dynamic problem $\problem$ such that any algorithm with $L$-list accurate prediction requires worst-case  total time $\Omega(Tf(n,L))$ for the request sequences in some $\calset\subseteq(\domain \cup \queries)^T$ and $f(n,L)=\Omega(L)$, then any algorithm with $\eps$-accurate prediction for the uniform distribution on $\calset$ requires $g(n,\eps)=\Omega(f(n,L))$ amortized time per request, for any $\eps \geq 1/L$.
\end{corollary}

Though this shows that $L$-list predictions are a stronger notion than $\eps$-accurate predictions, 
will show in Section~\ref{sec:striangle-lb-list} that for a wide range of problems, there called \emph{locally correctable} problems%
, strong lower bounds (similar to Proposition~\ref{prop:eps-accuracy-lower-bound}) hold: Any algorithm using $L$-list accurate predictions is subject to the same conditional lower bounds as a prediction-less online algorithm, unless the list accurate predictions are perfect, i.e. $L = 1$.

We thus seek to investigate even more powerful prediction models in the following.

\subsection{Bounded Delay Predictions}

Unlike $\eps$-accurate and $L$-list accurate predictions that aimed at predicting every individual time step $t\in[T]$, we further investigate
predictions that know all $T$ requests in advance, though the actual order of a request sequence may have various `small deviations' from the predicted sequence of requests.
That is, we will measure prediction accuracy by a notion of closeness for permutations.
We consider a permutation $\pi \in \permT(T)$ as a bijective map $\pi: [T] \rightarrow [T]$ on the integers $[T]$,
e.g. $\pi(1)$  is the first element of the permutation.

\begin{definition}
    \label{def:max-permutation-distance}
    Let $d \geq 0$ and $\pi, \sigma \in \permT(T)$.
    We call $\pi$ and $\sigma$ {\bf $d$-close}, denoted $|\pi - \sigma|_\infty \leq d$, if 
    $
        \left| \pi^{-1}(t) - \sigma^{-1}(t) \right| \leq d 
    $ for all $t \in [T]$.
    Further, $|\pi-\sigma|_1 = \sum_t |\pi^{-1}(t) - \sigma^{-1}(t) |  $ is called the {\bf total-distance} of $\pi$ and $\sigma$.
\end{definition}
To simplify exposition, we overload the notation of a permutation $\pi \in \permT(T)$ to yield a reordering of a request sequence of length $T$.
\begin{definition}%
    \label{def:request-permutation}
    For a request sequence $\rho \in (\domain\cup\queries)^T$ and $\pi \in \permT(T)$, %
    let $\pi(\rho) = \left(\rho_{\pi(1)}, \dotsc, \rho_{\pi(T)}\right)$ be the request sequence obtained by {\bf reordering the requests} in $\rho$ according to~$\pi$.
\end{definition}

Next we formalize what it means for predicted request sequence $\hat{\rho} \in (\domain \cup \queries)^T$ to be a bounded delay prediction for a set of input request sequences.

\begin{definition}[Bounded Delay Predictions]
    \label{def:delayed-predictions}
    Let $\calset \subseteq (\domain\cup \queries)^T$ be a set of request sequences of length $T$ and $\hat{\rho} = \left(\hat{\rho}_1, \hat{\rho}_2, \dotsc, \hat{\rho}_T\right)$ a given sequence of $T$ predicted requests.

    Then $\hat{\rho}$ has at most $d$ delay for $\calset$, called {\bf $d$-delayed} for $\calset$, if for all $\rho \in \calset$, there exists some permutation $\pi$ with $\pi(\rho)=\hat{\rho}$, and $\pi$ is $d$-close $|\pi-{\sf id}|_\infty\leq d$ to the identity permutation ${\sf id}$.

    Further, $\hat{\rho}$ has at most $d$ total-delay for $\calset$, called {\bf $d$-total-delayed} for $\calset$, if for all $\rho \in \calset$, there exists some permutation $\pi$ with $\pi(\rho)=\hat{\rho}$, and $\pi$ has at most $d$ total-distance $|\pi-{\sf id}|_1\leq d$ to the identity permutation ${\sf id}$.
\end{definition}

Clearly, every $d$-delayed prediction for $\calset$ has at most $dT$ total-delay for $\calset$.
Next, we show that bounded delay predictions are a stronger notion than list-accurate predictions, which were a stronger notion than $\eps$-accurate predictions (see Lemma~\ref{lem:list-accuracy-eps-accuracy}).

\begin{lemma} \label{lem:delay-list-accuracy}
    Given an integer $d\geq 0$ and a request sequence $\hat{\rho} \in (\domain\cup \queries)^T$, one can compute in $\Theta((d+1)T)$ time a $(2d+1)$-list accurate prediction $\hat{\rho}'=(\hat{\rho}'_1,\ldots,\hat{\rho}'_T)$ for all request sequence sets $\calset \subseteq (\domain \cup \queries)^T$ for which $\hat{\rho}$ is $d$-delayed.
\end{lemma}
\begin{proof}
    Since request $\hat{\rho}_t$ can appear in a sequence $\rho \in \calset$, with delay at most $d$, only in $\rho_{[t-d,t+d]}$, we can compute a list-prediction $\hat{\rho}'$ from $\hat{\rho}$ for a given $d$ by taking as list prediction for the $t$-th request the list $\hat{\rho}'_t:= \{ \hat{\rho}_{\tau} : \tau \in [t-d,t+d] \}$.

    Given $d \geq 0$ and $\hat{\rho}=(\hat{\rho}_1,\ldots,\hat{\rho}_T)$, this takes $\Theta((d+1)T)$ time.
\end{proof}

Though the lemma requires that integer $d\geq 0$ is given as input, as opposed to the reduction in Lemma~\ref{lem:list-accuracy-eps-accuracy}, we can still reduce, from bounded delay to $L$-list accurate algorithms, using an `Alternating Parallel Simulation' that allows to search for a $2$-approximation of a minimum $d$ value for an online input $\rho$, in a way that is efficient in the amortized sense.

\begin{lemma}[Alternating Parallel Simulation]\label{lem:parallel-simulation-search}
    Let $\rho, \hat{\rho} \in (\domain \cup \queries)^T$ and $d^* \in [0,T]$ be minimal such that prediction $\hat{\rho}$ has delay at most $d^*$ for the request sequence $\rho$.
    Suppose there is an algorithm $\cal A'$ that solves, given an $L$-list accurate prediction, request sequences of length $T$ of a dynamic problem $\cal P$ in time $O(T f'(n,L))$ after at most $P'(n,L)$ preprocessing time.
    If $LT + P'(n,L)= O(Tf'(n,L)) $,
    then there is an algorithm $\cal A$ that solves $\rho$, given the delay prediction $\hat{\rho}$, in 
    $O(f'(n,4d^*+1) \log T)$ 
    amortized time per request, after $O(T)$ time for preprocessing of $\hat{\rho}$.
\end{lemma}

Note that the $LT + P'(n,L)= O(Tf'(n,L)) $ condition on $\cal A'$ is very mild, i.e.\ its preprocessing of a list prediction of size $O(LT)$ takes not more time than solving a request sequence for which the prediction is $L$-list accurate.

\begin{proof} 
    We consider $O(\log T )$ values to find an approximation $d$ with $d^* \in [d/2,d]$, i.e. $d\in \{0,1,2,4,\ldots\}$.
    The algorithm $\cal A$ for problem $\cal P$ will spawn copies of algorithm $\cal A'$ and selectively pause/resume their computation.
    To spawn ${\cal A}'_d$ for a $d$ value, we first compute a list prediction of size $L=(2d+1)$ using Lemma~\ref{lem:delay-list-accuracy}.
    With this list-prediction, the copy ${\cal A}'_d$ starts its preprocessing and its computation for all online requests $(\rho_1, \rho_2,\ldots )$ that arrived thus far.
    For each copy, we track the total time ${\cal T}_d$ spent thus far as $\cal A$ chooses to pause/resume individual copies.
    (Note that a copy ${\cal A}'_i$ may be paused/resumed several times while still in its preprocessing phase.)
    We define for each copy ${\cal A}'_d$ its \emph{progress}, which is the number of completed requests from the online sequence $(\rho_1,\rho_2,\ldots)$ divided by ${\cal T}_d$, i.e.\ total execution time (including preprocessing and computation) spend thus far.

    Initially, algorithm $\cal A$ spawns only one copy of ${\cal A}'_0$ for $d=0$.
    To avoid that the runtime $\Theta(Td)$ of Lemma~\ref{lem:delay-list-accuracy} dominates overall execution for $d\gg d^*$, algorithm $\cal A$ delays starting the computation of the list-prediction for the next larger value $d$ until at least one of the spawned copies $\{{\cal A}'_i\}$ has spent a total execution time ${\cal T}_i$ of at least $Td$ thus far.
    If this threshold is met, then $\cal A$ pauses the parallel simulation of all spawned copies, executes the computation of Lemma~\ref{lem:delay-list-accuracy}, and resumes the parallel simulation of all spawned copies (including the new copy) afterwards.
    Further, the parallel simulation of algorithm $\cal A$ pauses/resumes any one of the spawned copies $\{ {\cal A}'_i\}$, if total execution time is not within a constant factor of the total execution time of the \emph{fastest progressing copy}, i.e. the copy with maximal progress.
    Note that the fastest progressing copy is not paused by the parallel simulation that $\cal A$ performs, and that all spawned copies (are allowed to) spent at least as much total execution time as the fastest progressing copy.
    To answer the $t$-th online request $\rho_t$ in case it is a query, $\cal A$ simply takes the result from the fastest progressing copy.
    This completes the description of algorithm $\cal A$.

    From the $O(\log T)$ values, let $d'$ be the value that minimizes the total execution time of $\cal A'$ on the request sequence $(\rho_1,\ldots,\rho_T)$.
    (Note that $d'$ is the fastest progressing copy when at the last online request $\rho_T$.)
    We will show next that the total execution time of $\cal A$ to finish all online requests in $\rho$ is, amortized over the $T$ requests, at most

\[
    O\left(\frac{Tf'(n,2d'+1) + P'(n,2d'+1)+Td'}{T}~\log T \right) = O(f'(n,4d^*+1)\log T)~.
\] 
    Clearly, the number of spawned copies is $O(\log T)$ at all times.
    Since none of the copies spends more than a constant of the total execution time of the fastest progressing copy, the sum of the total execution times of all spawned copies $\sum_i {\cal T}_i$ is bounded within a $O(\log T)$-factor of the time of the copy that uses the value $d'$, which has ${\cal T}_{d'}=O(Tf'(n,2d'+1)+P'(n,2d'+1))$.
    It remains to argue for the runtime cost due to executing Lemma~\ref{lem:delay-list-accuracy}.
    Since the size of the input list prediction is upper bounded by the total execution time for solving with this list prediction, i.e. $(2d'+1)T = O(Tf'(n,2d'+1))$, algorithm $\cal A$ must spawn the copy that has value $d'$.
    Further, since any copy with a specific $d$ value is only spawned if $Td < \max_i {\cal T}_i$, we have that the cost $\Theta(dT)$ is negligible in the amortized sense.
    Finally, $f'(n,2d'+1)=O(f'(n,4d^*+1)$ since $d^* \in [d'/2,d']$.
\end{proof}

Note that the `alternating parallel simulation' technique to show the reduction in the previous lemma is quite general, though we only use it to reduce from algorithms with bounded delay to algorithms with list-prediction (i.e. taking Lemma~\ref{lem:delay-list-accuracy}).
The reduction in the previous lemma immediately yields the following general, lower bounds.

\begin{corollary} \label{cor:LB-delay-to-LB-list}
    Suppose there is a dynamic problem $\cal P$ such that any algorithm with $d$-delayed prediction processes $T$ requests on instances of size $n$ requires time $\Omega(T f(n,d))$.
    Then any algorithm with $L$-list predictions for $\cal P$ requires amortized time $\tilde{\Omega}(f(n,(d-1)/4))$ per request.
\end{corollary}

Clearly, there are sets of request sequences that do not admit delay predictions with small~$d$.
In Section~\ref{sec:locally-reducible-lb-bounded-delay}, we will however show that for many problems with OMv-based lower bounds, it is possible to construct sets of request sequences $\calset$ admitting bounded predictions while being simultaneously powerful enough to express an arbitrary OMv instance.
Concretely, for the class of locally reducible dynamic problems (Definitions \ref{def:fully-dynamic-locally-reducible} and \ref{def:partially-dynamic-locally-reducible}) we will show lower bounds for bounded delay predictions (even with no outliers) in Section~\ref{sec:locally-reducible-problems}.

\section{Extensions of the OMv Conjecture} 
\label{sec:omv-extensions}
Before discussing our lower bounds against general dynamic problems, we revisit generalizations and extensions of the OMv conjecture.
The OMv and OuMv conjectures generalize to  non-square dimensions (i.e. Definition~2.1 and 2.6 in~\cite{HenzingerKNS15}). 
To state this, we need to introduce the $\tilde{\tilde{o}}$-notation for multivariate functions  (cf.~\cite[Definition~1.2]{HenzingerKNS15}).

\begin{definition}[polynomially lower $\tilde{\tilde{o}}$-notation]
    \label{def:tto}
    For $f: \N^3 \to \N$ and any constants $c_1,c_2,c_3 \geq 0$, we write  
     $f(n_1, n_2, n_3) = \tto{n_1^{c_1} n_2^{c_2} n_3^{c_3}}$  
    if and only if 
    there exist constants $\eps, N, C > 0$ such that 
    $f(n_1,n_2,n_3)\leq C(n_1^{c_1 - \eps} n_2^{c_2} n_3^{c_3} + n_1^{c_1} n_2^{c_2 - \eps} n_3^{c_3} + n_1^{c_1} n_2^{c_2} n_3^{c_3 - \eps})$ for all $n_1,n_2,n_3 > N$. 
    
    We use the analogous definition for functions with one or two parameters.
\end{definition}

Recall that the standard $\tilde{O}$ and $\tilde{o}$-notation suppresses factors that are \emph{polylogarithmic} in the problem size. %

\begin{definition}[Rectangular $\gamma$-OMv and $\gamma$-OuMv]
Let $\gamma > 0$ be a fixed constant.
An algorithm for the $\gamma$-OMv (resp. $\gamma$-OuMv) problem is given parameters $n_2, n_3$ as its input.
Next, it is given a Boolean matrix $M$ of size $n_1 \times n_2$ that can be preprocessed, where $n_1 := \floor{n_2^{\gamma}}$.
This is followed by $n_3$-rounds of processing online input vectors.

A {\bf \em $\gamma$-OMv} algorithm is given an online sequence of $n_3$ vectors $\vec v_1, \ldots , \vec v_{n_3}$, one vector after the other, and the task is to report each result of Boolean Matrix-Vector multiplication $M \vec v_t$ before $\vec v_{t+1}$ arrives.

A {\bf \em $\gamma$-OuMv} algorithm is given an online sequence of $n_3$ vector pairs $(\vec u_1, \vec v_1), \ldots , (\vec u_{n_3} , \vec v_{n_3} )$, one pair after the other, and the task is to report each result of Boolean Vector-Matrix-Vector multiplication
$(\vec u_t )^\top M \vec v_t$ before $(\vec u_{t+1}, \vec v_{t+1} )$ arrives.

The {\bf \em $\gamma$-uMv} problem is the special case of $\gamma$-OuMv with $n_3 = 1$.
\end{definition}

Clearly, the OMv and OuMv problems are the special cases of $\gamma$-OMv and $\gamma$-OuMv with $\gamma=1$ and $n_1 = n_2 = n_3 = n$.
The OMv conjecture implies an analogous lower bound for the $\gamma$-OuMv problem (cf. Theorem 2.2 and 2.7 in~\cite{HenzingerKNS15}).

\begin{theorem}[Hardness of $\gamma$-OMv and $\gamma$-OuMv]
    \label{thm:hardness-gen-oumv}
    For any constant $\gamma > 0$, the OMv conjecture implies that there is no algorithm for $\gamma$-OMv with parameters $n_2, n_3$ that has preprocessing time $P( n_2) = \poly(n_2)$, total running time for all requests of $\tto{n_1 n_2 n_3}$, where $n_1 = \floor{n_2^{\gamma}}$, and error probability at most $1/3$.
    
    For any constant $\gamma$, the OuMv conjecture implies that there is no algorithm for $\gamma$-OuMv with parameters $n_2, n_3$ that has preprocessing time $P( n_2) = \poly(n_2)$, total running time for all requests of $\tto{n_1 n_2 n_3}$, and error probability at most $1/3$.
\end{theorem}

It is possible to solve the OMv problem faster than $\Theta(n^3)$.
Green~Larsen and Williams~\cite{LarsenW17} gave a non-combinatorial OMv algorithm that runs in $O(n^3/2^{\Omega(\sqrt{\log n})})$ time.
Williams~\cite{Williams07} gave a combinatorial algorithm that, after $O(n^{2+\eps})$ preprocessing, solves any OMv round in $O(n^2/\log^2 n)$ time.
Chakraborty, Kamma and Larsen~\cite{ChakrabortyKL18} settled the cell probe complexity, showing that any data structure storing $r\in (n,n^2)$ bits must have a query time $t$, i.e. the number of reads from memory cells, with $r\cdot t = \Omega(n^3)$ and that this lower bound is tight, by giving an algorithm with $r=t=\tO{n^{3/2}}$ cell probes.

\subsection{Sparse OMv Conjecture}

We show in this section that the difficulty of the OMv and OuMv problem ``degrades gracefully'' with increased sparsity of query vectors.
For any integer $n$, let $[n]$ denote the set $\set{1, 2, \dotsc, n}$.

    The {\bf \em support} of a vector $\vec{v} \in \R^n$ is the set of indices where $\vec{v}$ is non-zero, i.e. 
    
    \begin{equation}
        \supp(\vec{v}) = \set{~i \in [n] ~:~ \vec{v}[i] \neq 0~}~,
    \end{equation}
    
    and the {\bf \em restriction} $\vec{v} \mid_{S}$ of $\vec{v}$ to an index-subset $S$ is the vector from $\R^n$ that has in the $k$-th component
    
    \begin{equation}
        \left(\vec{v} \mid_S\right)[k] = \begin{cases}
            \vec{v}[k] & k \in S \\
            0 & k \not\in S
        \end{cases} \quad.
    \end{equation}

Next we define the problem variants that have sparse input vectors, with respect to fixed sets of indices.

\begin{definition}[Sparse $S$-$\gamma$-OMv and $S$-$\gamma$-OuMv]
    The \emph{$S$-$\gamma$-OMv} problem differs from the $\gamma${-OMv} problem only by having an additional input $S_2 \subseteq [n_2]$
    of size $|S_2| \leq {n_2^t}$ for some $t \in (0,1]$, which is given during the preprocessing phase. %
    In the online phase, each of the $n_3$ query vectors $\vec{v}_i$ must fulfill support $\supp(\vec{v}_i) \subseteq S_2$.

    The \emph{$S$-$\gamma$-OuMv} problem differs from the $\gamma$-OuMv problem only by having an additional input $S_1 \subseteq [n_1]$ and $S_2 \subseteq [n_2]$ of size $|S_1| \leq {n_1^{t}}$ and $|S_2| \leq {n_2^t}$ for some $t \in (0,1]$, which are given during the preprocessing phase. %
    In the online phase, each of the $n_3$ pairs of query vectors $(\vec{v}_i,\vec{u}_i)$ has support $\supp(\vec{u}_i) \subseteq S_1$ and $\supp(\vec{v}_i) \subseteq S_2$.

\end{definition}

Clearly, each $S$-$\gamma$-OMv query round can be answered in $O(n_2^{\gamma+t})$ time and each $S$-$\gamma$-OuMv query round can be answered in $O(n_2^{\gamma t+t})$ time.

\begin{conjecture}[$S$-$\gamma$-OMv and $S$-$\gamma$-OuMv]
    \label{thm:sparse-omv-conjecture}
    Let $n_2, n_3, t, \gamma$ be parameters for the $S$-$\gamma$-OMv and $S$-$\gamma$-OuMv problem.
    
    There is no $\tto{n_2^{\gamma+t} n_3}$-time algorithm that solves the $S$-$\gamma$-OMv problem with an error probability of at most $1/3$ after processing in time  polynomial in $n_2$.

    There is no $\tto{n_2^{\gamma t + t} n_3}$-time algorithm that solves the $S$-$\gamma$-OuMv problem with an error probability of at most $1/3$ after processing in time  polynomial in $n_2$.

\end{conjecture}

For the $S$-OMv and $S$-OuMv problems setting $n_1 = n_2 = n_3 = n$, the above conjecture is equivalent to saying there is no $\tto{n^{2 + t}}$ algorithm for the $S$-OMv problem and no $\tto{n^{1 + 2t}}$ algorithm for the $S$-OuMv problem.

\begin{theorem}
    \label{thm:sparse-oumv-hardness}
    Conjecture \ref{thm:sparse-omv-conjecture} for $S$-$\gamma$-OMv ($S$-$\gamma$-OuMv) is true if Conjecture \ref{thm:omv-conjecture} for $\gamma$-OMv ($\gamma$-OuMv) is true.
\end{theorem}

\begin{proof}
    Consider an $\gamma$-OMv instance with the parameters $n_1, n_2, n_3$, matrix $M$ and vectors $\set{\vec{v}_k}$.
    Recall that $n_1=\floor{n_2^\gamma}$.
    
    For contradiction, suppose there exists a $t \in (0, 1)$  %
    where there is an algorithm solving $S$-$\gamma$-OMv instances in time $\tto{n_1 n_2^t n_3}$.
    Partition the set $[n_2]$ into $\ell = O(n_2^{1 - t})$ sets, where $\ell - 1$ have size $\floor{n_2^t}$ and one has size at most $\floor{n_2^t}$.
    Label each set $S_1, S_2, \dotsc, S_{\ell}$.
    Then, for each query $\vec{v}_k$, construct $\ell$ restricted vectors $\vec{v}_{k, i} := \vec{v}_k \mid_{S_i}$ by taking only the non-zero entries of indices in $S_i$.
    Since we can answer the $\gamma$-OMv query $\vec{v}_k$ by summing the $\ell$ restricted results, i.e.
    
    \begin{equation*}
        M \vec{v}_k = \sum_{i = 1}^{\ell} M \vec{v}_{k, i}~,
    \end{equation*}
    we can solve the entire $\gamma$-OMv instance in time
        $O(n_2 n_3)+\ell \cdot\tto{ n_1 n_2^t  n_3} = \tto{n_1 n_2 n_3}$, 
    since (1) solving with the assumed algorithm takes $\ell \cdot \tto{n_1 n_2^t n_3 }$ time and (2) computing those sums takes $O(n_2)$ time in each round and $O(n_2 n_3) = \tto{n_1 n_2 n_3}$ since $\gamma>0$. 
    This contradicts the hardness of $\gamma$-OMv under the OMv conjecture. 

    Consider now an $\gamma$-OuMv instance.
    Partition $[n_1]$ into $\ell_1 = O(n_1^{1 - t})$ sets where $\ell_1 - 1$ have size $\floor{n_1^t}$ and one has size at most $\floor{n_1^t}$.
    Partition $[n_2]$ into $\ell_2 = O(n_2^{1 - t})$ sets where $\ell_2 - 1$ have size $\floor{n_2^t}$ and one has size at most $\floor{n_2^t}$.
    Analogously, we construct $\ell_1$ many vectors $\vec{u}_{k, i} := \vec{u}_k \mid_{S'_i}$ and $\ell_2$ many vectors $\vec{v}_{k,j} := \vec{v}_k \mid_{S_j}$ given the OuMv vectors $(\vec{u}_k, \vec{v}_k)$.
    Then, we can compute,
    
    \begin{equation*}
        \vec{u}^\top_k M \vec{v}_k = \sum_{i, j} \vec{u}_{k, i}^\top M \vec{v}_{k, j}~,
    \end{equation*}
    
    thus solving the $\gamma$-OuMv instance in time
        $O(\ell_1 \ell_2 n_3) + \ell_1 \ell_2 \tto{ n_1^t n_2^t n_3} = \tto{n_1 n_2 n_3}$,   %
    contradicting the hardness of $\gamma$-OuMv under the OuMv conjecture.
\end{proof}

\section{Locally Correctable Problems: Lower Bounds for List-accurate Predictions}
\label{sec:locally-correctable}
We show in this section that certain problems, which we formally define in Definition~\ref{def:locally-correctable}, allow for remarkably strong lower bounds, in contrast to our general reductions in Section~\ref{sec:prediction-models}.

\subsection{Preliminaries: Edge Updates in Dynamic Graphs}

We will primarily focus on dynamic graphs with edge updates.
Formally, a special case of Definition \ref{def:gen-dynamic-model} is the classic edge update model for problems in dynamic $n$-vertex graphs.

\begin{definition}[Edge-Updates and Queries in Dynamic Graphs]
    \label{def:edge-flipping-model}
    Let $\problem$ be a dynamic graph problem.
    Let $V$ be a set of $n$ vertices.
    Let $\domain(V) = \set{(u, v) \in V \times V \given u \neq v}$ denote the set of possible edge flip updates.
    In an undirected graph, $\domain$ contains all unordered pairs of vertices, while in a directed graph $\domain$ contains all ordered pairs.
    $\queries(V)$ denotes the set of queries that are possible for $\cal P$.
    When the underlying graph is clear, we omit $V$ and write $\domain, \queries$.
    In the pre-processing step, the algorithm receives as input an initial graph $G_0$ on vertices $V$.
    At each time step $t$, the algorithm receives some request $\rho_t \in \domain \cup \queries$.
    When given a query $\rho_t \in \queries$, the algorithm must answer the query correctly on the current graph $G_t$, obtained by applying request sequence $(\rho_1, \rho_2, \dotsc, \rho_{t - 1})$ to the initial graph $G_0$.
    The query must be answered before the following request $\rho_{t + 1}$ is revealed.
\end{definition}

In the Maximum Matching problem for example, the query set $\queries$ consists of a single element $q$, resembling \emph{`What is the size of a maximum matching in the current graph?'}.
In above's edge update model of dynamic graphs, a sequence of $T$ requests (updates or queries) arrive in an online manner, one request after the other.
To study the potential of algorithms with predictions for the offline-online gap of dynamic problems, we assume that the algorithm is given in advance, i.e. for pre-processing, some form of prediction for the $T$ requests in the online phase.

For a dynamic graph $G$ on $n$ vertices and an update sequence $\rho = (\rho_1, \rho_2, \dotsc, \rho_T)$, we denote with $G_0 = (V, E_0)$ the  initial graph and with $G_t(\rho) = (V, E_t(\rho))$ the graph after applying the $t$-th request of $\rho$, i.e. $E_t(\rho)$ is the edge set after applying all updates in the first $t$ requests to $E_0$.
When the request sequence is clear, we omit $\rho$ and write $G_t = (V, E_t)$.

We also use the notion of an \emph{edge flip}: An edge flip of edge $e$ inserts $e$ if it is currently not in the graph and removes it otherwise.

\subsection{An OuMv Reduction for the \texorpdfstring{$\striangleF$}{$s$-triangle} Problem}
\label{sec:striangle-lb-list}
We now motivate our definition of locally correctable problems by the example of a simple, conditional lower bound construction for the $\striangleF$ problem.

In the $\striangleF$ problem, each query asks to report the number of triangles in a dynamic $n$ vertex graph that contain a fixed vertex $s$ (cf. \cite{HenzingerKNS15}).
As a warm up, we give a lower bound in for the online setting (without predictions).

\begin{figure}[b]
    \centering
    \includegraphics[width=0.7\textwidth]{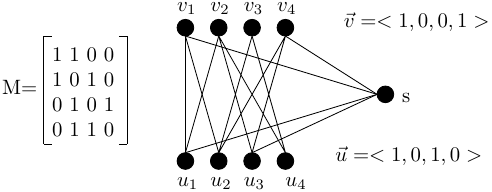}
    \caption{A small example of the OuMv lower bound construction for the $\striangleF$ problem. %
    The diagram shows that a $1$ in the matrix or in vectors corresponds to an edge existing in the graph.}
    \label{fig:sTriangle}
\end{figure}

\begin{theorem}
\label{thm:striangle-lb-worst}
There is no algorithm solving the $\striangleF$ problem in dynamic $n$ vertex graphs with update time $U(n)$, query time $Q(n)$, and pre-processing time $P(n)= poly(n)$, satisfying
\begin{equation*}
    n^2 U(n) + nQ(n) = \tto{n^{3}} ~,
\end{equation*}
if the OuMv conjecture (Conjecture \ref{thm:omv-conjecture}) is true. 
\end{theorem}

\begin{proof}
We design a sequence of $\striangleF$ updates such that any algorithm correctly answering all queries solves the OuMv problem.

{\bf Setup Phase.}
Consider an $1$-OuMv instance with parameters $n_1, n_2, n_3 = n$.
We construct a graph $G_0$ on $2n + 1$ vertices $\set{u_1, \dotsc, u_{n}} \cup \set{v_1, \dotsc, v_{n}} \cup \set{s}$.
$G_0$ contains no edges, beside edges of form $\set{u_i, v_j} \in E$ if and only if $M[i][j] = 1$.

{\bf Dynamic Phase.}
In each round $k$, we are given vectors $\vec{u}_k, \vec{v}_k$.
We use at most $2n$ updates to ensure that $\set{s, u_i} \in E$ if and only if $\vec{u}_k[i] = 1$ and $\set{s, v_j} \in E$ if and only if $\vec{v}_k[j] = 1$.
For any edge $\set{u, v}$, let $\chi(u, v)$ be the indicator for whether $\set{u, v} \in E$.
To answer the query of the OuMv round, we observe that 

\begin{equation*}
    \vec{u}_k^\top M \vec{v}_k = \sum_{i, j \in [n]} \vec{u}_k[i] \cdot M[i][j] \cdot \vec{v}_k[j] = \sum_{i, j \in [n]} \chi(u_i, s) \cdot \chi(u_i, w_j) \cdot \chi(w_j, s) = \striangleF~.
\end{equation*}
Thus, the number of triangles containing $s$ is $\striangleF > 0$ if and only if the OuMv round must report the $1$-bit.
Overall, this takes $n$ queries and $O(n^2)$ updates to answer the OuMv rounds.
Thus, we cannot have for the $\striangleF$ problem that

\begin{equation*}
    n^2 U(n) + n Q(n) = \tto{n^{3}}~,
\end{equation*}
if Conjecture~\ref{thm:omv-conjecture} is true.
\end{proof}

Next, we show a lower bound against algorithms with list-accurate predictions.
Unless the predictions are perfect, i.e. $2^0$-list accurate, any algorithm solving the $\striangleF$ problem is still subject to the same OuMv-based lower bound of Theorem \ref{thm:striangle-lb-worst}.

\begin{theorem}
    \label{thm:striangle-list-pred-lb}
    Suppose there is an algorithm $\innerAlg$ that, given $2^1$-list accurate predictions, solves the $\striangleF$ problem in polynomial preprocessing time $P(n)$, update time $U(n)$, and query time $Q(n)$.
    Then, $U(n)$ and $Q(n)$ cannot satisfy
    
    \begin{equation*}
        n^2 U(n) + n Q(n) = \tto{n^3}~,
    \end{equation*}
    if the OuMv conjecture (Conjecture \ref{thm:omv-conjecture}) is true.
\end{theorem}

\begin{proof}
    We will use the structure of the request sequences from the previous proof to show that there is \emph{one} generic $2$-list accurate prediction that is suitable for all OuMv instances.
    
    Consider an OuMv instance of size $n$ with matrix $M$ and query vectors $\set{(\vec{u}_k,\vec{v}_k)}$.
    We will reduce to a $\striangleF$ instance with $2n + 2$ vertices $\set{u_1,\ldots,u_n}\cup\set{v_1,\ldots,v_n}\cup \set{s,t}$.
    Unlike in the last proof, there are \emph{two} special vertices, $s$ and $t$, in this graph.
    \paragraph{Setup Phase.}
    We choose the initial graph to only contain edges that encode the matrix $M$, that is edge $\set{u_i,v_i}$ is present if and only if $M[i][j] = 1$.
    \paragraph{Dynamic Phase.}
    For the $k$-th OuMv query $(\vec{u}_k, \vec{v}_k)$, we proceed as follows.
    For each index $i \in [n]$, we ensure that $\set{s, u_i}$ is an edge if and only if $\vec{u}_k[i] = 1$, and that $\set{s, v_i}$ is an edge  if and only if $\vec{v}_k[i] = 1$.
    To do so, we perform exactly $n$ flips to satisfy the condition for $\vec{u}$, followed by exactly $n$ flips to satisfy the condition for $\vec{v}$.
    Whenever the condition is already satisfied, we flip $\set{t, u_i}$ instead of $\set{s, u_i}$ and $\set{t, v_i}$ instead of $\set{s, v_i}$, so that the flip has no influence on the number of $\striangleI$'s in graph $G$.
    For the queries we again have that,
    
    \begin{equation*}
        \vec{u}_k^\top M \vec{v}_k = \sum_{i, j \in [n]} \vec{u}_k[i] \cdot M[i][j] \cdot \vec{v}_k[j] = \sum_{i, j \in [n]} \chi(u_i, s) \cdot \chi(u_i, v_j) \cdot \chi(w_j, s) = \striangleF~,
    \end{equation*}
    where $\chi(u, v)$ is the indicator variable for the existence of edge $\set{u, v}$.
    The $\striangleF$ queries again exactly answer the OuMv rounds, proving the correctness of our reduction.
    
    For each OuMv round, this requires $O(n)$ updates. Thus, requiring $O(n^2)$ edge flips over all $n$ rounds.
    By the OuMv conjecture, any dynamic algorithm, with polynomial preprocessing time, correctly answering all queries on this $\striangleF$ request sequence cannot satisfy
    
    \begin{equation*}
        n^2 U(n) + n Q(n) = \tto{n^3}~,
    \end{equation*}
    as desired.

    To show the desired lower bound, we show that there is one, generic $2$-list accurate prediction that is suitable for all online OuMv request sequences:
        For each OuMv round $k \in [n]$, we use the prediction
        $\set{\set{s, u_i}, \set{t, u_i}}$ for each $i \in [n]$, followed by
        $\set{(s, v_i), (t, v_i)}$ for each $i \in [n]$, followed by one query.
    By our above discussion, this prediction is a $2$-list accurate prediction, regardless of the OuMv request sequence, and 
    can be constructed in $O(n^2)$ time during preprocessing.
    Therefore, no algorithm solving the $\striangleF$ problem with $2$-list accurate predictions can have update time $U(n)$ and query time $Q(n)$ that satisfy
    
    \begin{equation*}
        n^2 U(n) + n Q(n) = \tto{n^3}~,
    \end{equation*}
     if the OuMv conjecture (Conjecture \ref{thm:omv-conjecture}) is true.
\end{proof}

\subsection{Locally Correctable Dynamic Problems}
Next, we formalize this approach of OuMv lower bounds for algorithms with list accurate predictions in our definition of \emph{locally correctable} problems.

In the previous proof, we observed that every pair of query vectors can be simulated by taking a \emph{subsequence} of a universal update request sequence
    $(s, u_1), \dotsc, (s, u_n), (s, v_1), \dotsc, (s, v_n)$,
followed by one query, where we only flip the edges necessary to ensure the edges $(s, u_i)$ correctly encode $\vec{u}$ and the edges $(s, v_i)$ correctly encode $\vec{v}$.
However, in the general reduction (Theorem~\ref{thm:striangle-lb-worst}), the actual subsequence depends heavily on the specific instance. %
By \emph{augmenting} the graph with one dummy vertex $t$ that remains non-adjacent to $s$, we restricted the number of edge flips needed at any one time step, without affecting the query result, to a list of \emph{two} update requests.
That is flipping $(t, u_i)$ or $(s, u_i)$ for the bits in $\vec{u}$ and flipping $(t, v_i)$ or $(s, v_i)$ for the bits in $\vec v$.
In our $\striangleF$ example, the query computation in the augmented instance immediately yields the correct answer for the non-augmented instance, without further computations needed for \emph{correcting} the query results.

Next, we formally define the locally correctable problems and then prove the generalization of the technique in our reduction below.

\begin{definition}[Locally Correctable Dynamic Problem]
    \label{def:locally-correctable}
    Let $\problem$ be a dynamic problem.
    
    Suppose there is no algorithm for $\problem$ with update time $U(n)$ and query time $Q(n)$ satisfying
    \begin{equation*}
        n_3 \Big( 
        u(n_1, n_2) U(n) + q(n_1, n_2) Q(n) 
        \Big) 
        = \tto{n_1 n_2 n_3}
    \end{equation*}
    if the OuMv conjecture is true, where %
     $n_1, n_2, n_3$ are integers with $n_1 = \floor{n_2^{\gamma}}$ for some constant $\gamma > 0$, functions
     $u, q: \N \times \N \rightarrow \N$, and $n = n(n_1, n_2)$ is the size of the $\problem$ instance in the reduction.
    
    Then, $\problem$ is {\bf locally correctable} if there exists a universal sequence $\bar{\rho} = \bar{\rho}(n_1, n_2)$ of requests from $\domain \cup \queries$, an augmentation function $f$ for pre-processing, and a correction function $g$, satisfying:
    \begin{enumerate}
        \item The sequence $\bar{\rho}$ can be partitioned into $n_3$ subsequences $\bar{\rho} = B_1 \circ \ldots \circ B_{n_3}$, where block $B_k$ contains $u(n_1, n_2)$ updates and $q(n_1, n_2)$ queries.
        \label{cond:local-correctable:partition}
        \item For any $\gamma$-OuMv instance with $n_1 \times n_2$ matrix $M$ and query vector pairs $\set{(\vec{u}_k, \vec{v}_k)}_{k = 1}^{n_3}$ the reduction constructs initial data $D_0$ of problem $\problem$ and a request sequence $\rho$ satisfying:
        \label{cond:local-correctable:reduction}
        \begin{enumerate}
            \item $\rho$ is the concatenation of request sequences $B_1' \circ B_2' \circ \dotsc \circ B_{n_3}'$, where each $B_k'$ is a subsequence of $B_k$. 
            \item For each $\gamma$-OuMv request $(\vec{u}_k, \vec{v}_k)$, the result bit $\vec{u}_k^\top M \vec{v}_k$ can be computed in $O(q(n_1, n_2))$ time based on the answers given to the queries in $B_1' \circ B_2' \circ \dotsc \circ B_k'$.
        \end{enumerate}
        \item {\bf Augmentability:} $f(D_0)$ is an instance of size $O(n)$ and $\domain(D_0) \subset \domain(f(D_0))$.
        \label{cond:local-correctable:f-map}
        \item {\bf Correctability:} There is a non-empty subset $\domain^* \subseteq \domain(f(D_0)) \setminus \domain(D_0)$ such that, for all time steps $t$ and queries $q$, the function $g$ yields $g(q(Y), Y) = q(D_t)$, where $D_t$ is the data structure after request sequence $\rho_{\leq t}$, and $Y$ is the result of request sequence $\rho_{\leq t}$ with arbitrary requests from $\domain^*$ inserted, applied to $f(D_0)$.
        \label{cond:local-correctable:g-map}
        \item $f$ is computable in polynomial time and $g$ is computable in $\tO{1}$ time.
        \label{cond:local-correctable:efficiency}
    \end{enumerate}
\end{definition}

The request sequence $\bar{\rho}$ is universal in the sense that it does not depend on any specific $\gamma$-OuMv instance.
However, the sequence $\bar{\rho}$ depends on the dynamic problem $\problem$ and the reduction from $\gamma$-OuMv to $\problem$.
Specifically, $\bar{\rho}$ consists of all updates that \emph{might} be necessary in the reduction from $\gamma$-OuMv to encode a vector update $(\vec{u}_k, \vec{v}_k)$ into a $\problem$ instance.

The following theorem shows that all problems $\problem$, that have a reduction from OuMv satisfying Definition \ref{def:locally-correctable}, $2$-list accurate predictions offer no improvement over a dynamic algorithm with no predictions.

\begin{theorem}
    \label{thm:locally-correctable-lb}
    Suppose $\problem$ is a dynamic locally correctable problem. 
    Then there is no algorithm solving $\problem$ with $2$-list accurate predictions with update time $U(n)$ and query time $Q(n)$ satisfying
    \begin{equation*}
        n_3 \Big( 
        u(n_1, n_2) U(n) + q(n_1, n_2) Q(n) 
        \Big) 
        = \tto{n_1 n_2 n_3}
    \end{equation*}
    if the OuMv conjecture is true.
\end{theorem}

\begin{proof}
    We begin by constructing a reduction from $\gamma$-OuMv to $\problem$ that admits efficiently computable $2$-list accurate predictions.
    If an efficient algorithm with predictions exists, then we can design a dynamic algorithm without predictions as follows.
    First, we compute the efficiently computable predictions, and then run the algorithm with predictions as a sub-routine, violating the lower bound based on the OuMv conjecture.

    Consider a $\gamma$-OuMv instance with $n_1 \times n_2$ matrix $M$ and vector updates $\set{(\vec{u}_k, \vec{v}_k)}_{k = 1}^{n_3}$.
    By assumption, there is an initial data structure $D_0$ and request sequence $\rho' = B_1' \circ B_2' \circ \dotsc \circ B_{n_3}'$ such that each $\gamma$-OuMv request $(\vec{u}_k, \vec{v}_k)$ can be computed in $O(q(n_1, n_2))$ time given the answers to the queries in $B_1' \circ B_2' \circ \dotsc \circ B_k'$.
    We construct a new request sequence $\rho^* = B_1^* \circ B_2^* \circ \dotsc \circ B_{n_3}^*$ on the augmented initial data structure $f(D_0)$.
    Fix an arbitrary update $x^* \in \domain^*$.
    For each update in $B_k$, $B_k^*$ will contain $\rho_t$ if $\rho_t \in B_k'$ and $x^*$ otherwise.
    Consider a query at time step $t$.
    Let $Y_t$ be the current state of the data structure, that is $f(D_0)$ with request sequence $\rho^*_{\leq t}$ applied.
    By assumption, $g(q(Y_t)) = q(D_{t'})$ where $D_{t'}$ is $D_0$ with $\rho'_{\leq t'}$ applied and $\rho'_{\leq t'} \subset \rho'$ is the longest prefix such that $\rho'_{\leq t'} \subseteq \rho^*_{\leq t}$.
    Then, given the query computations after $B_1^* \circ B_2^* \circ \dotsc \circ B_k^*$, we can compute $\vec{u}_k^\top M \vec{v}_k$ in $\tO{q(n_1, n_2)}$ time as $g$ is efficiently computable and $B_1' \circ B_2' \circ \dotsc \circ B_k' \subseteq B_1^* \circ B_2^* \circ \dotsc \circ B_k^*$ contains all the queries required to compute $\vec{u}_k^T M \vec{v}_k$.

    Next, we claim that there is an efficient prediction for the above reduction.
    That is, consider the prediction $\hat{\rho} = \set{(\rho_t, x^*)}_{t \in [T]}$.
    Since $\rho^*$ contains either $\rho_t$ or $x^*$ at the $t$-th position, this is a $2$-list accurate prediction.
    Furthermore $\hat{\rho}$ is efficiently computable.
    
    Therefore, suppose there is an efficient algorithm with $2$-list accurate predictions.
    Then, given a $\gamma$-OuMv instance, we compute $f(D_0)$ and $\hat{\rho}$ in the preprocessing phase in polynomial time, providing this as the initial input to the algorithm with predictions.
    Then, we compute each vector update using the appropriate query computations from $\rho^*$, therefore obtaining a dynamic algorithm for the $\gamma$-OuMv instance.
    Thus, the update and query times must not satisfy
    
    \begin{equation*}
        n_3 \Big( 
        u(n_1, n_2) U(n) + q(n_1, n_2) Q(n) 
        \Big) 
        = \tto{n_1 n_2 n_3}~,
    \end{equation*}
    if the OuMv conjecture is true.
\end{proof}

\section{Locally Reducible Problems: Lower Bounds for Delay Predictions}
\label{sec:locally-reducible-lb-bounded-delay}
In this section, we will provide a framework for proving trade-off conditional lower bounds against algorithms with bounded delay predictions given a conditional lower bound against online algorithms (without predictions).
To do so, we introduce the notion of \emph{locally reducible} dynamic problems (Definitions \ref{def:fully-dynamic-locally-reducible} and \ref{def:partially-dynamic-locally-reducible}) and show that for this large class of problems, the OuMv-based lower bounds carry through to the setting of algorithms with $d$-delayed predictions.
For each problem, we show that there is a delay threshold (roughly the number of updates and queries used to process one round in the OuMv problem) below which predictions offer no benefit over a generic online algorithm. 
The basic idea is that, when allowed sufficient delay, every OuMv  request sequence can be generated by simply \emph{reordering} one generic request sequence. %
Thus, even with prediction, a locally reducible dynamic problem is still powerful enough to solve any $\gamma$-OuMv instance, and is, thus, still difficult to compute.
This implies that a prediction algorithm does not only need to know \emph{what} operations will happen, but also \emph{when} the operations will happen.
We also show that the lower bound degrades gracefully as the prediction quality surpasses this threshold.

\subsection{Locally Reducible Dynamic Problems}
\label{sec:locally-reducible-problems}

We now define the class of locally reducible dynamic problems. 
Then we show in Theorems \ref{thm:locally-reducible-lower-fully-dynamic} and \ref{thm:locally-reducible-lower-partially-dynamic} that for any locally reducible problem, OuMv-based lower bounds extend to dynamic algorithms with bounded delay predictions.
For a multi-set $S$, let $\uniqueSet(S)$ denote the set of elements that occur in $S$ at least once.
For two request sequences $\rho_1, \rho_2$, let $\rho_1 \circ \rho_2$ denote the concatenation of the two request sequences.
We define also the notion of a cyclic update.

\begin{definition}
    \label{def:cyclic-order}
    Consider a dynamic problem $\problem$ with updates $\domain$ and queries $\queries$.
    An update $x \in \domain$ has {\bf cyclic order} $\order(x)$ if for any request sequence $\rho$, inserting or removing exactly $\order(x)$ identical copies of $x$ into the sequence $\rho$ between indices $i_0, i_1$ does not change the results of any queries
    that are not part of $\rho$ between the indices $i_0$ and $i_1$. 
    We say $x$ is {\bf cyclic} if $1 \leq \order(x) = O(1)$.
\end{definition}

For example, consider an edge insertion and an edge deletion to not be two separate operations, but instead consider it one operation called \emph{(edge)  flipping}  (which inserts the edge if it exists and deletes it if it does not exist). 
Now note that edge flipping has cyclic order $2$.
In order to place strong bounds on the positions of each individual update and query, we will insert redundant updates so that the positions of requests are more predictable.
For example, if at the current time, an edge was predicted to be flipped 2 more times than it has already been flipped, we can flip the edge twice without changing the dynamic graph to correct this prediction error.
Having updates of small cyclic order therefore allow us to insert redundant updates without significantly blowing up the size of the problem instance and therefore weakening our lower bounds.

In a bit more detail, consider a typical OMv-based lower bound.
Let $\problem$ be a dynamic problem for which there is an OuMv lower bound.
First, for some fixed $\gamma$, there is a generic reduction from any arbitrary $\gamma$-OuMv instance with arbitrary parameters $n_1, n_2, n_3$ to an instance of $\problem$ of size $n = n(n_1, n_2)$.
In this reduction, each of the $n_3$ rounds of the $\gamma$-OuMv instance is simulated separately, i.e. for each round there are $u(n_1, n_2)$ updates and $q(n_1, n_2)$ queries for the problem instance $\problem$.
To simulate a single OuMv round, we choose some subset of necessary updates from a universal request block to correctly encode the $\gamma$-OuMv instance into the $\problem$ instance.
To construct our prediction, we predict for each block that the whole universal request set occurs.
Since the request sequence is a subset of the universal set, an update cannot occur in a block \emph{before} it is predicted to.
However, this prediction could very well predict a request to occur in a block long \emph{after} it is predicted to.
For example, in the reduction of Theorem~\ref{thm:striangle-lb-worst} if there is an index $i_0$ such that $\vec{u}_k[i_0] = 0$ for all $k$, then the edge $(s, u_{i_0})$ will never be flipped in the original request sequence.
To solve this, we add (after the query for the vector update $(\vec{u}_k, \vec{v}_k)$ has arrived and before the next vector update arrives)  $\order(x)$ copies of update $x$ whenever an update $x$ has occurred in fewer than $k - \order(x)$ request blocks when simulating the vector update $(\vec{u}_k, \vec{v}_k)$. These redundant requests do not change the result of any query computation and ensure that a request cannot occur more than $\order(x)$ blocks later than it is predicted to.

We now give separate definitions for fully dynamic and partially dynamic problems, beginning with the fully dynamic setting.

\begin{definition}[Fully Dynamic Locally Reducible]
    \label{def:fully-dynamic-locally-reducible}
    Let $\problem$ be a fully dynamic problem with update set $\domain$ and query set $\queries$.
    Suppose there is no algorithm for $\problem$ with update time $U(n)$ and query time $Q(n)$ satisfying
    
    \begin{equation*}
        n_3 \Big( 
        u(n_1, n_2) U(n) + q(n_1, n_2) Q(n) 
        \Big) 
        = \tto{n_1 n_2 n_3}~,
    \end{equation*}
    if the OuMv conjecture is true, where $n_1, n_2, n_3$ are integers satisfying $n_1 = \floor{n_2^{\gamma}}$ for some constant $\gamma > 0$, functions $u, q: \N \times \N \rightarrow \N$, and $n = n(n_1, n_2)$ is the size of the $\problem$ instance in the reduction.
    
    $\problem$ is {\bf $(u, q)$-locally reducible from $\gamma$-OuMv} if there exists a universal request sequence $\bar{\rho} = \bar{\rho}(n_1, n_2)$ of requests satisfying the following properties:
    \begin{enumerate}
        \item The universal request sequence $\bar{\rho}$ consists of $n_3$ identical copies of the sequence $B = B(n_1, n_2)$, indexed $B_1, B_2, \dotsc, B_{n_3}$.
        The request sequence $B$ contains $u(n_1, n_2)$ updates and $q(n_1, n_2)$ queries.
        \label{cond:local-reducible:partition}
        \item For any $\gamma$-OuMv instance with $n_1 \times n_2$ matrix $M$ and vector updates $\set{(\vec{u}_k, \vec{v}_k)}_{k = 1}^{n_3}$ the reduction constructs an initial data structure $D_0$ and request sequence $\rho$ satisfying:
        \label{cond:local-reducible:omv}
        \begin{enumerate}
            \item $\rho$ is the concatenation of request sequences $B_1' \circ B_2' \circ \dotsc \circ B_{n_3}'$ where $B_k' = \pi_k(B_k'')$ is an ordering of $B_k''$ for some $B_k'' \subseteq B$.
            \item Each $\gamma$-OuMv request $(\vec{u}_k, \vec{v}_k)$, $\vec{u}_k^\top M \vec{v}_k$ can be computed in $O(q(n_1, n_2))$ time based on the answers given to the queries in $B_1' \circ B_2' \circ \dotsc \circ B_k'$.
        \end{enumerate}
        \item Every update $x \in \bar{\rho}$ in the universal sequence $\bar{\rho}$ is cyclic. 
        \label{cond:local-reducible:cyclic}
    \end{enumerate}
\end{definition}

\begin{remark}
    Our proof assumes that each update has some small finite cyclic order.
    For example, the edge flip operation has order 2.
    Alternatively, we can view an update in a fully dynamic algorithm to have an inverse operation.
    For example, removal and insertion of the same edge are inverse operations.
    The proof of Theorem \ref{thm:locally-reducible-lower-fully-dynamic} for locally reducible fully dynamic problems follows in this case as well.
    Whenever we insert two edge flips in the proof of Theorem \ref{thm:locally-reducible-lower-fully-dynamic}, we can insert an operation and its inverse operation counterpart.
    Both of these sequences of two updates have the desired effect of leaving the underlying data structure unmodified.
\end{remark}

Let us compare the Definitions \ref{def:locally-correctable} and \ref{def:fully-dynamic-locally-reducible}.
In both cases, the universal request sequence $\bar{\rho}$ depends only on the reduction from $\gamma$-OuMv to the dynamic problem $\problem$.
It is universal in the sense that $\bar{\rho}$ is \emph{independent} of any specific $\gamma$-OuMv instance.

In Definition \ref{def:fully-dynamic-locally-reducible}, each block of the request sequence $\rho$ does not have to respect the order of $\bar{\rho}$. 
Each block of the request sequence in the reduction to a locally correctable problem must be a subsequence of the corresponding block in the universal request sequence.
In the reduction to a locally reducible problem, we may instead arbitrarily permute a subsequence of a block of the universal request sequence.
To see why this is the case, observe that a list accurate prediction imposes the constraint that a certain update can occur only at $O(1)$ time steps in each block (since only $2$ updates can occur at a given time step).
Instead, bounded delay predictions allow the update to be placed at any point within a range of the predicted update.
We are therefore free to order the subset of block of requests without being forced to adhere to the original order in the universal sequence.

Furthermore, instead of requiring that any instance can be efficiently augmented to a larger instance with an efficient ``correction" function to the query computations (Conditions \ref{cond:local-correctable:f-map}, \ref{cond:local-correctable:g-map}, and \ref{cond:local-correctable:efficiency}), we now require that each update is cyclic (Condition \ref{cond:local-reducible:cyclic}).
We now give the definition for partially dynamic locally reducible problems.

\begin{definition}[Partially Dynamic Locally Reducible]
    \label{def:partially-dynamic-locally-reducible}
    Let $\gamma > 0$ be a constant and $n_1, n_2, n_3$ be integers satisfying $n_1 = \floor{n_2^{\gamma}}$.
    Let $u, q: \N \times \N \rightarrow \N$.
    Let $\problem$ be an incremental (resp. decremental) dynamic problem with update set $\domain$ and query set $\queries$.

    Suppose there is no algorithm for $\problem$ with update time $U(n)$ and query time $Q(n)$ satisfying,
    \begin{equation*}
        n_3 \Big( u(n_1, n_2) U(n) + q(n_1, n_2) Q(n) \Big) = \tto{n_1 n_2 n_3}
    \end{equation*}
    if the OuMv conjecture is true, where $n = n(n_1, n_2)$ is the size of the $\problem$ instance in the reduction.
    $\problem$ is {\bf $(u, q)$-locally reducible from $\gamma$-OuMv} if there exists a universal request sequence $\bar{\rho} = \bar{\rho}(n_1, n_2)$ of requests (where only queries can occur more than once) satisfying the following properties:
    \begin{enumerate}
        \item The universal request sequence $\bar{\rho}$ consists of $n_3$ subsequences $\set{B_k}_{k = 1}^{n_3}$ where request block $B_k$ contains $u(n_1, n_2)$ updates and $q(n_1, n_2)$ queries.
        \item For any $\gamma$-OuMv instance with $n_1 \times n_2$ matrix $M$ and vector updates $\set{(\vec{u}_k, \vec{v}_k)}_{k = 1}^{n_3}$ the reduction constructs an initial data structure $D_0$ and request sequence $\rho$ satisfying:
        \begin{enumerate}
            \item $\rho$ is the concatenation of request sequences $B_1' \circ B_2' \circ \dotsc \circ B_{n_3}'$ where $B_k' = \pi_k(B_k)$ is an ordering of $B_k$.
            \item Each $\gamma$-OuMv request $(\vec{u}_k, \vec{v}_k)$, $\vec{u}_k^\top M \vec{v}_k$ can be computed in $O(q(n_1, n_2))$ time based on the answers given to the queries in $B_1' \circ B_2' \circ \dotsc \circ B_k'$.
        \end{enumerate}
    \end{enumerate}
\end{definition}

We are now ready to present our main lower bound result.

\begin{restatable}{theorem}{locallyreduciblelbfully}
    \label{thm:locally-reducible-lower-fully-dynamic}
    Let $\gamma > 0$ be a constant and $u, q$ be functions.
    Suppose $\problem$ is a fully dynamic problem that is $(u, q)$-locally reducible from $\gamma$-OuMv and let $C = \max_{x \in {\cal X}} \order(x)$.
    
    Then there is no algorithm solving $\problem$ with $(1+C)(u(n_1, n_2) + q(n_1, n_2))$ delayed predictions with update time $U(n)$ and query time $Q(n)$ satisfying
    \begin{equation*}
        n_3 \Big( u(n_1, n_2) U(n) + q(n_1, n_2) Q(n) \Big) = \tto{n_1 n_2 n_3}
    \end{equation*}
    if the OuMv conjecture is true.
\end{restatable}

We begin by defining some useful notation, denoting the position in which the $k$-th instance of a request occurs.

\begin{definition}
    \label{def:op-pos}
    Let $\domain$ denote the set of updates and $\queries$ the set of queries.
    Let $\rho \in (\domain \cup \queries)^T$ be a sequence of requests.
    For a given request $a \in \domain \cup \queries$ and $k \in \N$, define $\pos(a, k, \rho)$ to be the position in $\rho$ of the $k$-th occurrence of $a$.
    If $a$ does not occur $k$ times in $\rho$, $\pos(a, k, \rho) = \bot$.
    When the underlying request sequence is clear, we omit the sequence and write $\pos(a, k)$.
\end{definition}

We now prove Theorem \ref{thm:locally-reducible-lower-fully-dynamic}.

\begin{proof}
    The key ingredient for the lower bound will be a reduction from a $\gamma$-OuMv instance to the dynamic problem $\problem$.
    However, we will require the reduction to construct the online request sequence in such a way that we can efficiently compute a very simple prediction with $(1 + C)(u(n_1, n_2) + q(n_1, n_2))$ delay.
    Then, if an efficient algorithm $\innerAlg$ with $(1 + C)(u(n_1, n_2) + q(n_1, n_2))$ bounded delay predictions exists, we can solve the $\gamma$-OuMv problem by constructing the prediction and running the algorithm $\innerAlg$ as a sub-routine, violating the $\gamma$-OuMv lower bound.

    \paragraph{Preliminaries.}

    Since our reduction will be constructed by modifying an existing reduction, we begin by describing the existing reduction given by Condition \ref{cond:local-reducible:omv}.
    Consider a $\gamma$-OuMv instance $(M, {\cal U})$ consisting of a $n_1 \times n_2$ matrix $M$ and a length-$n_3$ sequence of vector updates ${\cal U} = \set{(\vec{u}_k, \vec{v}_k)}_{k = 1}^{n_3}$.
    We use $\rho'({\cal U}) = B_1' \circ B_2' \circ \dotsc \circ B_{n_3}'$ to denote the request sequence given by the reduction such that each $\gamma$-OuMv request $(\vec{u}_k, \vec{v}_k)$ can be computed in $O(q(n_1, n_2))$ time given the answers to the queries in $B_1' \circ B_2' \circ \dotsc \circ B_k'$.

    We now describe the universal request sequence $\bar{\rho}$.
    In the given $\gamma$-OuMv reduction, each vector update $(\vec{u}_k, \vec{v}_k)$ is encoded into the data structure using some set of updates.
    For any request $x \in \domain \cup \queries$, define $M(x)$ to be the maximum number of times an update $x \in \domain$ occurs in a single block $B_k' \subset \rho'({\cal U})$ over all $k$
    and all possible vector update sequences ${\cal U}$ (not just the worst-case one).
    The set $B$ then contains $M(x)$ copies of $x$ for all requests $x \in \domain \cup \queries$. 
    Note that if $M(x) = 0$, $B$ does not contain any copy of $M(x)$.
    Additionally, we give an arbitrary, fixed order to the requests in $B$, so that $B$ is a sequence.
    The universal request sequence $\bar{\rho}$ consists of $n_3$ copies of the sequence $B$.
    Define the predicted requested sequence $\hat{\rho} = \bar{\rho} = B_1 \circ B_2 \circ \dotsc \circ B_{n_3}$ to be the universal request sequence.
    
    Note that for a specific $\gamma$-OuMv instance $(M, {\cal U})$, we may not need every request in $B_k$ to encode the vector update $(\vec{u}_k, \vec{v}_k)$ (with $k \in [n_3]$) given the state of the data structure after the block computing the previous vector update $(\vec{u}_{k - 1}, \vec{v}_{k - 1})$.
    However, by our definition of $M(x)$, it is possible to encode the vector update using some subset $B_k' \subseteq B_k$.
    This is precisely the block of requests in the request sequence $\rho'({\cal U})$.
    
    \emph{Notation.} During the proof, we will focus on three request sequences.
    We use $\rho'$ to denote the request sequence that is generated by the reduction of a \emph{worst-case} $\gamma$-OuMv instance to  $\cal P$. 
    We denote by $\hat{\rho}$ the predicted request sequence constructed from the universal sequence $\bar{\rho}$ as discussed above. 
    Note that it can be constructed without knowledge of  $\rho'$.
    In this proof we will modify $\rho'$  into a request sequence $\rho^*$ encoding the same instance, with the additional property that $\pi(\rho^*) = \hat{\rho}$ for some permutation $|\pi - {\sf id}| \leq d$.
    Finally, we denote by $\rho$ an arbitrary request sequence for $\cal P$ partitioned into $n_3$ blocks.
    For any $1 \leq k \leq n_3$ let $\rho_{(k)}$ denote the first $k$ blocks of $\rho$ and
    for a request $x \in \domain \cup \queries$, let $N(x, k, \rho)$ denote the number of times that $x$ occurs in $\rho_{(k)}$.
    
    Our proof will proceed in three parts.
    In Part 1, we describe how to modify $\rho'$ into $\rho^*$ for any $\gamma$-OuMv instance.
    In Part 2, we show that $\pi(\rho^*) = \hat{\rho}$ for some permutation $|\pi - {\sf id}| \leq d$.
    In Part 3, we complete the proof by showing how an algorithm with bounded delay can be given $\hat \rho$ as prediction and can be used to answer any $\gamma$-OuMv instance.

    \subsubsection*{Part 1: Constructing $\rho^*$ from $\rho'$} 
    We begin by describing the construction of $\rho^*$. 
    In Lemma \ref{lemma:rho-*-properties}, we will argue that $\rho^*$ constructed from $\rho'$ correctly encodes the $\gamma$-OuMv instance, while satisfying certain properties that we will use in Part 2 to show that $\pi(\rho^*) = \hat{\rho}$ for some permutation $\pi$ that is $d$-close to the identity permutation. We construct $\rho^*$ sequentially, appending requests to the end of $\rho^*$. 
    Recall that we have constructed the universal request sequence $\bar{\rho}$ by imposing an arbitrary order onto $B$ and concatenating $n_3$ copies of $B$.
    Whenever we append a request $x$, we always append the copy of $x$ that occurs \emph{earliest} in the universal sequence $\bar{\rho}$ out of all requests of $\bar{\rho}$ that we have not already added to $\rho^*$.

    We proceed by induction on $k$.
    For $k = 1$, we create $B_1^*$ in three steps. 
    
    \begin{enumerate}
        \item We begin by initializing $B_1^*$ to $B_1'$.
        \item Then, for every update $x \in \domain$ satisfying $N(x, 1, \rho') \leq M(x) - \order(x)$, we append 
        
        $\bigFloor{\frac{M(x) - N(x, 1, \rho')}{\order(x)}} \cdot \order(x)$ copies of $x$ to the end of $B_1^*$. 
        Recall that we always append the copy of an update that occurs earliest in $\bar{\rho}$ first.
        \item Finally, for every query $q \in \queries$ such that $N(q, 1, \rho') < M(q)$, we append $M(q) - N(q, 1, \rho')$ copies of $q$ to the end of $B_1^*$. 
        Recall that we always append the copy of a query that occurs earliest in $\bar{\rho}$ first.
     \end{enumerate} 
    
    Denote this augmented sequence by $\rho^*_{(1)} = B_1^*$.
    For $k > 1$ we extend $\rho_{(k - 1)}^*$ to $\rho_{(k)}^*$ also in three steps.
    For $i \in \set{1, 2, 3}$, let $\rho^*_{(k), i}$ denote the sequence $\rho^*$ after step $i$ when extending $\rho^*_{(k - 1)}$ to $\rho^*_{(k)}$.
    Note $\rho^*_{(k)} = \rho^*_{(k), 3}$.
    
    \begin{enumerate}
        \item We begin by concatenating $B_k'$ to obtain $\rho^*_{(k), 1} \gets B_1^* \circ B_2^* \circ \dotsc \circ B_{k - 1}^* \circ B_k' = \rho^*_{(k - 1)} \circ B_k'$.
        We emphasize that we always append the copy of an update that occurs earliest in $\bar{\rho}$ first. 
        In particular, in this step we may in fact append a copy of $x$ from $B_j$ for $j < k$ rather than  from $B_k$.
        \item Then, for every update $x \in \domain$ such that $N\left(x, k, \rho_{(k), 1}^*\right) \leq k M(x) - \order(x)$, we append $\bigFloor{\frac{k M(x) - N\left(x, k, \rho_{(k), 1}^*\right)}  {\order(x)}} \cdot\order(x)$ copies of $x$ to the end of $\rho^*_{(k), 1}$, emphasizing that we always append the copy of an update that occurs earliest in $\bar{\rho}$ first. 
        \item For every query $q \in \queries$ such that $N\left(q, k, \rho_{(k), 2}^*\right) < k M(q)$, append $k M(q) - N\left(q, k, \rho_{(k), 2}^*\right)$ copies of $q$ to the end of $\rho^*_{(k), 2}$, emphasizing that we always append the copy of a query that occurs earliest in $\bar{\rho}$ first. 
     \end{enumerate}

    Finally, after the final block $B_{n_3}'$, we add in all remaining unused requests from the universal request sequence $\bar{\rho}$. 
    Thus $\rho^*$ contains exactly all requests of $\hat{\rho}$.

    We now claim that $\rho^*$ computes the same $\gamma$-OuMv instance as $\rho'$, while satisfying certain additional properties we will use in Part 2.
    \begin{lemma}
        \label{lemma:rho-*-properties}
        The constructed sequence $\rho^*$ satisfies the following properties.
        \begin{enumerate}
            \item \label{req-cond:update} For all $x \in \domain$ and $1 \leq k \leq n_3$,~ $k M(x) - \order(x) < N(x, k, \rho^*) \leq k M(x)$.
            \item \label{req-cond:query} For all $q \in \queries$ and $1 \leq k \leq n_3$,~ $N(q, k, \rho^*) = k M(q)$.
            \item \label{req-cond:compute} Each $\gamma$-OuMv request $(\vec{u}_k, \vec{v}_k)$, $\vec{u}_k^\top M \vec{v}_k$ can be answered in $O(q(n_1, n_2))$ time given the answers to the queries in the request sequence $\rho^*_{(k)} = B_1^* \circ B_2^* \circ \dotsc \circ B_{k}^*$.
        \end{enumerate}
    \end{lemma}
    \begin{proof}
        We proceed by induction on $k$.
        Let $k = 1$.
        We begin by verifying Condition \ref{req-cond:update}.
        After Step 1, we have $N(x, 1, \rho') \leq M(x)$ copies of $x$ in $B_1^*$, as $B_1' \subseteq B_1$.
        In Step 2, we append
        \begin{equation*}
            (M(x) - N(x, 1, \rho')) - \order(x) < \bigFloor{\frac{k \cdot M(x) - N(x, 1, \rho')}{\order(x)}} \order(x) \leq M(x) - N(x, 1, \rho')
        \end{equation*}
        copies of $x$ to the end of $B_1^*$ so that,
        \begin{equation*}
            M(x) - \order(x) < N(x, 1, \rho^*) \leq M(x)
        \end{equation*}
        Since we only append queries in Step 3, Condition \ref{req-cond:update} is satisfied.
        We satisfy Condition \ref{req-cond:query} with a simular argument, since we have $N(q, 1, \rho') \leq M(q)$ copies of $q$ in $B_1^*$ after Step 1 as $B_1' \subseteq B_1$, and we append $M(q) - N(q, 1, \rho')$ copies of $q$ in Step 3, therefore obtaining $M(q)$ copies of $q$ in $B_1^*$.
        
        To verify Condition \ref{req-cond:compute}, observe that the appended queries occur after $B_1'$, so that the answers to the queries in $B_1'$ are the same in $\rho'_{(1)}$ and $\rho^*_{(1)}$.
        Therefore, by Condition \ref{cond:local-reducible:omv} of the Theorem, we compute $\vec{u}_k^T M \vec{v}_k$ in $O(q(n_1, n_2))$ time given the answers to the queries in $B_1'$.

        Now, let $k > 1$.
        We begin with Condition \ref{req-cond:update}. 
        By the inductive hypothesis, $N(x, k - 1, \rho^*) \leq (k - 1) M(x)$.
        In Step 1, we append $B_k' \subset B_k$ which contains at most $M(x)$ copies of the update $x$ so that $N\left(x, k, \rho^*_{(k), 1}\right) \leq k M(x)$.
        In Step 2, we append
        \begin{align*}
            \left(k \cdot M(x) - N\left(x, k, \rho^*_{(k), 1}\right)\right) - \order(x) &< \bigFloor{\frac{M(x) - N\left(x, k, \rho^*_{(k), 1}\right)}{\order(x)}} \order(x) \\
            &\leq \left(k \cdot M(x) - N\left(x, k, \rho^*_{(k), 1}\right)\right)
        \end{align*}

        copies of $x$ to $\rho^*_{(k), 1}$.
        Since we append only queries in Step 3,
        \begin{equation*}
            N\left(x, k, \rho^*\right) = N\left(x, k, \rho^*_{(k), 3}\right) = N\left(x, k, \rho^*_{(k), 2}\right) \in \left[ k M(x) - \order(x) + 1, k M(x) \right]
        \end{equation*}

        Following a similar argument, we verify Condition \ref{req-cond:query} and note that,
        \begin{equation*}
            N\left(q, k, \rho^*\right) = N\left(q, k, \rho^*_{(k), 3}\right) = k M(q) 
        \end{equation*}

        Finally, we verify Condition \ref{req-cond:compute}.
        For all $k$, $B_k^* \setminus B_k'$ contains $l \cdot \order(x)$ copies of each update $x \in \domain$ for some $l \geq 0$, as all updates in $B_k^* \setminus B_k'$ are appended in Step 2.
        Furthermore, queries do not modify the underlying data structure and therefore do not affect the answers given to other queries.
        Then, the answers to the queries in $B_k'$ are the same in the request sequences $\rho', \rho^*$.
        Therefore, we may compute $\vec{u}_k^T M \vec{v}_k$ using the answers from the queries in $B_k'$ in $O(q(n_1, n_2))$ time.
    \end{proof}

    \subsubsection*{Part 2: Showing that $\pi(\rho^*) = \hat{\rho}$ for some $\pi \in \permT(T)$ such that $|\pi - {\sf id}| \leq d$}

    We show that $\rho^*$ can be obtained by re-ordering the predicted sequence $\hat{\rho}$, for some permutation $\pi$ that is $d$-close to the identity.
    In particular, this will show that $\hat{\rho}$ is a prediction with bounded delay $d$ for request sequence set $\calset$ consisting of all request sequences $\rho^*$ which can be produced in Step 1.
    
    Recall that $\hat{\rho}$ is the universal request sequence $\bar{\rho}$ obtained by concatenating $n_3$ copies of $B$.
    We will prove that $\hat{\rho}$ has $(C+1)(u(n_1, n_2) + q(n_1, n_2))$ delay with the following steps.
    First, in Lemma \ref{lemma:append-requests-block}, we show that any request occurs at most $O(1)$ blocks away from its predicted position.
    Then, in Lemma \ref{lemma:delay-block-size}, we bound the size of each block.
    Combining, we show in Lemma \ref{lemma:prediction-delay-bound} that we obtain an upper bound on the delay of prediction $\hat{\rho}$.

    \begin{lemma}
        \label{lemma:append-requests-block}
        Let $1 \leq k \leq n_3$.
        
        If $x \in B_k$ is an update, then $x \in B_i^*$ for $k \leq i \leq k + \bigCeil{\frac{\order(x)}{M(x)}} - 1$.
        
        If $q \in B_k$ is a query, then $q \in B_k^*$.
    \end{lemma}

    \begin{proof}
        Consider an update $x \in \domain$.
        Since we insert copies of $x$ in the order that they occur in the universal sequence $\bar{\rho}$ during the construction of $\rho^*$, the $j$-th copy of $x$ in $\hat{\rho}$ is the $j$-th copy of $x$ in $\rho^*$, for all $1 \leq j \leq n_3 M(x)$.

        First, we show that $x \in B_i^*$ for $k \leq i \leq k + \order(x) - 1$.
        Suppose for contradiction $i \leq k - 1$.
        Since $x \in B_k$ it is at least the $((k - 1) M(x) + 1)$-th occurrence of $x$ in $\rho^*$.
        Then, $N(x, i, \rho^*) \geq (k - 1) M(x) + 1 \geq i M(x) + 1$, contradicting Condition \ref{req-cond:update} of Lemma \ref{lemma:rho-*-properties}.
        Otherwise, suppose $i \geq k + \bigCeil{\frac{\order(x)}{M(x)}}$.
        Since $x \in B_k$ it is at most the $(k M(x))$-th occurrence of $x$ in $\rho^*$.
        Then, 
        \begin{align*}
            N\left(x, k + \bigCeil{\frac{\order(x)}{M(x)}} - 1, \rho^*\right) &\leq k M(x) - 1 \\
            &\leq k M(x) - M(x) \\
            &= k M(x) + \order(x) - M(x) - \order(x) \\
            &\leq \left(k + \bigCeil{\frac{\order(x)}{M(x)}} - 1\right) M(x) - \order(x)
        \end{align*}
        contradicting Condition \ref{req-cond:update} of Lemma \ref{lemma:rho-*-properties}.

        Now, we show $q \in B_k^*$.
        Since from Condition \ref{req-cond:query} of Lemma \ref{lemma:rho-*-properties}, $N(q, k, \rho^*) = k M(q)$ for all $k$, there are exactly $M(q)$ copies of $q$ in each block $B_k^*$.
        By definition, there are also exactly $M(q)$ copies of $q$ in each block $B_k$ of $\hat{\rho}$.
        Since the copies of $q$ are in the same order in $\rho^*$ as in $\hat{\rho}$, we have $q \in B_k^*$.
        
    \end{proof}

    \begin{lemma}
        \label{lemma:delay-block-size}
        Let $\rho^*$ be a request sequence as constructed in Part 1.
        Let $C = \max_{x \in B} \order(x)$.
        For all $1 \leq k \leq n_3$,
        \begin{equation*}
            k (u(n_1, n_2) + q(n_1, n_2)) - C \cdot u(n_1, n_2) < \left|\rho_{(k)}^*\right| \leq k (u(n_1, n_2) + q(n_1, n_2))
        \end{equation*}
    \end{lemma}

    \begin{proof}
         Fix an update $x \in B$.
         By Conditions \ref{req-cond:update} and \ref{req-cond:query} of Lemma \ref{lemma:rho-*-properties}, $\rho^*_{(k)}$ contains at least $k M(x) - \order(x) + 1$ and at most $k M(x)$ copies of $x$.
         If we fix a query $q \in B$, $\rho^*_{(k)}$ contains exactly $k M(q)$ copies of $q$.
         Summing over all unique updates and queries in $B$, we obtain
         \begin{equation*}
             |\rho_{(k)}^*| \leq \sum_{x \in B} k M(x) + \sum_{q \in B} k M(q) \leq k |B| = k (u(n_1, n_2) + q(n_1, n_2))
         \end{equation*}
         and 
         \begin{equation*}
             |\rho_{(k)}^*| > \sum_{x \in B} (k M(x) - \order(x)) + \sum_{q \in B} k M(q) \geq k (u(n_1, n_2) + q(n_1, n_2)) - C \cdot u(n_1, n_2)
         \end{equation*}
    \end{proof}

    \begin{lemma}
        \label{lemma:prediction-delay-bound}
        Let $\rho^*$ be a request sequence as constructed in Part 1.
        Let $C = \max_{x \in B} \order(x)$.
        Then, there exists permutation $\pi$ that is $(1 + C)(u(n_1, n_2) + q(n_1, n_2))$ close to the identity permutation and $\pi(\rho^*) = \hat{\rho}$.
    \end{lemma}

    \begin{proof}
        It suffices to show that for all requests $x \in B$ and $1 \leq j \leq n_3 M(x)$, the $j$-th copy of $x$ does not occur at an index more than $(1 + C)(u(n_1, n_2) + q(n_1, n_2))$ away from the index where the $j$-th copy of $x$ occurs in $\hat{\rho}$.
        Recall that $\pos(x, j, \rho)$ denotes the index in the request sequence $\rho$ where the $j$-th copy of request $x$ occurs.
        Our goal is then to bound the error $|\pos(x, j, \rho^*) - \pos(x, j, \hat{\rho})|$.

        Consider an update $x \in B$.
        Let $1 \leq j \leq n_3 M(x)$.
        The $j$-th copy of $x$ in $\hat{\rho}$ occurs in the $\bigCeil{j/M(x)}$-th block of $\hat{\rho}$.
        Then,
        
        \begin{equation*}
            \pos(x, j, \rho^*) > \left|\rho^*_{\left( \bigCeil{j/M(x)} - 1\right)}\right| \geq \left( \bigCeil{j/M(x)} - 1\right) (u(n_1, n_2) + q(n_1, n_2)) - C \cdot u(n_1, n_2)
        \end{equation*}
        where the first inequality follows from Lemma \ref{lemma:append-requests-block} and the second from Lemma \ref{lemma:delay-block-size}.
        Similarly,
        
        \begin{align*}
            \pos(x, j, \rho^*) &\leq \left|\rho^*_{\left( \bigCeil{j/M(x)} + \bigCeil{\order(x)/M(x)} - 1\right)}\right| \\
            &\leq \left( \bigCeil{\frac{j}{M(x)}} + \bigCeil{\frac{\order(x)}{M(x)}} - 1\right) (u(n_1, n_2) + q(n_1, n_2)) \\
            &\leq \left( \bigCeil{\frac{j}{M(x)}} + \bigCeil{\frac{C}{M(x)}} - 1\right) (u(n_1, n_2) + q(n_1, n_2)) 
        \end{align*}

        Since $\hat{\rho} = B_1 \circ B_2  \circ \dotsc \circ B_{n_3}$,
        
        \begin{align*}
            \pos(x, j, \hat{\rho}) &> \left|\rho^*_{\left( \bigCeil{j/M(x)} - 1\right)}\right| \geq \left( \bigCeil{j/M(x)} - 1\right) (u(n_1, n_2) + q(n_1, n_2)) \\
            \pos(x, j, \hat{\rho}) &\leq \left|\rho^*_{\left( \bigCeil{j/M(x)}\right)}\right| \leq \left( \bigCeil{j/M(x)} \right) (u(n_1, n_2) + q(n_1, n_2))
        \end{align*}
        as the $j$-th copy of $x$ must occur in the $\bigCeil{j/M(x)}$-th block of $\hat{\rho}$ and each block $B_k$ has size $u(n_1, n_2) + q(n_1, n_2)$.
        Combining our inequalities, we obtain,
        
        \begin{equation*}
            |\pos(x, j, \hat{\rho}) - \pos(x, j, \rho^*)| \leq (C + 1)(u(n_1, n_2) + q(n_1, n_2))
        \end{equation*}

        Consider now a query $q \in B$.
        In both sequences $\hat{\rho}, \rho^*$, the $j$-th copy of $q$ occurs in the $\bigCeil{j/M(x)}$-th block. 
        Again, we begin by bounding $\pos(q, j, \rho^*)$.
        Then,
        
        \begin{equation*}
            \pos(q, j, \rho^*) > \left|\rho^*_{\left( \bigCeil{j/M(x)} - 1\right)}\right| \geq \left( \bigCeil{j/M(x)} - 1\right) (u(n_1, n_2) + q(n_1, n_2)) - C \cdot u(n_1, n_2)
        \end{equation*}
        and
        
        \begin{equation*}
            \pos(x, j, \rho^*) \leq \left|\rho^*_{\left( \bigCeil{j/M(x)}\right)}\right| \leq \left( \bigCeil{j/M(x)}\right) (u(n_1, n_2) + q(n_1, n_2))
        \end{equation*}
        where in both equations, the first inequality follows from Lemma \ref{lemma:append-requests-block} and the second from Lemma \ref{lemma:delay-block-size}.
        On the predicted request sequence $\hat{\rho}$, we get the same bounds as for updates,
        
        \begin{equation*}
            \left( \bigCeil{j/M(q)} - 1\right) (u(n_1, n_2) + q(n_1, n_2)) < \pos(q, j, \hat{\rho}) \leq \left( \bigCeil{j/M(q)} \right) (u(n_1, n_2) + q(n_1, n_2))
        \end{equation*}
        Combining our inequalities again,
        
        \begin{equation*}
            |\pos(q, j, \hat{\rho}) - \pos(q, j, \rho^*)| \leq (C + 1)(u(n_1, n_2) + q(n_1, n_2))
        \end{equation*}
        completing the proof of the Lemma.
    \end{proof}

    \subsubsection*{Part 3: Proof of Theorem \ref{thm:locally-reducible-lower-fully-dynamic} for Fully Dynamic  Locally Reducible Problems}

    We now complete the proof of Theorem \ref{thm:locally-reducible-lower-fully-dynamic} for fully dynamic problems.
    Suppose there exists an algorithm $\innerAlg$ with $(1 + C)(u(n_1, n_2) + q(n_1, n_2))$ bounded delay predictions solving $\problem$ with polynomial preprocessing time, update time $U(n)$ and query time $Q(n)$.
    
    We design an algorithm $\outerAlg$ (that works without prediction) for the $\gamma$-OuMv problem. Let $(M, {\cal U})$ be a worst-case $\gamma$-OuMv instance.
    In the preprocessing step of $\outerAlg$, we compute the universal sequence $\bar{\rho} = \bar{\rho}(n_1, n_2)$ in polynomial time, construct the predicted request sequence $\hat{\rho} = \bar{\rho}$, and give $\hat \rho$ as input to $\innerAlg$.
    Note we do not need to see the matrix $M$ nor the request sequence ${\cal U}$ to construct $\hat{\rho}$.
    Recall that matrix $M$ is given to $\outerAlg$ during preprocessing. It gives $M$ to $\innerAlg$, which builds the initial data structure $D_0$. This completes the preprocessing phase.
    
    Next, given a vector update $(\vec{u}_k, \vec{v}_k)$, $\outerAlg$ constructs the sequence $B_k^*$ and asks $\innerAlg$ to perform this sequence of requests. $\innerAlg$ returns the correct answers to the requests in $B_k^*$ to $\outerAlg$ as, by Lemma~\ref{lemma:prediction-delay-bound}, $\hat \rho$ is
    a $(1 + C)(u(n_1, n_2), q(n_1, n_2))$ delayed prediction for $\rho' = \rho'({\cal U})$, and $\innerAlg$ is a correct algorithm when given $(1 + C)(u(n_1, n_2), q(n_1, n_2))$ delayed predictions.
    Thus, by Lemma \ref{lemma:rho-*-properties}, $\outerAlg$ can correctly answer $\vec{u}_k^T M \vec{v}_k$ in $O(q(n_1, n_2))$ time given the answers to the queries in $B_k'$.

    Let us analyze the complexity of $\outerAlg$.
    In the preprocessing phase, $\outerAlg$ constructs $\hat{\rho}$ and $D_0$, requiring only polynomial time.
    For each vector update, $\outerAlg$ computes $B_k^*$ in $O(u(n_1, n_2) + q(n_1, n_2))$ (Lemma \ref{lemma:rho-*-construction-efficiency}), asking $\innerAlg$ to perform the updates in $B_k^*$.
    Since $\outerAlg$ solves $\gamma$-OuMv, the OuMv conjecture states that $\innerAlg$ cannot satisfy,
    
    \begin{equation*}
        n_3 \Big( u(n_1, n_2) U(n) + q(n_1, n_2) Q(n) \Big) = \tto{n_1 n_2 n_3}
    \end{equation*}

    We conclude the proof by proving Lemma \ref{lemma:rho-*-construction-efficiency}.
    \begin{lemma}
        \label{lemma:rho-*-construction-efficiency}
        For all $1 \leq k \leq n_3$, $B_k^*$ can be constructed in $O(u(n_1, n_2) + q(n_1, n_2))$ time.
    \end{lemma}
    \begin{proof}
        Since $B_j^* \setminus B_j$ contains a multiple of $\order(x)$ copies of every update $x$, the state of the data structure after $B_{1}^* \circ \dotsc B_{k - 1}^*$ is the same as the state of the data structure after $B_{1}' \circ \dotsc B_{k - 1}'$.
        Then, given vector update $(\vec{u}_k, \vec{v}_k)$, construct $B_k'$ as promised by the $\gamma$-OuMv reduction.
        To compute $\rho^*_{(k), 2}, \rho^*_{(k), 3}$, we can keep count of the number of copies of $x$ inserted into $\rho^*$ so far.
        Computing the number of copies of $x$ to append to $\rho^*$ and then appending these requests to $\rho^*$ requires time $O(|B_k^*|)$.
        From Lemma \ref{lemma:delay-block-size}, we can conclude that $|B_k^*| = O(u(n_1, n_2) + q(n_1, n_2))$.
    \end{proof}
    
    This concludes our proof of Theorem~\ref{thm:locally-reducible-lower-fully-dynamic}.
\end{proof}

Next, we prove an analogous result for partially dynamic problems with a simpler argument.
We note that this lower bound gives a similar result to Theorem 1.3 in the independent work of~\cite{arxiv23}.

\begin{restatable}{theorem}{locallyreduciblelbpartially}
    \label{thm:locally-reducible-lower-partially-dynamic}
    Let $\gamma > 0$ be a constant and $u, q$ be functions.
    Suppose $\problem$ is a partially dynamic problem that is $(u, q)$-locally reducible from $\gamma$-OuMv.
    
    Then there is no algorithm solving $\problem$ with $(u(n_1, n_2) + q(n_1, n_2))$ delayed predictions with update time $U(n)$ and query time $Q(n)$ satisfying
    \begin{equation*}
        n_3 \Big( u(n_1, n_2) U(n) + q(n_1, n_2) Q(n) \Big) = \tto{n_1 n_2 n_3}
    \end{equation*}
    if the OuMv conjecture is true.
\end{restatable}

\begin{proof}
    For partially dynamic problems, we do not need to modify the reduction from $\gamma$-OuMv.
    In particular, we set $\rho^* = \rho'$.
    
    We again claim that the prediction $\hat{\rho} = \bar{\rho}$ has small bounded delay.
    Note that $\hat{\rho}$ has bounded delay at most $u(n_1, n_2) + q(n_1, n_2)$, since for any request $x$ the $j$-th occurrence of $x$ occurs in the same block in both the predicted sequence $\hat{\rho}$ and the actual sequence $\rho^*$.
    Furthermore, the sizes of the blocks in $\rho^*, \hat{\rho}$ are identical, both of size $|B_k| = u(n_1, n_2) + q(n_1, n_2)$.
    Therefore, suppose for contradiction that there is an algorithm $\innerAlg$ with update time $U(n)$ and query time $Q(n)$ satisfying,
    
    \begin{equation*}
        n_3 \Big( u(n_1, n_2) U(n) + q(n_1, n_2) Q(n) \Big) = \tto{n_1 n_2 n_3}
    \end{equation*}
    given $u(n_1, n_2) + q(n_1, n_2)$ bounded delay predictions.
    Then, we again have a pure dynamic algorithm $\outerAlg$ for $\gamma$-OuMv that constructs the prediction $\hat{\rho}$ in the preprocessing phase, providing this as input to $\innerAlg$.
    Then, given a vector update, $\outerAlg$ asks $\innerAlg$ to compute the request sequence $B_k'$, computing $\vec{u}_k^T M \vec{v}_k$ given the answers to the queries in $B_k'$.
    Following a similar argument to the fully dynamic case, $\outerAlg$ computes $\gamma$-OuMv in total time $\tto{n_1 n_2 n_3}$, contradicting the OuMv conjecture.
\end{proof}

Above, we have established that algorithms with predictions with delay $(1+C)(u(n_1, n_2) + q(n_1, n_2))$ 
with $C = \max_{x \in B} \order(x)$ cannot be more efficient than algorithms with no predictions at all.
In the following, we show that for smaller  delay the conditional lower bounds based on the OuMv conjecture degrade gracefully with the quality of the predictions.

\begin{theorem}
    \label{thm:locally-reducible-lower-d}
    Suppose $\problem$ is $(u, q)$-locally reducible from $\gamma$-OuMv with $n = n(n_1, n_2)$ a non-decreasing function.
    Let $t \in (0, 1)$ be a constant.
    Let $d_1 = \floor{d_2^{\gamma}}$ where $d_2 = \floor{n_2^t}$.
    
    Then there is no algorithm solving $\problem$ on instances of size $n = n(n_1, n_2)$ with $(1 + C)(u(d_1, d_2) + q(d_1, d_2))$ delayed predictions with update time $U(n)$ and query time $Q(n)$ satisfying
    
    \begin{equation*}
        n_3 \Big( u(d_1, d_2) U(n) + q(d_1, d_2) Q(n) \Big) = \tto{d_1 d_2 n_3} = \tto{n_1^t n_2^t n_3}
    \end{equation*}
    if the OuMv conjecture is true.
\end{theorem}

In particular, as the guaranteed prediction quality increases (as $t$ decreases towards 0), the lower bound weakens.

\begin{proof}
    Consider a $S$-$\gamma$-OuMv instance with $n_1 \times n_2$ matrix $M$ and vector requests $\set{(\vec{u}_k, \vec{v}_k)}_{k = 1}^{n_3}$.
    Suppose further that there are subsets $S_1 \subseteq [n_1]$ of size $d_1$ and $S_2 \subseteq [n_2]$ of size $d_2$ and $\supp(\vec{u}_k) \subset S_1$ and $\supp(\vec{v}_k) \subset S_2$ for all $k$.
    
    Now, observe that,
    
    \begin{equation*}
        \vec{u}_k^\top M \vec{v}_k = \sum_{i, j} \vec{u}_k[i] M[i][j] \vec{v}_k[j] = \sum_{i \in S_1, j \in S_2} \vec{u}_k[i] M[i][j] \vec{v}_k[j]
    \end{equation*}
    so only the values $M[i][j]$ where $i \in S_1, j \in S_2$ influence the final product.
    Consider then the $\gamma$-OuMv instance with $d_1 \times d_2$ matrix $M[S_1][S_2]$ and vector updates $\set{(\vec{u}_k[S_1], \vec{v}_k[S_2])}_{k = 1}^{n_3}$. 
    In particular, 
    Since $\problem$ is $(u, q)$-locally reducible from $\gamma$-OuMv, from Theorems \ref{thm:locally-reducible-lower-fully-dynamic} and \ref{thm:locally-reducible-lower-partially-dynamic} there is no algorithm solving $\problem$ with $(1 + C)(u(d_1, d_2) + q(d_1, d_2))$ delayed predictions with update time $U(d)$ and query time $Q(d)$ satisfying,

    \begin{equation*}
        n_3 (u(d_1, d_2) U(d) + q(d_1, d_2) Q(d)) = \tto{d_1 d_2 n_3}
    \end{equation*}
    where $d = n(d_1, d_2)$ is the size of the $\problem$ instance in the reduction.
    Since $n = n(n_1, n_2) \geq n(d_1, d_2) = d$, the lower bound also holds for any algorithm solving instances of size $n$.
    To conclude, note that $d_1 = \Theta(d_2^{\gamma}) = \Theta(n_2^{\gamma t}) = \Theta(n_1^t)$ and $d_2 = \Theta(n_1^t)$.
\end{proof}

\subsection{The \texorpdfstring{$\striangleF$}{$s$-triangle} Problem is Locally Reducible}
\label{sec:striangle-lb-delay}

Keeping the reduction from Theorem~\ref{thm:striangle-lb-worst} in mind, it is now easy to show that any algorithm with $O(n)$ delayed predictions is subject to the same conditional lower bound under the OMv Conjecture.

\begin{theorem}
    \label{thm:striangle-pred-lower}
    There is no algorithm solving the $\striangleF$ problem with $O(n)$ delayed predictions with update time $U(n)$ and query time $Q(n)$ satisfying
    
    \begin{equation*}
        n^2 U(n) + n Q(n) = \tto{n^{3}}
    \end{equation*}
    if the OuMv conjecture is true.
\end{theorem}

\begin{proof}
    We claim the $\striangleF$ problem is a fully dynamic locally-reducible problem.
    Set parameters $n_1 = n_2 = n_3 = n$.
    We claim that $\striangleF$ is $(n_1 + n_2, 1)$-locally reducible where $u(n_1, n_2) = n_1 + n_2$ and $q(n_1, n_2) = 1$.
    It is easily verified that the reduction of Theorem \ref{thm:striangle-lb-worst} from $\gamma$-OuMv to $\striangleF$ instances of size $n(n_1, n_2) = 1 + n_1 + n_2 = 2n + 1$ satisfies the required conditions, with $B_k = \set{(s, u_i)}_{i = 1}^{n} \cup \set{(s, w_i)}_{i = 1}^{n} \cup \set{q}$.
    Finally, we note that each update (edge flip) has cyclic order 2.
\end{proof}

In Theorem \ref{thm:striangle-delay-alg}, we will show that for any $d = O(n^{1 - \eps})$, there is an algorithm with $d$ delayed predictions overcoming the lower bound below.
In particular, we will show there is an update optimized algorithm with constant update time and $O(d^2)$ query time, as well as a query optimized algorithm with $O(d)$ update time and constant query time.
In particular, when the prediction quality is better than the linear threshold, there are algorithms with predictions that bypass OuMv-based lower bounds.
Next, we show that Theorem \ref{thm:striangle-delay-alg} is almost tight.
This follows immediately from Theorem \ref{thm:locally-reducible-lower-d} and the above observation that $\striangleF$ is locally reducible.

\begin{theorem}
    \label{thm:striangle-d-pred-lower}
    Let $t \in (0, 1)$ be a constant and $d = \floor{n^t}$.
    There is no algorithm solving the $\striangleF$ problem with $O(d)$-delayed predictions with update time $U(n)$ and query time $Q(n)$ satisfying
    
    \begin{equation*}
        n d U(n) + n Q(n) = \tto{n d^2}
    \end{equation*}
    if the OuMv conjecture is true.
\end{theorem}

In particular, either update time is not $\tto{d}$ or query time is not $\tto{d^2}$. 
Thus, Theorem \ref{thm:striangle-d-pred-lower} and Theorem \ref{thm:striangle-delay-alg} are (almost) tight.

\section{Further Locally Reducible and Locally Correctable Problems}
\label{sec:further-problems}

\label{sec:locally-reducible-examples}

We now give a list of examples of \emph{Locally Reducible Problems}, noting that this list includes almost all instances of dynamic problems that have OMv/OuMv-based lower bounds.
All referenced lower bounds are conditional on the OMv Conjecture unless otherwise stated.
Throughout this section, let $\delta \in (0, 1)$ be a constant.
Unless otherwise specified, all graphs are unweighted and undirected with $n$ vertices and $m$ edges.

\paragraph{Subgraph Connectivity}

In the subgraph connectivity problem, the algorithm is given a fixed graph $G$ and a subset $S \subset V$. 
Each update adds (resp. removes) a vertex $v$ to (resp. from) $S$.
Let $G[S]$ denote the subgraph of $G$ induced by $S$.
In the $(s, t)$ subgraph connectivity problem, each query asks for a fixed pair of vertices $(s, t)$ if $s, t$ are connected in $G[S]$.
The single source subgraph connectivity problem asks if a fixed source $s$ is connected to any vertex $v$ in $G[S]$.
The all pairs subgraph connectivity problems asks for any pair $u, v$ if they are connected in $G[S]$.
\cite{HenzingerKNS15} give a local reduction showing that no partially dynamic algorithm has worst case update time $\tto{n}$ and query time $\tto{n^2}$.
\cite{HenzingerKNS15} also give a local reduction $\frac{1 - \delta}{\delta}$-uMv to partially dynamic single source subgraph connectivity showing that no algorithm has update time $\tto{m^{\delta}}$ and query time $\tto{m^{1 - \delta}}$.
All of the above bounds hold for fully dynamic algorithm with amortized update and query time.

\paragraph{Reachability} 

In reachability problems, the algorithm is given an initial directed graph $G$.
Each update adds or removes an edge in $G$.
In the $(s, t)$ reachability problem, each query asks if for a fixed pair of vertices $(s, t)$ if $t$ is reachable from $s$.
The single source reachability problem asks if any vertex $v$ is reachable from a fixed source $s$.
The all pairs reachability (also known as Transitive Closure or TC) problem asks for any pair $u, v$ if $v$ is reachable from $u$.
A similar local reduction for the reachability problem shows similar bounds as in the subgraph connectivity problem.

\paragraph{Shortest Path} 

In the shortest path problem, the algorithm is given an initial graph $G$.
Let $\distT_G(u, v)$ denote the distance between vertices $u, v$ in the graph $G$, abbreviated $\distT(u, v)$ when the underlying graph is clear.
Each update adds or removes an edge in $G$.
The $(s, t)$-shortest path ($(s, t)$-SP) problem asks for a fixed pair $(s, t)$ the distance $\distT(s, t)$.
The single source shortest path (SSSP) problem asks for a fixed source $s$ and any vertex $v$ the distance $\distT(s, v)$.
The all pairs shortest path (APSP) problem asks for any pair $(u, v)$ the distance $\distT(u, v)$.
For any number $\alpha \geq 1$, an $\alpha$-approximate algorithm is required to return a distance estimate $\hat{d}(u, v)$ satisfying $\distT(u, v) \leq \hat{\distT}(u, v) \leq \alpha \cdot \distT(u, v)$.

A local reduction of \cite{HenzingerKNS15} shows that no partially dynamic $\left( \frac{5}{3} - \eps \right)$-approximation algorithm for $(s, t)$-shortest path can have worst case update time $\tto{n}$ and query time $\tto{n^2}$.
Furthermore, \cite{DBLP:conf/networking/HenzingerP021} show a similar local reduction if updates are restricted to only changing the weights of edges in the graph.
While updates in the model as defined are not cyclic, the reduction essentially toggles edge weights between the set $\set{1, 3}$, which is a cyclic request.

\cite{HenzingerKNS15} also show that no partially dynamic $(3 - \eps)$-approximation algorithm can have worst case update time $\tto{n^{1/2}}$ and query time $\tto{n}$, a lower bound that is extended to constant degree graphs by \cite{DBLP:conf/esa/HenzingerPS22}.
Furthermore, the same lower bound holds for expander and power law graphs in the fully dynamic case \cite{DBLP:conf/esa/HenzingerPS22}, while no algorithm can have update time $\tto{n^{(1 + t)/2}}$ and query time $\tto{n^{1 + t}}$ on graphs with maximum degree $n^t$.
No partially dynamic $(2 - \eps)$-approximation algorithm for SSSP can have worst case update time $\tto{m^{\delta}}$ and query time $\tto{m^{1 - \delta}}$.
The above bounds for partially dynamic algorithms hold also for fully dynamic algorithms with amortized complexity.

\cite{HenzingerKNS15} give a local reduction showing that there is no partially dynamic $(s, t)$-shortest path algorithm with total update time $\tto{m^{3/2}}$ and query time $\tto{m}$ or $(2 - \eps)$-approximate APSP algorithm with total update time $\tto{n^{\frac{2 - \delta}{1 - \delta}}}$ and query time $\tto{n^{\frac{\delta}{1 - \delta}}}$ for $\delta \leq \frac{1}{2}$.

On planar graphs, \cite{DBLP:conf/focs/AbboudD16} show that no fully dynamic APSP algorithm on weighted graphs can have amortized update time and query time $\tto{n^{1/2}}$.
On unweighted graphs, no fully dynamic algorithm can have amortized update time $u(n)$ and query time $q(n)$ satisfying 
\begin{equation*}
    \max(q(n)^2 u(n), q(n) u(n)^2) = \tto{n}
\end{equation*}
in particular showing that update and query time cannot both be $\tto{n^{1/3}}$.

\paragraph{Distance Spanners and Emulators}

Given an undirected, unweighted graph $G$, a subgraph $H \subset G$ is an $(\alpha, \beta)$ spanner if for every pair of vertices $x, y \in V$, $\distT_G(x, y) \leq \distT_H(x, y) \leq \alpha \distT_G(x, y) + \beta$.
A weighted graph $H'$ is an $(\alpha, \beta)$-emulator if it fulfills the above constraint, but is not necessarily a subgraph of $G$.
In particular, every spanner is an emulator but not vice versa.
In the dynamic spanner/emulator problem, the algorithm is given an initial graph $G_0$ and each update inserts/removes an edge and asked to maintain at each time step an $(\alpha, \beta)$ spanner of the dynamic graph.
\cite{BergamaschiHGWW21} show that there is no partially dynamic algorithm maintaining a $(1, n^{o(1)})$-emulator (and therefore spanner) with $\tto{m}$ edges, arbitrary polynomial preprocessing time, and total update time $\tto{m n}$.

\paragraph{Bipartite Maximum Cardinality Matching}

The algorithm is given an initial bipartite graph $G = (A, B)$.
Each update inserts/removes an edge with one endpoint in $A$ and another in $B$.
Each query asks for the cardinality of the maximum matching in $G$.
\cite{DBLP:conf/icalp/Dahlgaard16} gives a local reduction showing that no partially dynamic algorithm (even only on bipartite graphs) can have amortized update time $\tto{n}$ and query time $\tto{n^2}$, improving upon the $\tto{m^{1/2}}$ update time and $\tto{m}$ query time lower bound of \cite{HenzingerKNS15}.
\cite{DBLP:conf/esa/HenzingerPS22} extend these results to show that even for constant degree graphs, no partially dynamic algorithm can have amortized update time $\tto{n^{1/2}}$ and query time $\tto{n}$.
\cite{DBLP:conf/esa/HenzingerPS22} give the same bounds for fully dynamic algorithms on (not necessarily bipartite) expander graphs and power law graphs.
On graphs with maximum degree $n^t$, no fully dynamic algorithm can have amortized update time $\tto{n^{(1 + t)/2}}$ and query time $\tto{n^{1 + t}}$.
\cite{DBLP:conf/networking/HenzingerP021} also give a lower bound when updates are restricted to only change weights.
As with the case of shortest path, each update toggles the weight of an edge between $\set{1, 2}$, which is cyclic.

\paragraph{Maximum $(s, t)$ Flow}

The algorithm is given an initial graph with edge insertions and deletions as updates.
Each query asks to compute the maximum flow from a fixed source $s$ to a fixed sink $t$.
Combining the lower bound for maximum matching of \cite{DBLP:conf/icalp/Dahlgaard16} and standard reductions from bipartite matching to directed flows, we can also argue that there is no partially dynamic algorithm for maximum $(s, t)$ flow on unweighted directed graphs or weighted undirected graphs with amortized update and query time $\tto{n}$.

\paragraph{Triangle Detection and Counting}

In the triangle detection/counting problem, there is again a dynamic graph undergoing edge insertions/deletions.
Each query asks for the number of triangles in the graph $G$ (or if any triangle exists).
The $\striangleF$ problems asks for the number of triangles containing the vertex $s$ (or if any triangle exists).
\cite{HenzingerKNS15} show that no partially dynamic algorithm can have worst case update time $\tto{n}$ and query time $\tto{n^2}$, and the same bound holds for the amortized update time and query time of fully dynamic algorithms.

\paragraph{Densest Subgraph}

The algorithm is given an initial graph $G$.
Each update inserts/removes an edge.
For any subset $S \subset V$, let $E(S)$ denote the number of edges in the induced subgraph $G[S]$.
The density of the subgraph $G[S]$ is given by $|E(S)|/|S|$.
Each query asks for the density of the densest subgraph.
\cite{HenzingerKNS15} show that there is no partially dynamic algorithm with update time $\tto{n^{1/3}}$ and query time $\tto{n^{2/3}}$.
On constant degree graphs, expander graphs, and power law graphs, \cite{DBLP:conf/esa/HenzingerPS22} show that there are no fully dynamic algorithms with update time $\tto{n^{1/4}}$ and query time $\tto{n^{1/2}}$.

\paragraph{$d$-Failure Connectivity}

The algorithm is given a fixed graph $G_0$.
Each update restores the original graph $G_0$ and removes any $d$ vertices from the original fixed graph $G_0$.
Each query asks if $s, t$ are connected, for any pair $s, t \in V$.
For any constant $\delta \in (0, \frac{1}{2}]$ and $n, m = O(n^{1/(1 - \delta)})$, \cite{HenzingerKNS15} show that no algorithm has update time $\tto{nd}$ and query time $\tto{d}$.

\paragraph{Vertex Color Distance Oracle}

The algorithm is given a fixed graph $G_0$.
Each update changes the color of a vertex.
Each query asks for a given vertex $s$ and color $c$, what is the shortest distance from $s$ to any vertex of color $c$.
\cite{HenzingerKNS15} give a local reduction from $\frac{1 - \delta}{\delta}$-uMv to partially $(3 - \eps)$-dynamic vertex color distance oracles showing that no algorithm has update time $\tto{m^{\delta}}$ and query time $\tto{m^{1 - \delta}}$.
The result holds also for fully dynamic algorithm with amortized update and query time.

\paragraph{Weighted Diameter}

The algorithm is given an initial graph $G$.
Each update inserts/removes an edge in $G$.
Each query asks for the diameter of the graph.
\cite{HenzingerKNS15} show that there is no $(2 - \eps)$-approximate algorithm for $\set{0, 1}$-weighted diameter with update time $\tto{\sqrt{n}}$ and query time $\tto{n}$.
While the updates in the original problem (setting an edge to a weight) are not cyclic, we can modify the updates to cycle between the following states for each vertex pair $(u, v)$: no edge, weight 0 edge, weight 1 edge.
Each original update requires at most 2 cyclic updates, so that the lower bound still holds.

\paragraph{Strong Connectivity}

The algorithm is given an initial graph $G$.
Each update inserts/removes an edge in $G$.
Each query asks if the graph is strongly connected.
A similar local reduction for strong connectivity shows similar bounds as in the subgraph connectivity problem.

\paragraph{Electrical Flows}

Given a $(s, t)$-flow $f: E \rightarrow \R$, its \emph{energy} is defined as $\sum_{e} r(e) f(e)^2$ where $r(e)$ is the \emph{resistance} of an edge (see e.g. \cite{DBLP:conf/esa/GoranciHP17a}).
Generally, $r(e) = 1/w(e)$ is the inverse of the weight of an edge.
The $(s, t)$-electrical flow is the flow minimizing energy among all $(s, t)$-flows with unit value.
In the subgraph electrical flows problem, there is a fixed graph with a dynamic set of vertices undergoing insertions and deletions.
Each update either activates (inserts into the set) a vertex, or deactives (removes from the set) a vertex.
Each query asks for the $(s, t)$ electrical flow on the subgraph induced by the set of active vertices.
\cite{DBLP:conf/esa/GoranciHP17a} give a local reduction from $1$-OuMv to the subgraph electrical flows problem.

\paragraph{Erickson's Maximum Value Problem}

The algorithm is given an initial matrix of size $n \times n$.
Each update increments all the values in a given row or column.
Each query asks for the maximum value in the matrix.
\cite{HenzingerKNS15} show that there is no algorithm with update time $\tto{n}$ and query time $\tto{n^2}$.

\paragraph{Langerman's Zero Prefix Sum Problem}

The algorithm is given an initial array $A$ of $n$ integers.
Each update sets $A[i] = x$ for any $i, x$.
Each query asks if there exists $k$ such that $\sum_{i = 1}^{k} A[i] = 0$.
\cite{HenzingerKNS15} show that there is no algorithm with update time $\tto{\sqrt{n}}$ and query time $\tto{\sqrt{n}}$.
As with weighted diamater, the original problem does not have cyclic updates.
However, the reduction requires only a constant number of weights to be assigned to each array entry, and therefore we can replace each update with a constant number of cycles through the constant number of values.

\subsection{Locally Correctable Fully-Dynamic Problems}

We first note that for many fully dynamic graph problems listed in Section \ref{sec:locally-reducible-examples} are locally correctable if 
we define $f(G)$ to add two isolated vertices $w_0, w_1$ to the graph $G$.
Then, $\domain^*$ consists of the single update $(w_0, w_1)$.
For problems such as subgraph connectivity, reachability, shortest path, maximum flow, triangle detection, densest subgraph, $d$-failure connectivity, vertex color distance oracle, and electrical flows, it suffices to consider $g$ as the identity function, as the answer to the query depends only on the connected component in the original reduction.
For maximum matching, we can define $g(m) = m - \chi(w_0, w_1)$ to subtract 1 from the matching if and only if $(w_0, w_1)$ is an edge.
This is computable in constant time.

For diameter and strong connectivity, we will require a slightly more careful reduction, as $f(G)$ as defined above always has infinite diameter and is never strongly connected.
For strong connectivity, we use the reduction from $(s, t)$ reachability to strong connectivity of \cite{DBLP:conf/focs/AbboudW14}.
Given a graph $G$, $f(G)$ vertices $w_0, w_1$ and adds directed edges $(t, v)$ and $(v, s)$ for $v \notin \set{s, t}$.
Then, we claim that regardless of $(w_0, w_1) \in E$, $t$ is reachable from $s$ if and only if $f(G)$ is strongly connected.
If $t$ is reachable from $s$ in $G$ via path $P$, then every pair of vertices $u, v$ has the path $(u, s) \circ P \circ (t, v)$.
If $f(G)$ is strongly connected, there is a path from $s$ to $t$, which cannot use any of the additional edges, so that this path must be in $G$, and therefore $t$ is reachable in $G$ from $s$.

For weighted diameter, we consider the reduction of \cite{HenzingerKNS15}.
The reduction $f(G)$ picks a vertex $a$ and adds two vertices $w_0, w_1$ each connected to $a$ with weight 0 edges.
$\domain^*$ consists of turning on and off an edge of weight 0 between $w_0, w_1$.
Note that the diameter of $G$ is exactly that of $f(G)$, as $\set{a, w_0, w_1}$ can be treated as a single vertex regardless of the presence of edge $(w_0, w_1)$, and it suffices to set $g$ to the identity function.

Considering a non-graph problem, for Langerman's Zero Prefix Sum, we can easily define $f(x)$ to extend the array by one entry and set $\domain^*$ to be the update that sets the extra entry to $0$.
Clearly this does not affect the sum of any subarray, so it suffices to take $g$ as the identity.

\subsection{Locally Correctable Partially-Dynamic Problems}

Below, we briefly discuss a few examples of locally correctable partially dynamic problems.
For partially dynamic problems, we will require small modifications to the known reductions to ensure that the problem is locally correctable.

\paragraph{Shortest Path}

\cite{HenzingerKNS15} give a local reduction from $(s, t)$-shortest path to $1$-OuMv.
We recall their reduction and argue that $(s, t)$-shortest path is in fact a locally correctable problem.
We describe the reduction given by \cite{HenzingerKNS15} for the decremental setting, noting that a similar argument holds in the incremental setting.
Consider a $1$-OuMv instance with $n \times n$ matrix $M$ and vector updates $\set{(\vec{u}_k, \vec{v}_k)}_{k = 1}^{n}$.
Construct a bipartite graph $G_M$ with vertex sets $L, R$ of size $n$ and edges $(l_i, r_j)$ if and only if $M[i][j] = 1$.
Then, construct paths $P, Q$ of $n$ vertices each and connect all edges $(p_k, l_i)$ and $(q_k, r_j)$.
We define $s = p_1$ and $t = q_1$ to be the special vertices.
This defines the initial graph $G_0$.

In the original reduction, with each vector update $(\vec{u}_k, \vec{v}_k)$, we disconnect $(p_k, l_i)$ if $\vec{u}_k[i] = 0$ and $(q_k, r_j)$ if $\vec{v}_k[j] = 0$.
After querying for the $(s, t)$ shortest path, we disconnect any remaining edges $(p_k, l_i)$ and $(q_k, r_j)$, satisfying Conditions \ref{cond:local-correctable:partition} and \ref{cond:local-correctable:reduction}.
\cite{HenzingerKNS15} show that $\delta(s, t) \leq 2t + 1$ if $\vec{u}_k^T M \vec{v}_k = 1$ and $\delta(s, t) \geq 2t + 2$ otherwise.

Since we have to remove any remaining edges after the query step, we define each block of the universal request sequence to first disconnect all edges $(p_k, l_i), (q_k, r_j)$, query for the shortest $(s, t)$-path, and then disconnect all edges $(p_k, l_i), (q_k, r_j)$ again.
We need to remove all edges a second time to ensure that we can disconnect the remaining edges even when the request sequence is required to be a subsequence of the universal request sequence.

\paragraph{Maximum Matching}

\cite{DBLP:conf/icalp/Dahlgaard16} gives a local reduction from bipartite maximum matching to OuMv.
We will consider the incremental setting, as in \cite{DBLP:conf/icalp/Dahlgaard16}.
Consider a $1$-OuMv instance with $n \times n$ matrix $M$ and vector updates $\set{(\vec{u}_k, \vec{v}_k)}_{k = 1}^{n}$.
\cite{DBLP:conf/icalp/Dahlgaard16} constructs a graph with vertex sets $S, A, B, C, D, T$ each a bipartite graph on $2n$ nodes with $n$ nodes on the left and right side each.
Each $A, B, C, D$ consists of a perfect matching, connecting all edges $(x_i^l, x_i^r)$ for all $x \in \set{a, b, c, d}$ and $i \in [n]$.
Furthermore, connect $(b_i^r, c_j^l)$ if and only if $M[i][j] = 1$.
This constitutes the graph $G_0$.

Given a vector update $\vec{u}_k, \vec{v}_k$, \cite{DBLP:conf/icalp/Dahlgaard16} connects $(a_k^r, b_i^l)$ if $\vec{u}_k[i] = 1$ and $(c_j^r, d_k^l)$ if $\vec{v}_k[j] = 1$.
Next, add the edges $(s_k^r, a_k^l)$ and $(d_k^r, t_k^l)$ and query for a maximum matching.
Finally, add the edges $(s_k^l, s_k^r)$ and $(t_k^l, t_k^r)$ and any remaining edges $(a_k^r, b_i^l)$ and $(c_j^r, d_k^l)$.

Since we have to add remaining edges after the query, we define each block of the universal request sequence to 1) add all edges $(a_k^r, b_i^l), (c_j^r, d_k^l)$, 2) add edges $(s_k^r, a_k^l), (d_k^r, t_k^l)$, 3) query for the maximum matching, 4) add edges $(s_k^l, s_k^r), (t_k^l, t_k^r)$ and finally, 5) add all edges $(a_k^r, b_i^l), (c_j^r, d_k^l)$ again.

\paragraph{SSSP, APSP, and Transitive Closure}

\cite{HenzingerKNS15} give also local reductions from single source shortest path and $(2 - \eps)$-approximate all pairs shortest path to $1$-OuMv.
Following a similar approach to above, we copy all updates in the universal request sequence after the query, allowing the remaining updates to be computed while maintaining that the true request sequence is a subsequence of the universal request sequence.

\paragraph{Distance Spanners and Emulators}

The reduction of \cite{BergamaschiHGWW21} from OuMv to dynamic distance spanners and emulators can be augmented to a graph with a single additional edge disconnected from the remainder of the graph. 
Since each vector update is encoded into the dynamic graph by choosing a subsequence of some universal request sequence, distance spanners and emulators are a locally correctable problem.

\section{Dynamic Algorithms with Bounded Delay Predictions}
\label{sec:algs-with-predictions}
In this section, we give several dynamic algorithms with predictions to overcome conditional lower bounds under the OMv Conjecture.
While the overall computation time (including preprocessing time) may not be less than a prediction-less dynamic algorithm, we can maintain a data structure with more efficient updates and queries, whilst performing the expensive computations in the preprocessing phase.
We believe that this is justified since the precomputation can be done for a single prediction and then reused for any request sequence in the future for which the prediction is valid.
If the predictions are perfect, then we can simply %
handle updates and queries in constant time by returning the precomputed answers.
If predictions are reasonably accurate, our algorithms will %
make small, thereby efficient, adjustments to the precomputed answers.

We will use the following basic fact about predictions with bounded delay in our dynamic algorithms with predictions.

\begin{lemma} 
    \label{lemma:d-sub-sequence-containment}
    If $\hat{\rho}$ has at most $d$-delay for some set $\calset \subseteq (\domain \cup \queries)^T$, 
    then we have for each $\rho \in \calset$ that 
        $
        \hat{\rho}_{\leq t - d} ~\subseteq~ \rho_{\leq t} ~\subseteq~ \hat{\rho}_{\leq t + d}
        $
    holds for all $t \in [T]$.
\end{lemma}

\begin{proof}
    Let $\rho \in \calset$ be any request sequence.
    Since $\hat{\rho}$ is a $d$-delayed prediction for $\calset$, there is a permutation $\pi \in \permT(T)$ such that $\pi$ is $d$-close to the identity permutation ${\sf id}$ and $\hat{\rho} = \pi(\rho)$.

    Consider for a time step $j$ the request $\rho_{j} \in \rho$.
    Since $\pi$ and ${\sf id}$ are $d$-close, we have that
    \begin{equation*}
        d ~\geq~ \Big|\pi^{-1}\big(\pi(j)\big) - {\sf id}^{-1}\big(\pi(j)\big)\Big| ~=~ |j - \pi(j)|~,
    \end{equation*}
    which yields for $\pi(j)$ that
    $
        j - d \leq \pi(j) \leq j + d
    $.
    In particular, for all $\rho_{j} \in \rho_{\leq t}$, we have $\pi(j) \leq j + d \leq t + d$ and therefore $\rho_{j} \in \hat{\rho}_{\leq t + d}$.
    To show the other inclusion, consider some $\rho_{\pi(j)} \in \hat{\rho}_{\leq t - d}$.
    Then, $\pi(j) \leq t - d$ implies that $j \leq \pi(j) + d \leq t$ and therefore $\rho_{j} \in \rho_{\leq t}$.
\end{proof}

\subsection{Bounded Delay Predictions with Outliers} 

The bounded delay model requires that the prediction $\hat{\rho}$ and each sequence $\rho \in \calset$ have the same length and contain the same requests.
We will design algorithms with predictions that in fact will be able to handle small discrepancies between the set of predicted requests and actual requests.
In this section, we define a weaker notion of prediction by allowing that the set of elements, of predicted sequence $\hat{\rho}$ and request sequence $\rho \in \calset'$, may differ up to a small number $k\geq 0$ of \emph{outlier} elements, resulting in the following, more general prediction model (i.e. $\calset'\supseteq \calset$).

Our definition will use the following notation.
\begin{definition}%
    \label{def:subseq-indices}
    Let $\rho$ be a request sequence of length $T$ and $I \subset [T]$ a subset.
    The {\bf complement} of $I$ is denoted by $I^C = [T] \setminus I$.
    Then, the {\bf subsequence} $\rho_{I} = (\rho_i)_{i \in I}$ is the subsequence defined by taking the elements at index $i \in I$ of the request sequence $\rho$.
    
    For any $i \in I$, let $p(i, I) = |\set{1 \leq j \leq i \mid j \in I}|$ denote the {\bf index in $I$} where $i$ occurs.
\end{definition}

Note that $\rho_{[T]} = \rho$ and, for any request $\rho_t \in \rho$, we have  $p(t, I) \leq p(t, [T]) = t$ if $t \in I$.

\begin{definition}[Bounded Delay Predictions with Outliers]
    \label{def:delay-predictions-outliers}
    Let $\calset\subseteq (\domain \cup \queries)^T$ be a set of request sequences of length $T$, and $\hat{\rho} = \left(\hat{\rho}_1, \hat{\rho}_2, \dotsc, \hat{\rho}_T\right)$ a given sequence of $T$ predicted requests.

    Then $\hat{\rho}$ has at most $d$ delay and at most $k$ outliers for $\calset$, called {\bf $d$-delayed with $k$ outliers} for $\calset$, if  there exist for any $\rho \in \calset$\\
        (1) two sub-sequences $I, \hat{I} \subset [T]$ that both have a length $T'\geq T-k $, and\\ %
        (2) a $\pi \in \permT(T')$ such that $\pi(\rho_{I})=\hat{\rho}_{\hat{I}}$ and $\pi$ is $d$-close to the identity permutation.

    Relaxing (2) to at most $d$-total-delay is called {\bf $d$-total-delayed with $k$ outliers} prediction.
\end{definition}

Of course, a prediction sequence $\hat{\rho}$ is $d$-delayed if and only if $\hat{\rho}$ is $d$-delayed with $k=0$ outliers.
Next, we give a generalization of Lemma~\ref{lemma:d-sub-sequence-containment}.
At any time step, the symmetric difference between the predicted and actual request sequences is linear in the delay and number of outliers.

\begin{lemma}
    \label{lemma:d-k-sym-diff-size}
    Let $\hat{\rho}$ be a prediction that is $d$-delayed with $k$ outliers for $\calset \subseteq (\domain \cup \queries)^T$.
    For $\rho \in \calset$ and $t \in [T]$, let $D_t = \rho_{\leq t} \Delta \hat{\rho}_{\leq t}$ denote the symmetric difference between the set of the first $t$ predicted requests and the set of the first $t$ request that occurred in $\rho$ (including multiplicities).
    Then, the symmetric difference contains at most $|D_t| \leq 4k + 2d$ elements. %
\end{lemma}

\begin{proof}
    Let $\rho \in \calset$ be any request sequence.
    Since $\hat{\rho}$ is a $d$-delayed prediction with $k$ outliers, there are two sub-sequences $I, \hat{I} \subset [T]$ of length $T' \geq T-k$ and a permutation $\pi \in \permT(T')$ such that $\pi$ is $d$-close to the identity permutation ${\sf id}$ and $\pi(\rho_{I}) = \hat{\rho}_{\hat{I}}$.
    Let $\rho_{I^C}$ and $\hat{\rho}_{\hat{I}^C}$ be the sub-sequence of at most $k$ outlier requests in $\rho$ and $\hat{\rho}$, respectively.

    Consider a time step $t \in I$ with request $\rho_t$.
    Let $t' = p(t, I) \leq t$.
    Since $\pi$ and ${\sf id}$ are $d$-close, we have from
        $d \geq |\pi^{-1}(\pi(t')) - {\sf id}^{-1}(\pi(t'))| = |t' - \pi(t')| $
    that
    \begin{equation*}
        t' - d ~\leq~ \pi(t') ~\leq~ t' + d~.
    \end{equation*}
    From $ t' \in [t-k,t]$, %
    we that the index $\pi(t') \in [t-(k+d),t+d]$. 
    Since, $\pi(\rho_{I}) = \hat{\rho}_{\hat{I}}$ and $t \in I$, we note that $\rho_t \in \rho$ is shuffled to index $\pi(t')$ in $\hat{\rho}_{\hat{I}}$. %
    As there are at most $k$ outliers, $\rho_t$ is shuffled to index at most $\pi(t') + k \leq t + (k + d)$ in $\hat{\rho}$.
    Since outliers can only increase the index that $\rho_t$ is shuffled to in $\hat{\rho}$, we can also conclude that $\rho_t$ is shuffled to index at least $\pi(t')$. %
    We are now ready to bound $|D_t|$.
    
    First, we bound $\rho_{\leq t} \setminus \hat{\rho}_{\leq t}$.
    Consider a request at index $j \leq t$.
    At most $k$ indices in $[t]$ can be in $I^C$.
    Suppose $j \in I \cap [t]$ and let $j' = p(j, I)$ so that $j - k \leq j' \leq j$.
    From our above argument, the request $\rho_j$ is shuffled to at most index $j + (k + d) \leq t + (k + d)$ in $\hat{\rho}$.
    Therefore, at most $(k + d)$ requests in $I \cap [t]$ can be shuffled to index larger than $t$.
    Combined with the at most $k$ requests in $I^C$, we can, thus, bound $\rho_{\leq t} \setminus \hat{\rho}_{\leq t}$ to have size $(2k + d)$.
    
    Next, we bound $\hat{\rho}_{\leq t} \setminus \rho_{\leq t}$.
    Consider an index $\hat{j} \leq t$.
    Suppose $\hat{j} \in \hat{I} \cap [t]$, noting that most $k$ indices in $[t]$ can be in $\hat I^C$.
    Then, denote by $\hat{j}' = p(\hat{j}, \hat{I})$ the position of $\hat{j}$ in $\hat{I}$ and let $j'$ be such that $\hat{j}' = \pi(j')$ for some $j' = p(j, I)$ where $j$ is the index such that $j' = p(j,I)$.
    Since $\rho$ and $\hat{\rho}$ have at most $k$ outliers, we can upper bound,
    
    \begin{equation*}
        j ~\leq~ j' + k ~\leq~ \hat{j}' + (k + d) ~\leq~ \hat{j} + (k + d) ~\leq~ t + (k + d)
    \end{equation*}
    
    Since we want to bound $\hat{\rho}_{\leq t} \setminus \rho_{\leq t}$, we consider only the indices $j \geq t + 1$, of which there are at most $(k + d)$ indices satisfying $t + 1 \leq j \leq t + (k + d)$.
    Combined with the at most $k$ indices in $\hat{I}^C$, we can bound $\hat{\rho}_{\leq t} \setminus \rho_{\leq t}$ to have size $(2k + d)$.  
\end{proof}

Furthermore, the difference set can be maintained efficiently.

\begin{lemma}
    \label{lemma:d-k-sym-diff-update}
    For all $t \in [T]$, the symmetric difference $D_{t-1}$ can be updated to $D_t$ in $\tO{1}$ time.
    The update bound can be made $O(1)$ expected time. %
\end{lemma}

\begin{proof}
    Let $\rho_t$ be the actual update and $\hat{\rho}_t$ be the predicted update.
    If $\rho_t = \hat{\rho}_t$, then $D_t$ requires no modification, so we may assume $\rho_t \neq \hat{\rho}_t$.
    Clearly $D_{t - 1}$ and $D_t$ differ by at most two elements (specifically $\rho_t, \hat{\rho}_t$) and we can easily update $D_t$ in $\tO{1}$ time. %
\end{proof}

\subsection{Two Algorithms for the \texorpdfstring{$\striangleF$}{s-triangle} Problem}
\label{sec:striangle-delay-alg}

Our first example of algorithms with predictions is the $\striangleF$ problem.
In the $\striangleF$ problem, the algorithm is required to maintain the number of triangles that contain a fixed vertex $s$.
In Section \ref{sec:striangle-lb-delay}, we showed conditional lower bounds for the $\striangleF$ problem given bounded delay predictions.
We now prove that our lower bound is almost optimal.

As a warm-up, we recall how to solve the $\striangleF$ problem in the standard online setting.

\subsubsection*{Warmup: \texorpdfstring{$\striangleF$}{} Algorithm Without Predictions}

For completeness, we start by revisiting a well-known algorithm for the $\striangleF$ problem.
In our algorithms with bounded delay predictions for the $\striangleF$ problem (Theorem \ref{thm:striangle-delay-alg}), we use the online algorithm (without predictions) as a sub-routine in the preprocessing phase to compute the counts under the predicted update sequence.

\begin{theorem}
\label{thm:striangle-inter-alg}
There is an algorithm which solves the $\striangleF$ problem with $O(n^2)$ preprocessing time, $O(n)$ time per update, and $O(1)$ time per query.

There is also an algorithm which solves the $\striangleF$ problem with $O(n^2)$ preprocessing time, $O(1)$ time per update, and $O(n^2)$ time per query.
\end{theorem}

\begin{proof}
First, for any graph $G$, we can count the number of $\striangleF$'s in the graph $G$ in $O(n^2)$ time by considering all $n^2$ triples of vertices $(s, a, b)$ for $a, b \in V$.
Note that this gives automatically the second update-optimized algorithm.
In the pre-processing step, we use $O(n^2)$ time to save the adjacency matrix of the initial graph.
At each update, we update the adjacency matrix in constant time, and compute $\striangleF$ from scratch in $O(n^2)$ time at each query.

We now describe the first query-optimized algorithm.
In the preprocessing step, we save the adjacency structure of the initial graph $G$ and compute the initial number of $\striangleF$'s in $O(n^2)$ time.
Store this count in a variable $c$.
Consider now an update.

{\bf Case 1: Any update of an edge not adjacent to $s$}

Let the update be to the edge $(u,v)$ (either inserting or deleting the edge). 
Note that edge $(u, v)$ can only participate in the $\striangleF$ of three vertices $(s, u, v)$.
Therefore, we need to increment (resp. decrement) $c$ if and only if $(s, u), (s, v)$ are both edges in the graph $G$.
This requires $O(1)$ time.

{\bf Case 2: Any update of an edge adjacent to $s$} 

Let the update be to the edge $(s, u)$. 
$(s, u)$ participates in the at most $n - 2$ triangles $(s, u, v)$ where $(s, v), (u, v) \in E$.
We can count the number of such $v$, denoted $c_u$ in $O(n)$ time.
It suffices to increment $c$ by $c_u$ if $(s, u)$ is an insertion, and decrement $c$ by $c_u$ if $(s, u)$ is a deletion.

Since the adjacency structure can be updated in $O(1)$ time, any update can be processed in $O(n)$ time while any query can be answered in $O(1)$ time by returning $c$.
\end{proof}

Next, we explore the $\striangleF$ problem under our various prediction models, presenting tight upper bounds.

\subsubsection*{An Algorithm for Bounded Delay Predictions with Outliers}

We begin with an upper bound for the $\striangleF$ problem taking advantage of predictions with bounded delay and outliers.
This immediately implies an algorithm given predictions with bounded delay.

\begin{theorem}
    \label{thm:striangle-delay-alg}
    Let $T$ be polynomial in $n$.
    Let $\hat{\rho} = (\hat{\rho}_1, \dotsc, \hat{\rho}_T)$ be a $d$-delayed prediction with $k$ outliers for the $\striangleF$ problem (Definition \ref{def:delay-predictions-outliers}). 
    
    There is an algorithm solving the $\striangleF$ problem given $\hat{\rho}$ with polynomial preprocessing time $P(n)$, update time $U(n) = O(d + k)$, and query time $Q(n) = O(1)$.

    There is an algorithm solving the $\striangleF$ problem given $\hat{\rho}$ with polynomial preprocessing time $P(n)$, update time $U(n) = O(1)$, and query time $Q(n) = O((d + k)^2)$.
\end{theorem}

Observe that if the number of outliers is small (i.e. $k = O(d)$), then we have a query-optimized algorithm that has $U(n) = O(d)$ and $Q(n) = O(1)$ and a update-optimized algorithm that has $U(n) = O(1)$ and $Q(n) = O(d^2)$.
We now introduce a definition that will be useful for the $\striangleF$ problem.

\begin{definition}
    \label{def:striangle-sensitivity}
    Let $G$ be a dynamic graph with update sequence $\rho$.
    For a given vertex $u \in V$ and timestep $\tau$, the %
    {\bf sensitivity of vertex $u$ at time $\tau$ with respect to sequence $\rho$}, is 
    
    \[ s_{\striangleF}(u, \tau, \rho) := \Big|\set{~v \in V \setminus \set{s, u} ~:~ \set{u,v},\set{v,s} \in E_\tau(\rho) ~}\Big| ~. \]
    
    That is, if the edge $\set{s, u}$ is flipped in the $\tau$-th update, then the number of $\striangleF$'s in graph $G$ changes by the value $s_{\striangleF}(u, \tau, \rho)$.
    For convenience, the sensitivity may be denoted $s(u, \tau, \rho)$ or $s(u, \tau)$ when the update sequence $\rho$ is clear.
\end{definition}

\subsubsection*{The Query-Optimized Algorithm}
In the query-optimized algorithm, we maintain the following data structures:

\begin{enumerate}
    \item $c$ the count of $\striangleF$ in the current graph $G$.
    We answer each query in $O(1)$ time by returning $c$.

    \item The current graph $G = (V, E_t)$ where $E_t = E_t(\rho)$.

    \item The predicted graph $\hat{G} = (V, E_t(\hat{\rho}))$.
    
    \item $D_t = \rho_{\leq t} \Delta \hat{\rho}_{\leq t} = E_t(\rho) \Delta E_t(\hat{\rho})$ the symmetric difference of the predicted and actual update sequence (including multiplicity), or alternatively the symmetric difference of the edge sets $E_t(\rho), E_t(\hat{\rho})$.
    Observe that the second equality follows from the fact that the initial graph $G_0$ is known. 
    
    \item $V_{D_t} = \set{v \in V \given \textrm{$v$ is incident to some edge in $D_t$}}$ the set of vertices containing all endpoints of edges in $D_t$.
    
    \item The {\bf predicted sensitivities} $s(v, t, \hat{\rho})$ for all vertices $v$ and time steps $t$.
    For each vertex $v$, this is stored as a sequence of tuples $S(v) = (t, m)$ where for each $t$ with $s(v, t, \hat{\rho}) \neq s(v, t-1, \hat{\rho})$, there is an entry $(t, s(v, t, \hat{\rho}))$.
    The list $S(v)$ is stored in increasing order of $t$.
    For a given $v \in V, t \in [T]$, we can easily compute $s(v, t, \hat{\rho})$ by searching for the largest key $t_0$ such that $t_0 \leq t$ in $S(v)$ in $\tO{1}$ time.
\end{enumerate}

First, we bound the size of the above data structures.

\begin{lemma}
    \label{lemma:striangle-sym-diff-size}
    At any time step $t$, $|D_t|, |V_{D_t}| = O(d + k)$.
    The predicted sensitivities require space~$O(n T)$. 
\end{lemma}

\begin{proof}
    The bound on $D_t$ follows as a consequence of Lemma \ref{lemma:d-k-sym-diff-size}.
    Since each edge has two endpoints, we can bound the size of $V_{D_t}$ by $O(d + k)$.

    To bound the size of each sequence of predicted sensitivities, it suffices to note that there are $n$ vertices each with at most $T$ distinct time steps in which the sensitivity may change.
\end{proof}

Next, we claim that we can efficiently update each data structure.

\begin{lemma}
    \label{lemma:striangle-sym-diff-update}
    Let $D_t$ and $V_{D_t}$ denote the data structures $D, V_D$ at time step $t$ respectively.
    Then, $D_t, V_{D_t}$ can be updated from $D_{t - 1}, V_{D_{t - 1}}$ in $\tO{1}$ time.
\end{lemma}

\begin{proof}
    Let $e$ be the given update and $\hat{e}$ be the predicted update.
    If $e = \hat{e}$, then $D$ (and therefore $V_D$) require no modification, so we may assume $e \neq \hat{e}$.
    Clearly $D_{t}, D_{t - 1}$ differ by at most 2 elements (specifically $e, \hat{e}$) and $V_{D_t}, V_{D_{t - 1}}$ differ by at most 4 elements (their endpoints).
    We can easily update these data structures in $\tO{1}$ time.
\end{proof}

We are now ready to present the query-optimized algorithm.

\begin{algorithm}[H]

\SetKwInOut{Input}{Input}
\Input{Initial graph $G_0$ and $d$ delayed predictions $\hat{\rho}$ with $k$ outliers}

\BlankLine

Compute initial sensitivities $s(u, 0) \gets \#\set{v \notin \set{s, u} \given (u, v), (v, s) \in E_0}$ for all $u \neq s$

Initialize $S(v) \gets \set{(0, s(v, 0))}$ for all $v \in V$ and $v \neq s$.

Compute $c \gets \striangleF(G_0)$

Initialize $D \gets \emptyset$, $V_D \gets \emptyset$

\For{$t = 1$ to $T$}{
    \If{$\hat{\rho}_t = (s, u)$}{
        \For{$v \neq u$ and $(u, v) \in E_t(\hat{\rho})$}{
            $S(v) \gets S(v) \circ (t, s(v, t - 1, \hat{\rho}) + e)$ where $e = \begin{cases}
                +1 & \textrm{if $(s, u)$ inserted} \\
                -1 & \otherwise
            \end{cases}$
        }
    }
    \If{$\hat{\rho}_t = (u, v)$}{
        \If{$(s, v) \in E_t(\hat{\rho})$}{
            $S(u) \gets S(u) \circ (t, s(u, t - 1, \hat{\rho}) + e)$ where $e = \begin{cases}
                +1 & \textrm{if $(u, v)$ inserted} \\
                -1 & \otherwise
            \end{cases}$
        }
        \If{$(s, u) \in E_t(\hat{\rho})$}{
            $S(v) \gets S(v) \circ (t, s(v, t - 1, \hat{\rho}) + e)$ where $e = \begin{cases}
                +1 & \textrm{if $(u, v)$ inserted} \\
                -1 & \otherwise
            \end{cases}$
        }
    }
}

\caption{$\mathbf{UPreprocess\striangleF}(G_0, \hat{\rho})$}
\label{alg:update-striangle-preprocess}
\end{algorithm}

\begin{algorithm}[H]

\SetKwInOut{Input}{Input}
\Input{Current graph $G_{t - 1}$, current update $(u_t, v_t)$, $d$-delayed predictions $\hat{\rho}$ with $k$ outliers, and update history $\rho_{<t}$.}

\BlankLine

Update $G$ (resp. $\hat{G}$) by flipping edge $(u_t, v_t)$ (resp. $\hat{\rho}_t$)

Update $D \gets \rho_{\leq t} \Delta \hat{\rho}_{\leq t}$ and $V_D$ according to Lemma \ref{lemma:striangle-sym-diff-update}

\If{$s \notin (u_t, v_t)$}{
    \Return $c \gets c + e$ where $e = \begin{cases}
        +1 & (s, u_t), (s, v_t), (u_t, v_t) \in E_{t} \\
        -1 & (s, u_t), (s, v_t) \in E_{t}, (u_t, v_t) \notin E_{t} \\
        0 & \otherwise
    \end{cases}$ 
}
    
\If{$\set{s, x} = \set{u_t, v_t}$}{
    $\sensediff(x) \gets 0$
    
    \For{$v \in V_D \setminus \set{s, x}$}{
        $\sensediff(x) \gets \sensediff(x) + e$ where $e = \begin{cases}
            +1 & (x, v), (s, v) \in E_t \andT (x, v), (s, v) \notin E_t(\hat{\rho}) \\
            -1 & (x, v), (s, v) \in E_t(\hat{\rho}) \andT (x, v), (s, v) \notin E_t \\
            0 &\otherwise
        \end{cases}$
    }
    
    \Return $c \gets c + b (s(x, t, \hat{\rho}) + \sensediff(x))$ where $b = +1$ if $(s, x)$ inserted and $-1$ otherwise.
}

\caption{$\mathbf{UUpdate\striangleF}(G_{t - 1}, (u_t, v_t), \hat{\rho}, \rho_{\leq t})$}
\label{alg:update-striangle-update}
\end{algorithm}

Given a query, we return $c$ in $O(1)$ time.

We now provide the proof for the query-optimized algorithm.

\paragraph{Update Correctness}

\begin{proof}
    In each query, we return $c$ in constant time.
    It therefore suffices to show that after each request, $c$ contains the correct value $\striangleF(G_t)$.
    We then proceed by induction on $t$.
    Clearly in the base case $c = \striangleF(G_0)$ after Algorithm \ref{alg:update-striangle-preprocess}.

    We now prove the inductive case.
    As in Theorem \ref{thm:striangle-inter-alg}, if the update edge $(u_t, v_t)$ does not include $s$, we increment $c$ if the triangle $(s, u_t, v_t)$ is added, decrement $c$ if $(s, u_t, v_t)$ is deleted, and leave $c$ unchanged otherwise.
    If the edge update is of the form $(s, x)$, it suffices to show $\sensediff(x) = s(x, t, \rho) - s(x, t, \hat{\rho})$.
    Recall that $s(x, t, \rho)$ is the number of vertices $v \notin \set{s, x}$ such that $(s, v), (v, x) \in E_t(\rho)$.
    Thus, unless $v \in V_D$, the 2 edge path $(s, v, x)$ cannot exist in only one of $E_t(\rho)$ and $E_t(\hat{\rho})$.
    Therefore, it suffices to check only vertices $v \in V_D \setminus \set{s, x}$.
    For each such vertex, we increment (resp. decrement) $\sensediff(x)$ if the path $(s, v, x)$ exists in $E_t(\rho)$ but not $E_t(\hat{\rho})$ (resp. $E_t(\hat{\rho})$ but not $E_t(\rho)$).
\end{proof}

\paragraph{Update Time}

\begin{proof}
    The preprocessing algorithm computes $S(v)$ for all $v \in V$ in $O(nT)$ time.
    Initializing $c$ requires $O(n^2)$ time (Theorem \ref{thm:striangle-inter-alg}) and initializing $D, V_D$ requires $O(1)$ time.
    Overall, this is $O(nT + n^2) = \poly(n)$ time.

    The update algorithm updates $G, \hat{G}, D, V_D$ in $\tO{1}$ time.
    If $s \notin (u_t, v_t)$, $c$ is updated and returned in $O(1)$ time.
    Otherwise, $\sensediff(x)$ is computed in $O(d + k)$ time 
    and $c$ is again updated and returned in $O(1)$ time.

    The query algorithm of returning $c$ requires $O(1)$ time.
\end{proof}

In the update-optimized algorithm, we maintain the following data structures:

\begin{enumerate}
    \item $\set{\hat{c}_t}_{t \in [T]}$ where $\hat{c}_t = \striangleF(V, E_t(\hat{\rho}))$ where $\hat{c}_t$ is the number of s-triangles in the predicted graph $E_t(\hat{\rho})$.
    
    \item The current graph $G = (V, E_t)$ where $E_t = E_t(\rho)$.

    \item The predicted graph $\hat{G} = (V, E_t(\hat{\rho}))$.

    \item The predicted sensitivities $s(v, t, \hat{\rho})$ for all vertices $v$ and time steps $t$.
    We store this in the data structure $S(v)$ as in the query-optimized algorithm.
\end{enumerate}

\begin{algorithm}[H]

\SetKwInOut{Input}{Input}
\Input{Initial graph $G$, $d$-delayed predictions $\hat{\rho}$ with $k$ outliers.}

\BlankLine

Compute initial sensitivities $s(u, 0) \gets \#\set{v \notin \set{s, u} \given (u, v), (v, s) \in E_0}$ for all $u \neq s$

Initialize $S(v) \gets \set{(0, s(v, 0))}$ for all $v \in V$ and $v \neq s$.

Compute $\hat{c}_0 \gets \striangleF(G_0)$.
\For{$t = 1$ to $T$}{
    Compute $\hat{c}_t \gets \striangleF(V, E_t(\hat{\rho}))$ using $O(n)$ update algorithm of Theorem \ref{thm:striangle-inter-alg}
    
    \If{$\hat{\rho}_t = (s, u)$}{
        \For{$v \neq u$ and $(u, v) \in E_t(\hat{\rho})$}{
        
            $S(v) \gets S(v) \circ (t, s(v, t - 1, \hat{\rho}) + e)$ where $e = \begin{cases}
                +1 & \textrm{if $(s, u)$ inserted} \\
                -1 & \otherwise
            \end{cases}$ 
        }
    }
    \If{$\hat{\rho}_t = (u, v)$}{
        \If{$(s, v) \in E_t(\hat{\rho})$}{
        
             $S(u) \gets S(u) \circ (t, s(u, t - 1, \hat{\rho}) + e)$ where $e = \begin{cases}
                +1 & \textrm{if $(u, v)$ inserted} \\
                -1 & \otherwise
            \end{cases}$ 
        }
        \If{$(s, u) \in E_t(\hat{\rho})$}{
        
            $S(v) \gets S(v) \circ (t, s(v, t - 1, \hat{\rho}) + e)$ where $e = \begin{cases}
                +1 & \textrm{if $(u, v)$ inserted} \\
                -1 & \otherwise
            \end{cases}$ 
        }
    }
}

\caption{$\mathbf{QPreprocessing\striangleF}(G_0, \hat{\rho})$}
\label{alg:query-striangle-preprocess}
\end{algorithm}

Given an update, we update $G \gets (V, E_t)$ and $\hat{G} \gets (V, E_t(\hat{\rho}))$ in $O(1)$ time. 

\begin{algorithm}[H]

\SetKwInOut{Input}{Input}\SetKwInOut{Output}{Output}
\Input{Current graph $G_{t}$, $d$-delayed predictions $\hat{\rho}_t$ with $k$ outliers.}
\Output{$c = \striangleF(G_{t})$}

\BlankLine

Initialize $\cdiff \gets 0$

Let $V_D \setminus \set{s} = \set{v_1, v_2, \dotsc, v_m}$ for $m = O(d + k)$

\For{$v_i$ with $1 \leq i \leq m$}{
    Initialize $\sensediff(v_i) \gets 0$
    \label{line:q-query:init-sensitivity}
    
    \For{$v_j$ with $i \neq j$}{
        $\cdiff \gets \cdiff + e$ where $e = \begin{cases}
            +\frac{1}{2} & (s, v_j), (v_j, v_i), (v_i, s) \in E_t \andT (s, v_j), (v_j, v_i), (v_i, s) \notin E_t(\hat{\rho}) \\
            -\frac{1}{2} & (s, v_j), (v_j, v_i), (v_i, s) \in E_t(\hat{\rho}) \andT (s, v_j), (v_j, v_i), (v_i, s) \notin E_t \\
            0 & \otherwise
        \end{cases}$
        \label{line:q-query:d-triangle}
        
        $\sensediff(v_i) \gets \sensediff(v_i) + e$ where $e = \begin{cases}
            -1 & (s, v_j), (v_j, v_i) \in E_t(\hat{\rho}) \\
            0 & \otherwise
        \end{cases}$
        \label{line:q-query:lower-sensitivity}
    }
    
    $\cdiff \gets \cdiff + b (s(v_i, t, \hat{\rho}) + \sensediff(v_i))$ where $b = \begin{cases}
        +1 & (s, v_i) \in E_t \andT (s, v_i) \notin E_t(\hat{\rho}) \\
        -1 & (s, v_i) \in E_t(\hat{\rho}) \andT (s, v_i) \notin E_t \\
        0 & \otherwise
    \end{cases}$
    \label{line:q-query:adjust-sensitivity}
}

\Return $c \gets \hat{c}_t + \cdiff$

\caption{$\mathbf{QQuery\striangleF}(G_{t}, \hat{\rho}, \rho_{\leq t})$}
\label{alg:query-striangle-query}
\end{algorithm}

\subsubsection*{The Update-Optimized Algorithm}
We now give the proof of the second algorithm, completing the proof of Theorem \ref{thm:striangle-delay-alg}.

\paragraph{Query Correctness}

Our algorithm counts all the $\striangleI$'s the appear in exactly one of $E_t, E_t(\hat{\rho})$.
Consider an $\striangleI$ with vertex set $\set{s, a, b}$.
If $a, b \notin V_D$, then $\set{s, a, b}$ is a triangle in $E_t$ if and only if it is a triangle in $E_t(\hat{\rho})$.
In Line \ref{line:q-query:d-triangle}, we count all triangles with $a, b \in V_D$.
In Line \ref{line:q-query:adjust-sensitivity}, we count all triangles with $a \in V_D, b \notin V_D$, where we adjust the pre-computed sensitivity in Line \ref{line:q-query:lower-sensitivity} to avoid double counting.
We first prove an intermediate lemma showing that Line \ref{line:q-query:lower-sensitivity} correctly removes the $\striangleI$'s with only one (non-$s$) vertex in $V_D$.

\begin{lemma}
    \label{lemma:q-striangle:sensitivity-lemma}
    In Line \ref{line:q-query:adjust-sensitivity}, 
    \begin{equation*}
        s(v_i, t, \hat{\rho}) - \sensediff(v_i)
    \end{equation*}
    is the number of vertices $u \notin V_D \cup \set{s, v_i}$ such that the two edge path $(s, u, v_i) \in E_t$.
    This is also the number of vertices $u \notin V_D$ such that the two edge path $(s, u, v_i)$ is in $E_t(\hat{\rho})$.
\end{lemma}

\begin{proof}
    First, since $u \notin V_D$, the edge set incident to $u$ is identical in $E_t, E_t(\hat{\rho})$.
    Thus, the number of vertices $u \notin V_D$ such that $(s, u, v_i)$ is a path is the same regardless of the choice of edge set $E_t, E_t(\hat{\rho})$.

    Recall that $s(v_i, t, \hat{\rho})$ is the number of vertices $u \notin \set{s, v_i}$ such that the path $(s, u, v_i)$ is in $E_t(\hat{\rho})$.
    In Line \ref{line:q-query:adjust-sensitivity}, $\sensediff(v_i) = -m$ where $m$ is the number of vertices $v_j \in V_D$ such that $(s, v_j, v_i) \in E_t(\hat{\rho})$.
    In particular, $s(v_i, t, \hat{\rho}) - \sensediff(v_i)$ is exactly the number of vertices $u \notin V_D \cup \set{s, v_i}$ such that $(s, u, v_i)$ is in $E_t(\hat{\rho})$, proving the claim.
\end{proof}

We now prove the correctness of the algorithm.

\begin{proof}
    It suffices to show $\cdiff = \striangleF(E_t) - \striangleF(E_t(\hat{\rho}))$.
    By analyzing Algorithm \ref{alg:query-striangle-query}, we see that in Line \ref{line:q-query:d-triangle}, $\cdiff$ accumulates $+1$ for every triangle $(s, a, b)$ in $E_t$ and not in $E_t(\hat{\rho})$ and $-1$ for every triangle in $E_t(\hat{\rho})$ and not $E_t$ if $a, b \in V_D$.
    Since we iterate over both orderings of pairs $(i, j)$, we accumulate $\frac{1}{2}$ in each iteration.

    In Line \ref{line:q-query:adjust-sensitivity} $\cdiff$ increases (resp. decreases) by $(s(v_i, t, \hat{\rho}) - \sensediff(v_i))$ if $(s, v_i) \in E_t \setminus E_t(\hat{\rho})$ (resp. $(s, v_i) \in E_t(\hat{\rho}) \setminus E_t$).
    By Lemma \ref{lemma:q-striangle:sensitivity-lemma}, this is the number of two edge paths $(s, u, v_i)$ in $E_t$ and $E_t(\hat{\rho})$ with $u \notin V_D$.
    Therefore, if $(s, v_i) \in E_t \setminus E_t(\hat{\rho})$ (resp. $(s, v_i) \in E_t(\hat{\rho}) \setminus E_t$), this is precisely the number of $\striangleI$'s with vertex set $\set{s, u, v_i}$ with $u \notin V_D$ in $E_t$ but not $E_t(\hat{\rho})$ (resp. $E_t(\hat{\rho})$ but not $E_t$).
    Otherwise, if $(s, v_i)$ is in both or none of $E_t, E_t(\hat{\rho})$, the number of $\striangleI$'s with vertex set $\set{s, u, v_i}$ with $u \notin V_D$ is the same in both $E_t, E_t(\hat{\rho})$.
    
    Finally, if $a, b \notin V_D$, then $\set{s, a, b}$ is a triangle in $E_t$ if and only if it is a triangle $E_t(\hat{\rho})$ so that the number of $\striangleI$'s in $E_t, E_t(\hat{\rho})$ with vertex set $\set{s, a, b}$ is identical when $a, b \notin V_D$.

    We have thus shown that $\cdiff = \striangleF(E_t) - \striangleF(E_t(\hat{\rho}))$, concluding the proof.
\end{proof}

\paragraph{Query Time}
\begin{proof}
    Algorithm \ref{alg:query-striangle-preprocess} requires $O(n^2 + nT)$ time, following similar arguments to Theorem \ref{thm:striangle-inter-alg} and the analysis of Algorithm \ref{alg:update-striangle-preprocess}.
    
    We now examine Algorithm \ref{alg:query-striangle-query}.
    Within each loop, updating $\cdiff, \sensediff$ requires $O(1)$ time by examining the adjacency structures of $G, \hat{G}$.
    From Lemma \ref{lemma:striangle-sym-diff-size}, we have $|V_D| = O(d + k)$ so that Algorithm \ref{alg:query-striangle-query} requires time $O((d + k)^2)$.
\end{proof}

We have shown that given $d$ delayed predictions with $k$ outliers for $d + k = O(n^{1 - \eps})$, it is possible to design an algorithm that beats the conditional lower bounds for algorithms without predictions.
However, since the OMv Conjecture holds against any algorithm with polynomial preprocessing time, we can therefore conclude it is hard to make good predictions for the $\striangleF$ problem.
Using Theorem \ref{thm:striangle-delay-alg} we claim that no polynomial time prediction algorithm can yield truly sub-linear delay.

\begin{proposition}
    Under the OuMv conjecture, no polynomial time algorithm can output $\hat{\rho}$ that is a $d$ delayed prediction with $k$ outliers for the $\striangleF$ problem if $d + k = O(n^{1 - \eps})$ for $\eps > 0$.
\end{proposition}

\begin{proof}
    Suppose such a polynomial time algorithm exists.
    Then, we run this algorithm in the preprocessing step.
    Using Theorem \ref{thm:striangle-delay-alg}, we have an algorithm with polynomial preprocessing time, update time $U(n) = O(d + k)$, and query time $Q(n) = O(1)$
    , contradicting the OuMv conjecture.
\end{proof}

\subsection{Subgraph Connectivity}

Recall that in subgraph connectivity problem, the algorithm is given a fixed graph $G$ and a subset of vertices $S$.
Each update adds or removes a vertex from the set $S$ and each query asks for some pair of vertices $u, v \in S$ (in variants of the problem one of both of $u, v$ can be a fixed vertex) whether $u, v$ are connected on the subgraph $G[S]$ induced by $S$.
We give an upper bound for the subgraph connectivity problem using predictions with bounded delay.

\paragraph*{From Delay Aware to Delay Agnostic Algorithms.}

Generally, we do not expect an algorithm to be aware of the quality that the prediction will have for the online request sequence.
Indeed, consistency and robustness should both hold without the algorithm being given the quality of its prediction in addition to the prediction itself.
However, in the setting of dynamic algorithms Lemmas~\ref{lemma:d-k-sym-diff-size} and \ref{lemma:d-k-sym-diff-update} imply that with no extra cost, an algorithm can compute (up to a constant factor) the quality of the predictions it has seen so far.
Using this observation, we can design algorithms that are not only given a prediction $\hat{\rho}$, but a guarantee $d$ that $\hat{\rho}$ is $d$ delayed.
Then, taking in the current guess for the delay parameter, we can choose an appropriate parameter for the delay-aware algorithm to handle the request of the current time step.
We exhibit this transformation for the subgraph connectivity problem.

\subsubsection*{Warmup: Update-Optimized Algorithm with parameter $d$ known}

We begin with an algorithm that is not only given a $d$ delayed prediction $\hat{\rho}$, but $d$ such that $\hat{\rho}$ is guaranteed to be at most $d$ delayed.
In Section \ref{sec:d-aware-to-d-agnostic-subconn}, we give a transformation to a $d$ agnostic algorithm.

\begin{lemma}
    \label{lemma:subgraph-connectivty-delay-promise-alg}
    Let $T$ by polynomial in $n$.
    Let $\hat{\rho} = (\hat{\rho}_1, \hat{\rho}_2, \dotsc, \hat{\rho}_T)$ be a $d$-delayed prediction for the subgraph connectivity problem.
    
    There is an algorithm solving the subgraph connectivity problem given $d$ and $\hat{\rho}$ with polynomial preprocessing time $P(n)$, update time $U(n) = \tO{1}$ and query time $Q(n) = O(d^2)$.
\end{lemma}

We show that there exists an algorithm with polynomial preprocessing time that allows the update and query time to overcome the conditional lower bound imposed on purely dynamic algorithms.

We begin with some useful definitions.
Let $T$ be an integer and $\rho = (\rho_1, \rho_2, \dotsc, \rho_T)$ be a sequence of requests on some graph $G$ with initial vertex set $S_0$.
For any time $t$, let $S_t$ denote the set of vertices in $S$ after the first $t$ requests.
Let $1 \leq t_1 < t_2 \leq T$ be two time steps.
Consider the interval $[t_1, t_2] \subset [T]$.
A node $v$ is {\bf permanent in $\rho$ from $t_1$ to $t_2$} if $v \in S_t$ for all $t \in [t_1, t_2]$ and $v$ is not part of any query between time steps $t_1$ and $t_2$.
A node $v$ is {\bf active in $\rho$ from $t_1$ to $t_2$} if any request $\rho_t$ changes the membership of $v$ in $S_t$ or queries $v$ for $t_1 \leq t \leq t_2$.
Note that for any $t \in [t_1, t_2]$, every vertex in $S_t$ (which includes every queried vertex) must either be permanent or active.

We will also require the dynamic connectivity algorithm of Kapron, King, and Mountjoy, computing all pairs connectivity with worst case polylogarithmic update and query time.

\begin{theorem}[\cite{DBLP:conf/soda/KapronKM13}]
    \label{thm:kkm-connectivity}
    There is an algorithm supporting the following operations in $\tO{1}$ time:
    \begin{enumerate}
        \item $\mathbf{UPDATE}(u, v)$: Insert or remove edge $(u, v)$
        \item $\mathbf{QUERY}(a, b)$: Answer if $a, b$ are connected in the graph $G$
    \end{enumerate}
    The algorithm is always correct if $a, b$ are connected and correct with high probability when $a, b$ are not connected.
\end{theorem}

Note that applying the above algorithm $n$ times given a vertex insertion/deletion gives a conditionally optimal algorithm for the pure online case, as discussed in \cite{HenzingerKNS15}.

In both algorithms, we maintain the following data structures.

\begin{enumerate}
    \item $S_t$, the current set of vertices in $S$.

    \item For all $t \in T$ and $u, v \in A_t$, let $C(u, v, t)$ denote whether $u, v$ are connected in the graph $\hat{H}(u, v, t) = G[P_t \cup \set{u, v}]$, the subgraph induced by $P_t \cup \set{u, v}$, where $P_t$ is the set of permanent vertices in $\hat{\rho}$ from $t - d$ to $t + d$ and $A_t$ is the set of active vertices in $\hat{\rho}$ from $t - d$ to $t + d$.
    This is precomputed in the preprocessing phase and only queried in the dynamic phase.
    
    \item $Q_t = S_t \setminus P_t$, the set of vertices in $S_t$ that are not permanent in $\hat{\rho}$ from $t - d$ to $t + d$.
    To maintain $Q_t$ efficiently, note that $P_t$ can change by at most 1 vertex with each update and 2 vertices with each query.
    In particular, a vertex is added to $P_t$ when it is inserted into $S$ at time step $t - d$ and will remain in $S$ without being queried until time step $t + d$, while a vertex is removed from $P_t$ only when it is removed from $S$ or queried at time step $t + d$.
    We can therefore maintain $P_t$ as a sequence of insertions and deletions and maintain $Q_t$ in $O(1)$ time.
\end{enumerate}

The set $Q_t$ represents the possibly unexpected vertices in the set $S_t$ due to the error of the prediction.
We claim that the size of this set of vertices depends on the quality of the prediction.

\begin{lemma}
    \label{lemma:subconn-error-set-size}
    Suppose $\hat{\rho}$ is a $d$-delayed prediction.
    For $t \in [T]$, let $P_t$ be the set of permanent vertices in $\hat{\rho}$ from $t - d$ to $t + d$.
    Let $S_t$ be the set of vertices in $S$ after the first $t$ true request in the sequence $\rho$.
    Then, $S_t = P_t \cup Q_t$ where $|Q_t| \leq 4d + 2$.

    Furthermore, in each time step $S_t, Q_t, P_t$ each change by at most three vertices.
\end{lemma}

\begin{proof}
    First, we claim $P_t \subset S_t$.
    Let $v \in P_t$.
    Consider the true update sequence $(\rho_1, \dotsc, \rho_t)$.
    Since $\hat{\rho}$ is $d$-delayed, we have 
    \begin{equation*}
        (\hat{\rho}_1, \dotsc, \hat{\rho}_{t - d}) ~\subset~ (\rho_1, \dotsc, \rho_t) ~\subset~ (\hat{\rho}_1, \dotsc, \hat{\rho}_{t + d})~.
    \end{equation*}
    Since $v \in P_t$, $v \in S$ after $\hat{\rho}_{\leq t - d}$ and is not modified in the sub-sequence $\hat{\rho}_{[t - d, t + d]}$, so that $v \in S$ after $\rho_{\leq t}$ and $v \in S_t$.
    
    Next, we show that $|Q_t| \leq 5d$.
    Consider a vertex $w \in S_t$.
    If $w \notin P_t$, then there is some update in $\hat{\rho}_{[t - d, t + d]}$ that inserts $w$ into $S_t$, or some query involving $w$ in $\hat{\rho}_{[t - d, t + d}]$.
    Note that there are at most $2d + 1$ such requests, and each can include at most 2 vertices, so that $|Q_t| \leq 4d + 2$.

    By definition, $S_t$ changes by one vertex in each time step.
    Above we have argued that $P_t$ changes by at most 2 vertices in each time step.
    Combining, $Q_t$ changes by at most 3 vertices in each time step, and so can be maintained in $\tO{1}$ time.
\end{proof}

In the update-optimized algorithm, given a query pair $(u, v)$, we will query for connectivity in the graph with vertex set $Q_t \cup \set{u, v}$, with edges encoding pairs of vertices with a path between them.
Since the size of this graph is small, we are able to compute connectivity queries efficiently.

\begin{lemma}
    \label{lemma:subconn-equivalence}
    Let $t \in [T]$ be a time step.
    Let $u, v$ be two vertices in $S_t = P_t \cup Q_t$.
    Let $H_t$ denote the graph with vertex set $Q_t \cup \set{u, v}$ and $(a, b) \in E(H_t)$ if either $(a, b) \in E(G_t)$ is an edge in the original graph, or $C(a, b, t) = 1$, that is there is a path from $a, b$ with all internal vertices in $P_t$.
    Then, $u, v$ are connected in $G[S_t]$ if and only if $u, v$ are connected in $H_t$.
\end{lemma}

\begin{proof}
    Suppose $u, v$ are connected in $H_t$.
    It suffices to show that every edge in $H_t$ has endpoints that are connected in $G[S_t]$.
    An edge in $H_t$ is either an edge in the original graph or a path with all internal vertices in $P_t$.
    In the former case, note that every vertex in $H_t$ is in $S_t$, so that the edge is present in $G[S_t]$.
    In the latter case, again note that all vertices in the path are in $S_t$, and every edge is in the original graph.

    Suppose $u, v$ are connected in $G[S_t]$.
    Consider one path between $u, v$ and consider the subsequence of vertices on this path consisting of vertices in $Q_t \cup \set{u, v}$.
    This path exists in $H_t$ as any vertex not in this subsequence is necessarily in $P_t$, a case which is covered by the additional edges inserted.
\end{proof}

Therefore, to check connectivity in $G[S_t]$, it suffices to check connectivity in the smaller graph $H_t$.
We now present our algorithm with $\tO{1}$ update time and $O(d^2)$ query time.

\begin{algorithm}[H]

\SetKwInOut{Input}{Input}
\Input{Fixed graph $G$ with initial vertex set $S_0$.
$d$-delayed predictions $\hat{\rho}$.}

\BlankLine

\For{$t = 1$ to $T$}{
    Compute $P_{t} = \set{v \given v \textrm{ is permanent in $\hat{\rho}$ from $\max(1, t - d)$ to $\min(T, t + d)$}}$

    Compute $A_{t} = \set{v \given v \textrm{ is active in $\hat{\rho}$ from $\max(1, t - d)$ to $\min(T, t + d)$}}$

    Let $\hat{H}_{t}$ be the subgraph $G[P_{t}]$ induced by $P_{t}$

    \For{$u, v \in A_{t}$}{
        Run DFS on $\hat{H}(u, v, t) \gets G[P_t \cup \set{u, v}]$ and store in $C(u, v, t)$ a 1 if $u, v$ are connected in $\hat{H}(u, v, t)$ and 0 otherwise.
    }
}

\Return $\set{P_t}_{t}, \set{A_t}_{t}, \set{C(u, v, t)}_{u, v, t} $

\caption{$\mathbf{PromisePreprocessingSubConn}(G, S_0, \hat{\rho}, d)$}
\label{alg:promise-subconn-preprocess}
\end{algorithm}

We preprocess the input as described by Algorithm \ref{alg:promise-subconn-preprocess}.
Given an update, we maintain the sets $S_t, Q_t$ in $\tO{1}$ time.
To maintain $Q_t$ given $S_t$, note that the sets of permanent vertices $\set{P_t}$ can be maintained as a sequence of insertions and deletions, and we can process these insertions and deletions in $\tO{1}$ time.
In the query step, we will construct the graph $H_t$ as described in Lemma \ref{lemma:subconn-equivalence} and compute connectivity in this graph.

\begin{algorithm}[H]

\SetKwInOut{Input}{Input}\SetKwInOut{Output}{Output}
\Input{Fixed graph $G$ with current vertex set $S_t$,
$d$-delayed predictions $\hat{\rho}$.}
\Output{YES if $u, v$ are connected in $G[S_t]$, NO otherwise}

\BlankLine

Construct $H_t$ with vertex set $Q_t \cup \set{u, v}$ and edge set $G[Q_t \cup \set{u, v}]$ and additional edges $(a, b)$ where $C(a, b, t) = 1$.

Run DFS from $u$ in $H_t$ and return YES $u, v$ are connected and NO otherwise.

\caption{$\mathbf{PromiseQuerySubConn}(G, S_t, \hat{\rho}, \rho_{\leq t}, d)$}
\label{alg:promise-subconn-query}
\end{algorithm}

Correctness follows immediately from Lemma \ref{lemma:subconn-equivalence}.
We now conclude the proof by analyzing the time complexity of our algorithm.
In the preprocessing step, for each $t$, we can compute $P_t, A_t \in O(d) = O(n)$ time, and construct $\hat{H}_t$ in $O(d^2) = O(n^2)$ time.
Using Theorem \ref{thm:kkm-connectivity}, we can maintain $\hat{H}_t$ with $O(n)$ updates and compute $C(u, v, t)$ with $O(d^2 n)$ updates and $O(d^2) = O(n^2)$ queries.
Overall, this requires $\tO{T n^3}$ time in the preprocessing phase.
For each query, we construct a graph on $O(d)$ vertices and compute a DFS in $O(d^2)$ time.
Since $Q_t \subset A_t$, the additional edges can each be added in constant time.

\subsubsection{Generalization to $d$ Agnostic Algorithm}
\label{sec:d-aware-to-d-agnostic-subconn}

We now design an algorithm that is not given $d$ in the preprocessing phase.

For each $d = 2^{d'}$ where $1 \leq d' \leq \ceil{\log n}$, we maintain the set $S_{t} \setminus P_{t, d}$ where $P_{t, d}$ denotes the set of permanent vertices in $\hat{\rho}$ from $t - d$ to $t + d$.
As before, this set can be maintained in $\tO{1}$ time.
Since there are $O(\log n)$ such sets, we can complete an update in $\tO{1}$ time.

We also need to check if $P_{t, d} \subseteq S_t$.
We claim that this can be maintained for all $d$ in $\tO{1}$ time.
Recall that $P_{t, d}$ changes by at most 2 vertices in every time step (and these vertices must be involved at the $t + d$ request) and $S_t$ changes by at most 1 vertex in every time step.
Then, at each time step, we can simply query for membership in $S_t$ for any vertex added into $P_{t, d}$ and query for membership in $P_{t, d}$ for any vertex removed from $S_t$, requiring only $\tO{1}$ time.
Again, for all $O(\log n)$ values of $d$, this requires $\tO{1}$ time.

For $d$ satisfying $P_{t, d} \subseteq S_t$, let $Q_{t, d} = S_t \setminus P_{t, d}$ and let $d^*$ be the smallest $d^*$ such that $|Q_{t, d}| \leq 4d + 2$.
If $d^*$ does not exist, then $\hat{\rho}$ is not a $d$-delayed prediction for $d \leq n$ and we can afford to answer the query using the full graph $G[S_t]$ in $O(n^2)$ time.
Otherwise, $\hat{\rho}$ is $d^*$ delayed but not $\frac{d^*}{2}$ delayed.
If $\hat{\rho}$ where $\frac{d^*}{2}$ delayed, then $|Q_{t, d^*/2}| \leq 2 d^* + 2$ by Lemma \ref{lemma:subconn-error-set-size}.
Since $\hat{\rho}$ is $d^*$ delayed, we can construct $H_t$ on the vertex set $Q_{t, d^*}$ of size $O(d^*)$ and follow the algorithm of Lemma \ref{lemma:subgraph-connectivty-delay-promise-alg}.
We now state the final theorem, providing the algorithm and proof.

\begin{theorem}
    \label{thm:subgraph-connectivty-delay-alg}
    Let $T$ by polynomial in $n$.
    Let $\hat{\rho} = (\hat{\rho}_1, \hat{\rho}_2, \dotsc, \hat{\rho}_T)$ be a $d$-delayed prediction for the subgraph connectivity problem.
    
    There is an algorithm solving the subgraph connectivity problem given $\hat{\rho}$ with polynomial preprocessing time $P(n)$, update time $U(n) = \tO{1}$ and query time $Q(n) = O(d^2)$.
\end{theorem}

\begin{algorithm}[H]

\SetKwInOut{Input}{Input}
\Input{Fixed graph $G$ with initial vertex set $S_0$.
$d$-delayed predictions $\hat{\rho}$.}

\BlankLine

\For{$d = 2^{d'}$ where $d' = 1$ to $\ceil{\log n}$}{
    $\set{P_{t, d}}, \set{A_{t, d}}, C_d \gets \mathbf{PromisePreprocessingSubConn}(G, S_0, \hat{\rho}, d)$ 
}

\caption{$\mathbf{PreprocessingSubConn}(G, S_0, \hat{\rho})$}
\label{alg:subconn-preprocess}
\end{algorithm}

We preprocess the input as described by Algorithm \ref{alg:subconn-preprocess}.
Given an update, we maintain the sets $S_t, Q_{t, d}$ for all $d$ where $P_{t, d} \subset S_t$ in $O(\log n)$ time, since there are $O(\log n)$ values for $d$ and we can maintain each in constant time.
In the query step, we first compute a valid value of $d$ and construct the graph $H_t$ as before.

\begin{algorithm}[H]

\SetKwInOut{Input}{Input}\SetKwInOut{Output}{Output}
\Input{Fixed graph $G$ with current vertex set $S_t$,
$d$-delayed predictions $\hat{\rho}$.}
\Output{YES if $u, v$ are connected in $G[S_t]$, NO otherwise}

\BlankLine

Let $d^* \gets \min \set{d \given P_{t, d} \subset S_t \andT |Q(t, d)| \leq 4d + 2}$

\If{$d^* = \infty$}{
    Run DFS on $G[S_t]$ and return YES if $u, v$ are connected and NO otherwise.
}\Else{
    \Return $\mathbf{PromiseQuerySubConn}(G, S_t, \hat{\rho}, \rho_{\leq t}, d^*)$
}

\caption{$\mathbf{QuerySubConn}(G, S_t, \hat{\rho}, \rho_{\leq t})$}
\label{alg:subconn-query}
\end{algorithm}

We now prove Theorem \ref{thm:subgraph-connectivty-delay-alg}.

\begin{proof}
    First we show correctness.
    By Lemma \ref{lemma:subconn-error-set-size}, $d^* \leq 2d$.
    Since $P_{t, d^*} \subseteq S_t$, we apply Lemma \ref{lemma:subconn-equivalence} and observe that $u, v$ are connected in $H_t$ if and only if $u, v$ are connected in $G[S_t]$.
    When $d^* = \infty$, we compute connectivity on the graph $G[S_t]$ which is trivially correct.

    To analyze the preprocessing time, we note that Algorithm \ref{alg:promise-subconn-preprocess} required $O(T n^{3})$ time.
    Since we run this algorithm for $O(\log n)$ values of $d$, then the preprocessing algorithm required $O(T n^{3} \log n)$ time.

    To analyze Algorithm \ref{alg:subconn-query}, note that finding $d^*$ requires $O(\log n)$ time, as we can iterate over the values of $|Q(t, d)|$.
    In fact, $d^*$ can be maintained in the updating step by noting the size of $Q_{t, d}$ after the set is updated.
    Given that $d^* \leq 2d$, we construct the graph $H_t$ and run DFS in $O(d^2)$ time, as desired.
    If $d^* = \infty$, then we have $d = \Omega(n)$, and we can run DFS on the full graph $G[S_t]$ in $O(m) = O(n^2) = O(d^2)$ time.
\end{proof}

\subsection{Transitive Closure}

Following a similar approach as the subgraph connectivity problem, we obtain a (conditionally) optimal algorithm for the all pairs reachability (transitive closure) problem with constant update time and query time $O(d^2)$.
For simplicity, we again begin with an algorithm that is given the prediction delay $d$ as additional input, and use a similar transformation to design an algorithm with only prediction $\hat{\rho}$ as input.

Let $G_t$ denote the graph at time step $t$ with edge set $E_t$.
Given a request sequence $\rho$ and time steps $t_1, t_2$, an edge $(u, v)$ is {\bf permanent in $\rho$ from $t_1$ to $t_2$} if $(u, v) \in E_t$ for all $t \in [t_1, t_2]$.
A vertex $u$ is {\bf active in $\rho$ from $t_1$ to $t_2$} if $u$ is incident to any edge update in the request sequence $\rho_{[t_1, t_2]}$ or if $u$ is part of any query in the request sequence $\rho_{[t_1, t_2]}$.
We maintain the following data structures at each time step $t$.

\begin{enumerate}
    \item $G_t$, the current graph $G$. 
    This is maintained in $O(1)$ time by updating one bit in the adjacency matrix.

    \item For all $t \in T$ and $u, v \in A_t$, let $C(u, v, t)$ denote $v$ is reachable from $u$ in the graph $\hat{H}(t) = (V, P_t)$, where $A_t$ is the set of active vertices in $\hat{\rho}$ from $t - d$ to $t + d$ and $P_t$ is the set of permanent edges in $\hat{\rho}$ from $t - d$ to $t + d$.
    This is precomputed in the preprocessing phase and only queried in the dynamic phase.

    \item $F_t = E_t \setminus P_t$, the set of edges that are not permanent in $\hat{\rho}$ from $t - d$ to $t + d$.
    This can be maintained in $\tO{1}$ time as $E_t$ changes only by the edge specified in $\rho_t$, while $P_t$ can add one edge from the update $\hat{\rho}_{t - d}$ and lose one edge from the update $\hat{\rho}_{t + d}$.
    Furthermore, $P_t$ can in fact be computed in the preprocessing phase, so that maintaining $F_t$ can be accomplished by maintaining $E_t$ and checking which $O(1)$ edges are inserted and removed from $P_t$.
\end{enumerate}

In our algorithm, we will maintain the edge set $F_t$.
Let $V(F_t)$ denote all vertices incident to at least one edge in $F_t$.
On a given query $(u, v)$, we construct the graph $H_t$ with vertex set $V(F_t) \cup \set{u, v}$ with all edges in the induced subgraph augmented by an edge for every pair $(a, b)$ for which $b$ is reachable from $a$ using only edges in $P_t$.
On this graph we run DFS to check for reachability. 
To guarantee the performance of the algorithm, we require that $F_t$ is not too large.

\begin{lemma}
    \label{lemma:reachability-error-set}
    $P_t \subset E_t$ and $|F_t| \leq 2d + 1$
\end{lemma}

\begin{proof}
    By Lemma \ref{lemma:d-sub-sequence-containment}, if $\hat{\rho}$ is a $d$-delayed prediction, then for all $t$, $\hat{\rho}_{\leq t - d} \subset \rho_{\leq t} \subset \hat{\rho}_{\leq t + d}$.
    If $e \in P_t$, then $e \in E_{t'}(\hat{\rho})$ for all $t' \in [t - d, t + d]$ so that no update in $\hat{\rho}_{[t - d, t + d]}$ flips the edge $e$.
    Therefore, $e \in E_t$.

    For the second claim, we observe that $F_t \subset A_t$ and $|A_t| \leq 2d + 1$ since there are $2d + 1$ updates between $t - d$ and $t + d$ and each update can involve only one edge.
\end{proof}

To guarantee correctness, we require that the graph $H_t$ accurately encodes reachability relationships between vertices.

\begin{lemma}
    \label{lemma:reachability-d-equivalence}
    Let $u, v \in V(H_t)$.
    Then, $v$ is reachable from $u$ in $H_t$ if and only if $v$ is reachable from $u$ in $G_t$.
\end{lemma}

\begin{proof}
    Suppose $v$ is reachable from $u$ in $H_t$ with the path $(u, w_1, w_2, \dotsc, w_{l - 1}, v)$.
    Each edge is either an edge in the induced subgraph and therefore an edge in $G_t$ or an auxiliary edge $(w_{i}, w_{i + 1})$ added for $w_{i + 1}$ reachable from $w_{i}$ in using edges in $P_t \subset E_t$.
    Therefore, there is a path from $u$ to $v$ in $G_t$ as well.

    Suppose $v$ is reachable from $u$ in $G_t$ with a path $(u, w_1, w_2, \dotsc, w_{l - 1}, v)$.
    Consider the subsequence $\set{w_{i_j}}$ of vertices in $V(H_t)$, noting that $u, v$ are in this subsequence.
    It suffices to show that for all $j$, the edge $(w_{i_j}, w_{i_{j + 1}})$ is in $H_t$.
    Suppose not, then $i_{j + 1} > i_{j} + 1$, otherwise the edge is in the induced subgraph.
    Consider the subpath from $w_{i_j}$ to $w_{i_{j + 1}}$.
    Every edge is in $E_t$ but not the induced subgraph on $V(F_t) \cup \set{u, v}$, as at least one endpoint is not in $V(F_t) \cup \set{u, v}$.
    Then, since each edge is not in $F_t$, this edge must be in $P_t$, so that $C(w_{i_j}, w_{i_{j + 1}}, t)$ is true and the edge exists in $H_t$.
\end{proof}

Equipped with these tools, we state our update-optimized algorithm and the resulting theorem.

\begin{algorithm}[H]

\SetKwInOut{Input}{Input}
\Input{Initial graph $G_0$.
$d$-delayed predictions $\hat{\rho}$.}

\BlankLine

\For{$t = 1$ to $T$}{
    Compute $P_t = \set{e \given e \textrm{ is permanent in $\hat{\rho}$ from $\max(1, t - d)$ to $\min(T, t + d)$}}$

    Compute $A_t = \set{v \given v \textrm{ is active in $\hat{\rho}$ from $\max(1, t - d)$ to $\min(T, t + d)$}}$

    Let $\hat{H}_t = (V, P_t)$ be the subgraph with edge set $P_t$.

    \For{$u, v \in A_t$}{
        Run DFS on $\hat{H}_t \gets (V, P_t)$ and store in $C(u, v, t)$ a 1 if $v$ is reachable from $u$ in $\hat{H}_t$ and 0 otherwise.
    }
}

\caption{$\mathbf{PromisePreprocessingTC}(G_0, \hat{\rho}, d)$}
\label{alg:promise-transitive-closure-preprocess}
\end{algorithm}

We preprocess the input as described by Algorithm \ref{alg:promise-transitive-closure-preprocess}.
Given an update, we maintain $G_t, F_t$ in $O(1)$ time, as discussed above.
In the query step, we will construct the graph $H_t$ and compute connectivity in this graph.

\begin{algorithm}[H]

\SetKwInOut{Input}{Input}\SetKwInOut{Output}{Output}
\Input{Current graph $G_t$,
$d$-delayed predictions $\hat{\rho}$.}
\Output{YES if $v$ is reachable from $u$ in $G_t$, NO otherwise}

\BlankLine

Construct $H_t$ with vertex set $V(F_t) \cup \set{u, v}$ and edge set $G[V(F_t) \cup \set{u, v}]$ and additional edges $(a, b)$ where $C(a, b, t) = 1$.

Run DFS from $u$ in $H_t$ and return YES if $v$ is reachable from $u$ and NO otherwise.

\caption{$\mathbf{PromiseQueryTC}(G_t, \hat{\rho}, \rho_{\leq t}, d)$}
\label{alg:promise-transitive-closure-query}
\end{algorithm}

\begin{lemma}
    \label{lemma:transitive-closure-delay-promise-alg}
    Let $T$ by polynomial in $n$.
    Let $\hat{\rho} = (\hat{\rho}_1, \hat{\rho}_2, \dotsc, \hat{\rho}_T)$ be a $d$-delayed prediction for the transitive closure problem.

    There is an algorithm solving the transitive closure problem given $\hat{\rho}$ and $d$ with polynomial preprocessing time $P(n)$, update time $U(n) = \tO{1}$ and query time $Q(n) = O(d^2)$.
\end{lemma}

\begin{proof}
    Correctness of the algorithm follows from Lemma \ref{lemma:reachability-d-equivalence}.
    To examine the running time, in the preprocessing step we can compute $A_t, P_t$ in $\tO{T}$ time for all $t \in [T]$ by computing the change from $P_{t - 1}$ (resp. $A_{t - 1}$) to $P_t$ (resp. $A_t$) in $\tO{1}$ time.
    Computing reachability from $|A_t| = O(d)$ sources requires $O(n^2 d)$ time for each $t$, so that Algorithm \ref{alg:promise-transitive-closure-preprocess} requires time $O(T n^2 d)$.

    Algorithm \ref{alg:promise-transitive-closure-query} constructs a graph on $O(d)$ vertices (Lemma \ref{lemma:reachability-error-set}) and runs DFS, which requires at most $O(d^2)$ time.
\end{proof}

We use a similar transformation as in the subgraph connectivity problem to obtain an algorithm that does not require $d$ as input.
In the preprocessing algorithm, we run Algorithm \ref{alg:promise-transitive-closure-preprocess} for $O(\log n)$ values of $d = 2^{d'}$ where $1 \leq d' \leq \ceil{\log n}$, saving the data structures $\set{P_{t, d}} \set{A_{t, d}}$ and $\set{C_d(u, v, t)}$.
In the update step, we maintain in $\tO{1}$ time the dynamic edge set $E_t$ as well as the sets $F_{t, d} = E_t \setminus P_{t, d}$ for all $d$ such that $P_{t, d} \subset E_t$.
In each update, we also maintain the minimum $d^*$ such that $P_{t, d} \subset E_t$ and $|F_t| = |E_t \setminus P_{t, d}| \leq 2d + 1$.
Given a query, we return $\mathbf{PromiseQueryTC(G_t, \hat{\rho}, \rho_{\leq t}, d^*)}$, which is correct by Lemma \ref{lemma:reachability-d-equivalence} and $P_{t, d^*} \subset E_t$.

Furthermore, we run Algorithm \ref{alg:promise-transitive-closure-preprocess} $O(\log n)$ times in the preprocessing phase, so that the preprocessing time remains polynomial.
By Lemma \ref{lemma:reachability-error-set}, $d^* \leq 2d$ where $d$ is the delay of the prediction $\hat{\rho}$, so that each query requires $O(d^2)$ time.
If $d^* = \infty$, then $d = \Omega(n)$ and we can run DFS on the full graph $G_t$ in $O(m) = O(n^2) = O(d^2)$ time.
The above discussion yields the following theorem.

\begin{theorem}
    \label{thm:transitive-closure-delay-alg}
    Let $T$ by polynomial in $n$.
    Let $\hat{\rho} = (\hat{\rho}_1, \hat{\rho}_2, \dotsc, \hat{\rho}_T)$ be a $d$-delayed prediction for the transitive closure problem.

    There is an algorithm solving the transitive closure problem given $\hat{\rho}$ with polynomial preprocessing time $P(n)$, update time $U(n) = \tO{1}$ and query time $Q(n) = O(d^2)$.
\end{theorem}

\subsection{Shortest Paths}

With a slightly more careful analysis, we can obtain a similar result for all pairs shortest paths.
Keeping the same definitions as with transitive closure, we maintain the following data structures at each time step $t$.

\begin{enumerate}
    \item $G_t$, the current graph $G$ maintained in $O(1)$ time.

    \item For all $t \in T$ and $u, v \in A_t$, let $D(u, v, t)$ denote the distance from $u$ to $v$ in the graph $\hat{H}_t = (V, P_t)$.

    \item $F_t = E_t \setminus P_t$, the set of edges in $G_t$ that are not permanent in $\hat{\rho}$ from $t - d$ to $t + d$, maintained in $O(1)$ time.
\end{enumerate}

Our algorithm again maintains the set $F_t$.
By an identical argument to Lemma \ref{lemma:reachability-error-set}, we have $E_t = P_t \cup F_t$ where $|F_t| \leq 2d + 1$.
At each query $(u, v)$, construct the graph $H_t$ with vertex set $V(F_t) \cup \set{u, v}$ and edge set consisting of the induced subgraph $G_t[V(F_t) \cup \set{u, v}]$ and edge $(a, b)$ with weight $D(a, b, t)$.
If $(a, b)$ is already an edge in the induced subgraph, set the weight of $(a, b)$ to be the minimum of the edge weight and the distance $D(a, b, t)$.
We now require the following lemma for correctness.

\begin{lemma}
    \label{lemma:shortest-path-equivalence}
    Let $t$ be a time step.
    Let $(u, v)$ be the query at $\rho_t$.
    Let $H_t$ be the graph with vertex set $V(H_t) = V(F_t) \cup \set{u, v}$ and edge set $G[V(H_t)]$ and edge $(a, b)$ with weight $D(a, b, t)$ for all pairs of vertices $a, b$ where $D(a, b, t) < \infty$ is finite.

    Then, $\distT_G(u, v) = \distT_{H_t}(u, v)$.
\end{lemma}

\begin{proof}
    First, we argue that $\distT_{G_t}(u, v) \leq \distT_{H_t}(u, v)$.
    This follows as every edge in $H_t$ is either an edge in $G_t$ or an edge with weight equivalent to a path in $G_t$.
    Then, it suffices to verify that $\distT_{H_t}(u, v) \leq \distT_{G_t}(u, v)$.
    Consider a shortest path $P(u, v) = (u, w_1, \dotsc, w_{l - 1}, v)$ in $G_t$.
    We take the subsequence $P_{H} = (w_{i_j})$ of vertices $w_{i_j} \in V(H_t)$.
    It suffices to show that $P_H$ is a path in $H_t$ with weight at most $\distT_{G_t}(u, v)$, the total weight of path $P(u, v)$.
    We consider two cases.
    If $i_{j + 1} = i_j + 1$ then the edge $(w_{i_j}, w_{i_{j + 1}}) \in E(H_t)$ as an edge in the induced subgraph.
    Otherwise, the intermediate vertices are not in $V(H_t)$ and therefore not in $V(F_t)$.
    Since at least one endpoint of each edge is not in $F_t$, all the edges in $E_t$ between $w_{i_j}, w_{i_{j + 1}}$ are in $P_t$.
    In particular, $\wt(w_{i_j}, w_{i_{j + 1}}) = D(a, b, t)$ is at most the weight of the subpath in $G_t$, proving the desired statement.
\end{proof}

Following a similar approach as transitive closure, we obtain the following theorem.

\begin{lemma}
    \label{lemma:all-pairs-shortest-path-delay-promise-alg}
    Let $T$ by polynomial in $n$.
    Let $\hat{\rho} = (\hat{\rho}_1, \hat{\rho}_2, \dotsc, \hat{\rho}_T)$ be a $d$-delayed prediction for the all pairs shortest path problem on weighted, directed graphs.

    There is an algorithm solving the all pairs shortest path problem on weighted digraphs given $\hat{\rho}$ and $d$ with polynomial preprocessing time $P(n)$, update time $U(n) = \tO{1}$ and query time $Q(n) = O(d^2)$.
\end{lemma}

The algorithm and proof of correctness are given below.

\begin{algorithm}[H]

\SetKwInOut{Input}{Input}
\Input{Initial graph $G_0$.
$d$-delayed predictions $\hat{\rho}$.}

\BlankLine

\For{$t = 1$ to $T$}{
    Compute $P_t = \set{e \given e \textrm{ is permanent in $\hat{\rho}$ from $\max(1, t - d)$ to $\min(T, t + d)$}}$

    Compute $A_t = \set{v \given v \textrm{ is active in $\hat{\rho}$ from $\max(1, t - d)$ to $\min(T, t + d)$}}$

    Let $\hat{H}_t = (V, P_t)$ be the subgraph with edge set $P_t$.

    \For{$u, v \in A_t$}{
        Run Dijkstra's on $\hat{H}_t \gets (V, P_t)$ and store in $D(u, v, t) \gets \distT_{\hat{H}_t}(u, v)$
    }
}

\caption{$\mathbf{PromisePreprocessingAPSP}(G_0, \hat{\rho}, d)$}
\label{alg:promise-apsp-preprocess}
\end{algorithm}

We preprocess the input as described by Algorithm \ref{alg:promise-apsp-preprocess}.
Given an update, we maintain $G_t, F_t$ in $\tO{1}$ time.
In the query step, we will construct the graph $H_t$ and compute the shortest path distance in this graph.

\begin{algorithm}[H]

\SetKwInOut{Input}{Input}\SetKwInOut{Output}{Output}
\Input{Current graph $G_t$,
$d$-delayed predictions $\hat{\rho}$.}
\Output{Distance from $u$ to $v$ in $G_t$.}

\BlankLine

Construct $H_t$ with vertex set $V(F_t) \cup \set{u, v}$ and edge set $G[V(F_t) \cup \set{u, v}]$ and additional edges $(a, b)$ with weight $D(a, b, t)$ where $D(a, b, t) < \infty$ is finite.

Run Dijkstra's from $u$ in $H_t$ and return $\distT_{H_t}(u, v)$

\caption{$\mathbf{PromiseQueryAPSP}(G_t, \hat{\rho}, \rho_{\leq t}, d)$}
\label{alg:promise-apsp-query}
\end{algorithm}

\begin{proof}
    Correctness follows from Lemma \ref{lemma:shortest-path-equivalence}.
    The preprocessing algorithm requires $O(T n^2 d)$ as each invocation of Dijkstra's algorithm requires $O(n^2)$ time.
    The query algorithm requires $O(d^2)$ time as we invoke Dijkstra's algorithm on a graph with $d$ vertices.
\end{proof}

Again, we apply the same transformation as with transitive closure and subgraph connectivity to design an algorithm that requires input $\hat{\rho}$ only.
In the preprocessing algorithm, we run Algorithm \ref{alg:promise-apsp-preprocess} for $O(\log n)$ values of $d = 2^{d'}$ where $1 \leq d' \leq \ceil{\log n}$, saving the data structures $\set{P_{t, d}} \set{A_{t, d}}$ and $\set{D_d(u, v, t)}$.
In the update step, we maintain in $O(1)$ time the dynamic edge set $E_t$ as well as the sets $F_{t, d} = E_t \setminus P_{t, d}$ for all $d$ such that $P_{t, d} \subset E_t$.
In each update, we also maintain the minimum $d^*$ such that $P_{t, d} \subset E_t$ and $|F_t| = |E_t \setminus P_{t, d}| \leq 2d + 1$.
Given a query, we return $\mathbf{PromiseQueryAPSP(G_t, \hat{\rho}, \rho_{\leq t}, d^*)}$, which is correct by Lemma \ref{lemma:shortest-path-equivalence} and $P_{t, d^*} \subset E_t$.

Furthermore, we run Algorithm \ref{alg:promise-apsp-preprocess} $O(\log n)$ times in the preprocessing phase, so that the preprocessing time remains polynomial.
Since $|F_t| \leq 2d + 1$, $d^* \leq 2d$ where $d$ is the delay of the prediction $\hat{\rho}$, so that each query requires $O(d^2)$ time.
If $d^* = \infty$, then $d = \Omega(n)$ and we can run Dijkstra on the full graph $G_t$ in $O(m) = O(n^2) = O(d^2)$ time.
The above discussion yields the following theorem.

\begin{theorem}
    \label{thm:apsp-delay-alg}
    Let $T$ by polynomial in $n$.
    Let $\hat{\rho} = (\hat{\rho}_1, \hat{\rho}_2, \dotsc, \hat{\rho}_T)$ be a $d$-delayed prediction for the all pairs shortest path problem on weighted digraphs.

    There is an algorithm solving the all pairs shortest path problem on weighted digraphs given $\hat{\rho}$ with polynomial preprocessing time $P(n)$, update time $U(n) = \tO{1}$ and query time $Q(n) = O(d^2)$.
\end{theorem}

\subsection{Erickson's Maximum Value Problem}

Next, we showcase the optimality of our lower bounds for a non-graph problem.
Recall that in Erickson's Problem, the algorithm is given an initial matrix $M_0$.
Each update increments either a row or a column by 1.
Each query asks for the maximum value in the matrix $M_t$.

\begin{theorem}
    \label{thm:ericksons-delay-alg}
    Let $T$ by polynomial in $n$.
    Let $\hat{\rho} = (\hat{\rho}_1, \hat{\rho}_2, \dotsc, \hat{\rho}_T)$ be a $d$-delayed prediction for Erickson's Problem.
    
    There is an algorithm solving Erickson's problem given $\hat{\rho}$ with polynomial preprocessing time $P(n)$, update time $U(n) = \tO{d}$ and query time $Q(n) = \tO{1}$.

    There is an algorithm solving Erickson's problem given $\hat{\rho}$ with polynomial preprocessing time $P(n)$, update time $U(n) = \tO{1}$ and query time $Q(n) = \tO{d^2}$.
\end{theorem}

The trivial algorithm simply maintains the $n \times n$ dynamic matrix, processing an update in $O(n)$ time and storing the maximum value, which is returned in $O(1)$ time.
Alternatively, we can store an increment in $O(1)$ time (e.g. maintaining an array recording the number of increments for each row and column) and construct the matrix, computing the maximum in $O(n^2)$ time.

We now give some useful definitions that will help us design efficient algorithms with prediction.
Consider an initial matrix $M_0$ and request sequence $\rho$.
Let $M_t$ the current state of the matrix after $t$ requests.
For $i, j \in [n]$, let $r(\rho, t, i), c(\rho, t, j)$ be functions denoting the number of times row $i$ and column $j$ have been incremented at time step $t$, so that value of entry $(i, j)$ at time step $t$ is $M_t[i][j] = M_0[i][j] + r(\rho, t, i) + c(\rho, t, j)$.
There are at most $d$ rows and $d$ columns where the number of predicted increments does not match the number of actual increments, a set of errors which we can maintain efficiently.
For a given time step $t$, let $E_R(t) \subset [n]$ denote the set of rows where $r(\rho, t, i) \neq r(\hat{\rho}, t, i)$ and $E_C(t) \subset [n]$ denote the set of columns where $c(\rho, t, j) \neq c(\hat{\rho}, t, j)$.

As with subgraph connectivity, given a predicted request sequence $\rho$ and time steps $t_1, t_2$, a row $i$ is {\bf permanent in $\rho$ from $t_1$ to $t_2$} if $i$ is not incremented in $\rho_{[t_1, t_2]}$. 
A row $i$ is {\bf active in $\rho$ from $t_1$ to $t_2$} if $i$ is incremented in $\rho_{[t_1, t_2]}$.
Permanent and active columns are defined similarly.

\begin{lemma}
    \label{lemma:ericksons-error-sets}
    Let $\hat{\rho}$ be a prediction that is $d$ delayed with $k$ outliers.
    For all $t \in [T]$, $|E_R(t)|, |E_C(t)| = O(d + k)$.
    Given $E_R(t - 1), E_C(t - 1)$, the sets $E_R(t), E_C(t)$ can be maintained in $\tO{1}$ time per request.
    It is possible to also maintain $r(\rho, t, i) - r(\hat{\rho}, t, i)$ for all $i \in E_R(t)$ and $c(\rho, t, j) - c(\hat{\rho}, t, j)$ for all $j \in E_C(t)$.
\end{lemma}

\begin{proof}
    The bound on $|E_R(t)|, |E_C(t)|$ follow Lemma \ref{lemma:d-k-sym-diff-size} as $E_R(t) \cup E_C(t) \subset \rho_{\leq t} \Delta \hat{\rho}_{\leq t}$, the symmetric difference between the request sequences, which was shown to be of size at most $O(d + k)$.

    It remains to show that the sets can be maintained efficiently dynamically.
    Suppose we have $r(\rho, t - 1, i) - r(\hat{\rho}, t - 1, i)$ for all $i$ and $c(\rho, t - 1, j) - c(\hat{\rho}, t - 1, j)$ for all $j$.
    Given the predicted request $\hat{\rho}_t$ and true request $\rho_t$, we update at most 2 values in $\tO{1}$ time, leaving the remaining difference values as they remain unchanged.
    If an updated difference becomes 0, we remove this index from $E_R$ or $E_C$ as appropriate.
    In an updated difference becomes non-zero, we add this index to $E_R$ or $E_C$ as appropriate.
    This can be done in $\tO{1}$ time.
\end{proof}

Now, if we let $\hat{M}_t$ denote the matrix under the predicted request sequence $\hat{\rho}$ after the $t$-th time step, $M_t[i][j] = \hat{M}_t[i][j]$ if $i \notin E_R(t)$ and $j \notin E_C(t)$, so that the maximum entry among these $i, j$ remains the same.
For a fixed $i \in E_R(t)$ and among all $j \notin E_C(t)$, then the \emph{relative} difference between entries in the $i$-th row are the same in $M_t$ and $\hat{M}_t$, even if the absolute values are different.
Therefore, if we maintain a data structure (for example a heap) that keeps the predicted maximum values for the $i$-th row, then we need to only correct up to $O(d)$ entries to find the new maximum of the $i$-th row.
Only if both $i, j \in E_R(t), E_C(t)$ do we need to completely recalculate the maximum, but there are at most $O(d^2)$ such entries.

\paragraph{Query-Optimized $\tO{d + k}$ Update Algorithm}

We will in fact show a query-optimized algorithm that requires $d$ delay with $k$ outlier predictions.
By choosing $k = 0$, this gives the query-optimized algorithm of Theorem \ref{thm:ericksons-delay-alg}.
In our algorithm, we maintain the following data structures:

\begin{enumerate}
    \item $I_R, I_C: [n] \rightarrow \N$, two size $n$ list storing how many times each row and column is incremented. 
    This can be maintained in $O(1)$ time per update.

    \item $E_R, E_C$, the error sets described above.
    For each $i, j \in [n]$, we also maintain $r(\rho, t, i) - r(\hat{\rho}, t, i)$ and $c(\rho, t, j) - c(\hat{\rho}, t, j)$.
    This can be maintained in $\tO{1}$ time by Lemma \ref{lemma:ericksons-error-sets}.

    \item For each row $i$ and time step $t$, a binary heap $H_R(i, t)$ with column indices as keys and $\hat{M}_t[i][j]$ as the values.
    For each column $j$ and time step $t$, a binary heap $H_C(j, t)$ with row indices as keys $\hat{M}_t[i][j]$ as the values.
    This is pre-computed in the preprocessing phase.
    During the $t$-th request, there may be some modifications to $H_R(i, t)$ and $H_C(j, t)$.
    After the $t$-th request, the heaps can be discarded.

    \item The current maximum value $c$ and indices $i^*, j^*$ such that $M_t[i^*][j^*] = c$.
\end{enumerate}

We now describe our algorithm.

\begin{algorithm}[H]

\SetKwInOut{Input}{Input}
\Input{Initial Matrix $M_0$.
$d$-delayed with $k$ outliers predictions $\hat{\rho}$.}

\BlankLine

$I_R[i], I_C[j] \gets 0$ for all $i, j \in [n]$

$E_R, E_C \gets \emptyset$

$c \gets \max_{i, j} M_0[i][j]$ and $i^*, j^*$ such that $M_0[i^*][j^*] = c$.

\For{$t = 1$ to $T$}{
    Compute $\hat{M}_t$ from $\hat{M}_{t - 1}$ by incrementing the row or column specified by $\hat{y}_t$

    Construct binary heaps $H_R(i, t)$ for all $i \in [n]$ and $H_C(j, t)$ for all $j \in [n]$.
}

\caption{$\mathbf{UPreprocessingErickson}(M_0, \hat{\rho})$}
\label{alg:u-erickson-preprocess}
\end{algorithm}

For the update algorithm, we assume without loss of generality that the request $\rho_t$ increments row $i$.
An analogous algorithm exists if $\rho_t$ increments column $j$.
Let $i_t$ denote the row index incremented at time $t$.

\begin{algorithm}[H]

\SetKwInOut{Input}{Input}\SetKwInOut{Output}{Output}
\Input{Current matrix $M_t$, update $i_t$, $d$ delayed predictions $\hat{\rho}$, and request history $\rho_{\leq t}$}
                                                           
\BlankLine

Update $I_R[i_t] \gets I_R[i_t] + 1$

Update $E_R$ according to Lemma \ref{lemma:ericksons-error-sets}

\If{$i_t = i^*$}{
    $c \gets c + 1$, $i^* \gets i^*$ and $j^* \gets j^*$
}\Else{
    \lFor{$j \in E_C$}{Update key $j$ of $H_R(i_t, t)$ with value $\hat{M}_t [i_t][j] + (c(\rho, t, j) - c(\hat{\rho}, t, j))$}
    \label{line:erickson-u-u:update-heap}
    
    $j_0, c_0 \gets \max(H_R(i_t, t))$ maximum value with index $j_0$ and value $c_0$

    $c_1 \gets c_0 + (r(\rho, t, i_t) - r(\hat{\rho}, t, i_t))$
    \label{line:erickson-u-u:row-max}

    \If{$c_1 > c$}{
        $c \gets c_1$, $i^* \gets i_t$ and $j^* \gets j_0$
    }\Else{
        $c \gets c$, $i^* \gets i^*$ and $j^* \gets j^*$
    }
}

\caption{$\mathbf{UUpdateErickson}(M_t, i_t, \hat{\rho}, \rho_{< t})$}
\label{alg:u-erickson-update}
\end{algorithm}

On a given query, we simply return $c$ in $O(1)$ time.
In the following Lemma, we claim that Algorithm \ref{alg:u-erickson-update} maintains the correctness of the data structures defined above.

\begin{lemma}
    \label{lemma:update-erickson-correct}
    After Algorithm \ref{alg:u-erickson-update}, the data structures $I_R, I_C, E_R, E_C, c, i^*, j^*$ contain the desired values.
\end{lemma}

\begin{proof}
    We assume that the data structures are maintained correctly up to the $t - 1$ time step.
    By definition, $I_C$ does not change since we have a row update and only $I_R[i_t]$ increments as this is the row updated.
    $E_R, E_C$ are maintained according to Lemma \ref{lemma:ericksons-error-sets}.
    If $i_t = i^*$, then since $M_{t - 1}[i^*][j^*]$ was the maximum value of $M_{t - 1}$, $|M_t - M_{t - 1}|_{\infty} = 1$ and $M_t[i^*][j^*] = M_{t - 1}[i^*][j^*] + 1$, then $c + 1$ is the maximum value and this is achieved by $i^*, j^*$.
    Therefore, we can assume $i_t \neq i^*$ in the following.
    
    We begin by noting the following equalities which hold for all $i, j \in [n]$.

    \begin{align*}
        M_t[i][j] &= M_0[i][j] + r(\rho, t, i) + c(\rho, t, j) \\
        \hat{M}_t[i][j] &= M_0[i][j] + r(\hat{\rho}, t, i) + c(\hat{\rho}, t, j) \\
        M_t[i][j] &= \hat{M}_t[i][j] + (r(\rho, t, i) - r(\hat{\rho}, t, i)) + (c(\rho, t, i) - c(\hat{\rho}, t, j))
    \end{align*}

    After the preprocessing step, we have that $H_R(i_t, t)$ contains the values $\hat{M}[i_t][j]$ for all $j \in [n]$.
    After Line \ref{line:erickson-u-u:update-heap} then, we have $H_R(i_t, t)$ contains for all $j$,
    \begin{equation*}
        \hat{M}_t[i_t][j] + (c(\rho, t, j) - c(\hat{\rho}, t, j)) = M_t[i][j] - (r(\rho, t, i_t) - r(\hat{\rho}, t, i_t))
    \end{equation*}

    In particular, since the error term does not depend on $j$,
    \begin{equation*}
        j_0 = \arg \max_{j} M_t[i_t][j] - (r(\rho, t, i_t) - r(\hat{\rho}, t, i_t)) = \arg \max_{j} M_t[i_t][j]
    \end{equation*}

    so that after Line \ref{line:erickson-u-u:row-max}, $c_1 \gets \max_j M_t[i_t][j]$ and this maximum is attained at $j_0$.
    Since $M_{t - 1}[i][j] = M_t[i][j]$ for all $i \neq i_t$, we have that the maximum value is either $c$ or $M_t[i_t][j_0]$.
    Comparing the two and updating $c, i^*, j^*$ accordingly completes the proof of the Lemma.
\end{proof}

Given the above lemma, correctness easily follows as we handle queries by returning $c$.

\begin{lemma}
    \label{lemma:erickson-update-time}
    Algorithm \ref{alg:u-erickson-preprocess} requires $O(T n^2)$ time.
    Algorithm \ref{alg:u-erickson-update} requires $O((d + k) \log n)$ time.
\end{lemma}

\begin{proof}
    To examine Algorithm \ref{alg:u-erickson-preprocess}, note that $I_R, I_C, E_R, E_C$ can be computed in $O(n)$ time while $c, i^*, j^*$ can be computed in $O(n^2)$ time.
    Then, for each time step $t$, we compute $\hat{M}_t$ in $O(n)$ time following the trivial algorithm, while each heap construction requires $O(n)$ time.
    Since we construct $n$ heaps, this requires $O(n^2)$ time overall.

    Examining Algorithm \ref{alg:u-erickson-update}, updating $I_R$ and $E_R$ require $\tO{1}$ time.
    If $i_t = i^*$, the entire algorithm requires $\tO{1}$ time.
    Otherwise, we update $O(d)$ key values in the heap $H_R(i_t, t)$, requiring $O((d + k) \log n)$ time.
    Extracting $j_0, c_0$ and the remaining steps can be completed in $O(\log n)$ time.
\end{proof}

\paragraph{Update-Optimized $\tO{d^2}$ Query Algorithm}
To design an update-optimized algorithm, we observe that given $d$ delayed predictions $\hat{\rho}$, for all time steps $t$,
\begin{equation*}
        \hat{\rho}_{\leq t - d} \subseteq \rho_{\leq t} \subseteq \hat{\rho}_{\leq t + d}
\end{equation*}
by Lemma \ref{lemma:d-sub-sequence-containment}.

Again, for all $1 \leq d' \leq \ceil{\log n}$, let $d = 2^{d'}$.
We will maintain the following data structures.

\begin{enumerate}
    \item $I_R, I_C, E_R, E_C, \set{H_R(i, t)}_{i, t}, \set{H_C(j, t)}_{j, t}$ as in the query-optimized algorithm.
    \item $D_{R, d}^{(-)}$, an $n$-dimensional array containing $r(\rho, t, i) - r(\hat{\rho}, t - d, i)$ for all $i \in [n]$, $d = 2^{d'}$
    \item $D_{R, d}^{(+)}$, an $n$-dimensional array containing $r(\hat{\rho}, t + d, i) - r(\rho, t, i)$ for all $i \in [n]$, $d = 2^{d'}$
    \item $D_{C, d}^{(-)}$, an $n$-dimensional array containing $c(\rho, t, j) - c(\hat{\rho}, t - d, j)$ for all $j \in [n]$, $d = 2^{d'}$
    \item $D_{C, d}^{(+)}$, an $n$-dimensional array containing $c(\hat{\rho}, t + d, j) - c(\rho, t, j)$ for all $j \in [n]$, $d = 2^{d'}$
    \item $B_{R, d} = \set{i \in [n] \given D_{R, d}^{(-)}[i], D_{R, d}^{(+)}[i] \geq 0}$ and $B_{C, d} = \set{j \in [n] \given D_{C, d}^{(-)}[j], D_{C, d}^{(+)}[j] \geq 0}$ for all $d = 2^{d'}$
    \item $B_{R, d}^{(>)} = \set{i \in [n] \given D_{R, d}^{(-)}[i] > 0}$ and $B_{C, d}^{(>)} = \set{j \in [n] \given D_{C, d}^{(-)}[j] > 0}$ for all $d = 2^{d'}$
\end{enumerate}

In the preprocessing step, initialize $I_R, I_C, E_R, E_C$ and compute the binary heaps in the update-optimized algorithm, as well as a partial maximum $\hat{p}_t$, consisting of the maximum entry where both the row and column indices are permanent in $\hat{\rho}$ from $t - d$ to $t + d$.
We also initialize the data structures $D_{R, d}^{(-)}, D_{R, d}^{(+)}, D_{C, d}^{(-)}, D_{C, d}^{(+)}, B_{R, d}, B_{C, d}, B_{R, d}^{(>)}, B_{C, d}^{(>)}$ to their initial values.
We do this for all values of $d$.

\begin{algorithm}[H]

\SetKwInOut{Input}{Input}
\Input{Initial Matrix $M_0$.
$d$-delayed predictions $\hat{\rho}$.}

\BlankLine

$I_R[i], I_C[j] \gets 0$ for all $i, j \in [n]$

$E_R, E_C \gets \emptyset$

\For{$d = 2^{d'}$ with $d' = 1$ to $\ceil{\log n}$}{
    $D_{R, d}^{(-)}, D_{C, d}^{(-)} \gets \vec{0}$

    $D_{R, d}^{(+)}[i] \gets r(\hat{\rho}, d, i)$ for all $i \in [n]$

    $D_{R, d}^{(+)}[j] \gets c(\hat{\rho}, d, j)$ for all $j \in [n]$

    $B_{R, d}, B_{C, d} \gets [n]$

    $B_{R, d}^{(>)}, B_{C, d}^{(>)} \gets \emptyset$
}

\For{$t = 1$ to $T$}{
    Compute $\hat{M}_t$ from $\hat{M}_{t - 1}$ by incrementing the row or column specified by $\hat{y}_t$

    Construct binary heaps $H_R(i, t)$ for all $i \in [n]$ and $H_C(j, t)$ for all $j \in [n]$

    \For{$d = 2^{d'}$ with $d' = 1$ to $\ceil{\log n}$}{
        $P_R(t, d) \gets \set{i \given \textrm{row $i$ permanent in $\hat{\rho}$ from $t - d$ to $t + d$}}$
    
        $P_C(t, d) \gets \set{j \given \textrm{column $j$ permanent in $\hat{\rho}$ from $t - d$ to $t + d$}}$

        Compute $\hat{p}_{t, d} \gets \max \set{\hat{M}_t[i][j] \given i \in P_R(t, d), j \in P_C(t, d)}$
    }
}

\caption{$\mathbf{QPreprocessingErickson}(M_0, \hat{\rho})$}
\label{alg:q-erickson-preprocess}
\end{algorithm}

Consider now an update.
Our first step is to maintain $d^*$, the smallest $d$ satisfying $\hat{\rho}_{\leq t - d} \subset \rho_{\leq t} \subset \hat{\rho}_{\leq t + d}$.

First, we describe how to maintain the data structures initialized in Algorithm \ref{alg:q-erickson-preprocess}.
Note $I_R, I_C, E_R, E_C$ can be maintained in $\tO{1}$ time by Lemma \ref{lemma:erickson-update-time}.
Fix a single $d$.
The arrays $D_{R, d}^{(-)}, D_{R, d}^{(+)}, D_{C, d}^{(-)}, D_{C, d}^{(+)}$ can be maintained in $O(1)$ time by maintaining the appropriate array entries.
Membership in $B_{R, d}, B_{C, d}$ can only change when an array entry crosses between $0, -1$.
When this occurs (say $D_{R, d}^{(-)}[i] = 0$) we check the corresponding entry in its associated array (in this case $D_{R, d}^{(+)}[i]$) and update membership in $B_{R, d}$ as appropriate.
Using, for example a hash table, $B_{R, d}, B_{C, d}$ can be maintained in $O(1)$ time.
A similar argument maintains $B_{R, d}^{(>)}, B_{C, d}^{(>)}$ in $O(1)$ time.
Thus, we maintain the above data structures for all $d$ in $\tO{1}$ time.

Furthermore, let $d^*$ be the smallest value $d$ such that $|B_{R, d}| = |B_{C, d}| = n$.
Whenever $\hat{\rho}_{\leq t - d} \subseteq \rho_{\leq t} \subseteq \hat{\rho}_{\leq t + d}$ holds, we have $|B_{R, d}| = |B_{C, d}| = n$, so that we may bound $d^* \leq 2d$, if the prediction $\hat{\rho}$ is $d$ delayed.
Furthermore, $d^*$ can be maintained with no additional cost per update by checking the size of $B_{R, d}, B_{C, d}$ after updating the hash tables.

Given a query, we use the following algorithm.

\begin{algorithm}[H]

\SetKwInOut{Input}{Input}\SetKwInOut{Output}{Output}
\Input{Matrix $M_t$, $d$ delayed predictions $\hat{\rho}$, request history $\rho_{\leq t}$, empirical delay $d^*$}
\Output{$\max_{i, j} M_t[i][j]$}
                                                           
\BlankLine

\If{$d^* = \infty$}{
    Construct matrix $M_t[i][j] \gets M_0[i][j] + I_R(i) + I_C(j)$

    \Return $\max \set{M_t[i][j]}$
}

$c \gets \hat{p}_{t, d^*}$
\label{line:q-erickson:perm-max}

\For{$i \not\in P_R(t, d^*)$}{
    \For{$j \not\in P_C(t, d^*)$}{
        Update key $j$ of $H_R(i, t)$ with value $\hat{M}_t[i][j] + (c(y, t, j) - c(\hat{y}, t, j))$

        Update key $i$ of $H_C(j, t)$ with value $\hat{M}_t[i][j] + (r(y, t, i) - r(\hat{y}, t, i))$
        \label{line:q-erickson:update-heap}
    }
}

\lFor{$i \not\in P_R(t, d^*)$}{
    $c \gets \max(c, (r(y, t, i) - r(\hat{y}, t, i)) + \max(H_R(i, t)))$
}
\label{line:q-erickson:row-err-max}

\lFor{$j \not\in P_C(t, d^*)$}{
    $c \gets \max(c, (c(y, t, j) - c(\hat{y}, t, j)) + \max(H_C(j, t)))$
}
\label{line:q-erickson:col-err-max}

\Return $c$

\caption{$\mathbf{QQueryErickson}(M_t, \hat{\rho}, \rho_{\leq t}, d^*)$}
\label{alg:q-erickson-query}
\end{algorithm}

We now prove that Algorithm \ref{alg:q-erickson-query} returns the correct value $c$.

\begin{proof}
First, we argue that for all $(i, j) \in P_R(t, d^*) \times P_C(t, d^*)$, the predicted matrix $\hat{M}_t[i][j] = M_t[i][j]$.
This follows as no update in $\hat{\rho}_{[t - d^*, t + d^*]}$ increments the $i$-th row or the $j$-th column.
Therefore, 
\begin{equation*}
    \hat{p}_{t, d^*} = \max \set{M_t[i][j] \given i \in P_R(t, d^*), j \in P_C(t, d^*)}
\end{equation*}
and $c$ is set to this value in Line \ref{line:q-erickson:perm-max}.

Next, we argue that after Line \ref{line:q-erickson:update-heap}, each heap $H_R(i, t)$ and $H_C(j, t)$ contain the correct maximum index.
In particular, we claim that the maximum index $j$ in $H_R(i, t)$ is exactly $\arg \max_{j} M_t[i][j]$.
After preprocessing, each entry $H_R(i, t)$ is 
\begin{equation*}
    \hat{M}_t[i][j] = M_t[i][j] - (r(\rho, t, i) - r(\hat{\rho}, t, i)) - (c(\rho, t, j) - c(\hat{\rho}, t, j))
\end{equation*}
For all $j \in P_C(t, d^*)$, note $c(\rho, t, j) = c(\hat{\rho}, t, j)$.
Then, after Line \ref{line:q-erickson:update-heap}, we have,
\begin{equation*}
    \hat{M}_t[i][j] = M_t[i][j] - (r(\rho, t, i) - r(\hat{\rho}, t, i))
\end{equation*}
and the claim follows as the error does not depend on $j$.
In particular, in Line \ref{line:q-erickson:row-err-max}, we update $c$ with maximum of $\max_{j} M_t[i][j]$ for all $i \not\in P_R(t, d^*)$.
Following a similar argument, Line \ref{line:q-erickson:col-err-max} updates $c$ with the maximum of $\max_{i} M_t[i][j]$ for all $j \not\in P_C(t, d^*)$.
Note that this covers all $i, j \in [n] \times [n]$, proving the correctness of the algorithm.
\end{proof}

We now analyze the time complexity of the above algorithms.

\begin{proof}
    The preprocessing algorithm requires preprocessing time $O(T n^2 \log n)$ as constructing the matrices $\hat{M}_t$ and heaps requires $O(n^2)$ time.
    For each $n$, maintaining $P_R(t, d), P_C(t, d)$ requires $O(1)$ time while computing $\hat{p}_{t, d}$ requires $O(n^2)$ time.
    We conclude by noting there are $\log n$ values of $d$.

    Above, we have argued that the update algorithm requires $\tO{1}$ time.
    
    Consider now the query algorithm.
    We claim $|[n] \setminus P_R(t, d^*)|, |[n] \setminus P_C(t, d^*)| = O(d)$.
    This follows as $d^* \leq 2d$ and there are at most $O(d^*) = O(d)$ active rows and columns from $t - 2d$ to $t + 2d$.
    Therefore, computing the double for loop in Line \ref{line:q-erickson:update-heap} requires $O(d^2 \log n) = \tO{d^2}$ time.
    In Lines \ref{line:q-erickson:row-err-max} and \ref{line:q-erickson:col-err-max} we extract the maximum from $O(d)$ heaps in $O(d \log n) = \tO{d}$ time.
\end{proof}

\clearpage
\bibliographystyle{beta}
\bibliography{references}

\clearpage
\appendix

\section{Lower Bounds for OMv with Predictions}
\label{sec:omv-pred-lb-proofs}

We give the proofs for lower bounds against algorithms with prediction solving the OMv problem.

\omvhamminglb*

\begin{proof}
    Consider an $S$-OMv instance $M$ of size $n \times n$ with $n$ vectors $\vec{v}_1, \dots , \vec{v}_{n}$ such that $\supp(\vec{v}_i) \subset S$ for some index-set $S \subset [n]$ of size $|S| \leq \Delta$. 
    Observe that the $n$-dimensional vector $\vec{0}$ has $L_1$-distance, and, thus, EH-distance, at most $\Delta$ to $\vec v_i$ for every $i$.

    Now consider an algorithm with prediction $\innerAlg$ that takes amortized time $Q(n, \Delta)$ per round, after polynomial preprocessing time.
    Computing $M \vec{v}_i$ for all $i$ with $1 \leq i \leq n$ solves the original $S$-OMv problem in time $n \cdot Q(n, \Delta)$. 
    Note that the prediction $(\vec{0}, \dotsc, \vec{0})$ has EH distance at most $\Delta$ and can be passed to $\innerAlg$ for the preprocessing, after which the dynamic algorithm runs in time $O(n Q(n, \Delta))$.
    
    According to the Theorem~\ref{thm:sparse-oumv-hardness}, there is no dynamic algorithm running in $\tto{n^2\Delta}$ time.
    This implies that here is no algorithm with amortized time $Q(n, \Delta) = \tto{n\Delta}$ per round, if the OMv conjecture (Conjecture~\ref{thm:omv-conjecture}) is true.
\end{proof}

Alternatively, we can consider a prediction model which reveals some information about each round.
Let $\calset$ be a set of vector sequences.
We will say a prediction sequence $(\hat{v}_1, \dotsc, \hat{v}_n)$ is {\bf $b$-bit accurate for $\calset$} if each $\hat{v}_k \in \set{0, 1, *}^{n}$ has at least $b$ indices with $\hat{v}_k[i] \in \set{0, 1}$ and for all $(\vec{v}_1, \dotsc, \vec{v}_n) \in \calset$, we have $\vec{v}_k[i] = \hat{v}_k[i]$ for all $i, k \in [n]$ where $\hat{v}_k[i] \neq *$.
Intuitively, $\hat{v}_k$ specifies at least $b$ bits of vector $\vec{v}_k$.
Trivially, there is an algorithm with $O(n^{\omega})$ preprocessing time computing each product $M \vec{v}_k$ in $Q(n, b) = O(n(n - b))$ time.
To see this, define the vectors $(\vec{v}_1', \dotsc, \vec{v}_n')$ as follows.

\begin{equation*}
    \vec{v}_k'[i] = \begin{cases}
        1 & \hat{v}_k[i] = 1 \\
        0 & \otherwise
    \end{cases}
\end{equation*}
and in the preprocessing step compute $M \vec{v}_k'$ for all $k \in [n]$ using fast matrix multiplication, which requires $O(n^{\omega})$ time.
Given a vector update, define $\bar{v}_k = \vec{v}_k - \vec{v}_k'$ so that $|\supp(\bar{v}_k)| \leq n - b$.
Then, compute $M \bar{v}_k$ in $O(n(n - b))$ time and compute $M \vec{v}_k = M \vec{v}'_k + M \bar{v}_k$.
We show that this is almost optimal.

\begin{restatable}{theorem}{omvbitaccuratelb}
    \label{thm:OMv-bit-accurate-lb}
    Let $b \in [n]$.
    There is no algorithm with $b$-bit accurate predictions for the OMv problem with amortized time $Q(n, b)= \tto{n (n - b)}$ per round, if the OMv-conjecture is true.
\end{restatable}

\begin{proof}
    Consider an $S$-OMv instance $M$ of size $n \times n$ with $n$ vectors $\vec{v}_1, \dots , \vec{v}_{n}$ such that $\supp(\vec{v}_i) \subset S$ for some index-set $S \subset [n]$ of size $|S| \leq n - b$. 
    Observe that the prediction $(\hat{v}, \dotsc, \hat{v})$ is $b$-bit accurate by defining $\hat{v}$ as follows.
    
    \begin{equation*}
        \hat{v}[i] = \begin{cases}
            0 & i \notin S \\
            * & \otherwise
        \end{cases}
    \end{equation*}

    Now consider an algorithm with prediction $\innerAlg$ that takes amortized time $Q(n, b)$ per round, after polynomial preprocessing time.
    Computing $M \vec{v}_k$ for all $k$ with $1 \leq k \leq n$ solves the original $S$-OMv problem in time $n \cdot Q(n, b)$. 
    Note that the prediction $(\hat{v}, \dotsc, \hat{v})$ is $b$-bit accurate and can be passed to $\innerAlg$ for the preprocessing, after which the dynamic algorithm runs in time $O(n Q(n, \Delta))$.
    
    According to the Theorem~\ref{thm:sparse-oumv-hardness}, there is no dynamic algorithm running in $\tto{n^2(n - b)}$ time.
    This implies that here is no algorithm with amortized time $Q(n, b) = \tto{n (n - b)}$ per round, if the OMv conjecture (Conjecture~\ref{thm:omv-conjecture}) is true.
\end{proof}

\section{Concurrent Work}
\label{sec:concurrent-work}

In an independent work, van den Brand, Forster, Nazari, and Polak \cite{arxiv23} design a variety of dynamic graph algorithms with predictions, considering both edge and vertex updates.
\cite{arxiv23} considers three prediction models: 1) the bounded delay model, 2) a $\ell_1$ variant of the bounded delay model, and 3) a fully dynamic prediction model where the predictions are not given during the preprocessing time, but instead a vertex's deletion time is predicted upon the insertion of that vertex.
\cite{DBLP:journals/corr/abs-2307-08890} applies the third model to edge updates, giving a generic framework for designing robust \emph{fully} dynamic algorithms that achieve the worst-case update time of \emph{partially} dynamic algorithms when the predicted deletion times are accurate.

\cite{arxiv23} design several algorithms with prediction relying on fast matrix multiplication. 
We describe a few that are related to our work below.
\begin{enumerate}
    \item A partially dynamic algorithm for dynamic transitive closure under edge updates with $d$ bounded delay predictions.
    We obtain the same update ($\tO{1}$) and query time $\tO{d^2)}$ (Theorem \ref{thm:transitive-closure-delay-alg}) with a combinatorial fully dynamic algorithm, with worse preprocessing time.

    \item A partially dynamic algorithm for $(1 + \eps)$ approximate all pairs shortest path under edge updates with bounded delay predictions.
    We obtain the same update ($\tO{1}$) and query time $\tO{d^2)}$ (Theorem \ref{thm:apsp-delay-alg}) for exact APSP with a combinatorial fully dynamic algorithm, with worse preprocessing time.

    \item A fully dynamic algorithm for triangle detection and single source reachability (among other problems) under vertex updates with bounded delay predictions.
    They obtain constant query time and update time $\tO{n^{\omega - 1} + n \min(d_i, n)}$ where $d_i$ denotes the delay of the $i$-th vertex update.
    For $\striangleF$, we obtain the following two algorithms (both with worse preprocessing time than the algorithm of \cite{arxiv23}).
    \begin{enumerate}
        \item An update optimized algorithm with constant update time and $O((d + k)^2)$ query time.
        \item A query optimized algorithm with constant query time and update time $O(d + k)$.
    \end{enumerate}
    Note that our algorithms additionally support predictions that are $d$ delayed with $k$ outliers.
    However, the performance of our algorithms depend on the largest delay of any update, while the performance of \cite{arxiv23} depends only on the delay of the current update.
\end{enumerate}

\cite{arxiv23} also give a lower bound for partially dynamic algorithms with prediction for the transitive closure and approximate all pairs shortest path problems.
We give a similar result for locally reducible problems in the partially dynamic setting (Definition \ref{def:partially-dynamic-locally-reducible} and Theorem \ref{thm:locally-reducible-lower-partially-dynamic}).

\end{document}